\documentclass[a4paper, 11pt]{article}
\usepackage[usenames]{xcolor}
\usepackage{xifthen}
\usepackage{algorithmic}
\usepackage{comment}
\usepackage{tcolorbox}
\usepackage{relsize}
\usepackage[font=footnotesize]{caption}
\usepackage[normalem]{ulem}
\usepackage{bm}
\usepackage{verbatim}
\usepackage{mdframed}
\usepackage{natbib}
\usepackage[textwidth=2cm,textsize=footnotesize]{todonotes}

\def\fontsettingup{2} 
\usepackage{amsmath, amsthm, amssymb}
\usepackage{cases}
\usepackage[T1]{fontenc}
\ifthenelse{\fontsettingup = 1}{ \usepackage{eulerpx, eucal, tgpagella}   }{}
\ifthenelse{\fontsettingup = 2}{ \usepackage{mathpazo, eufrak, tgpagella} }{}
\usepackage{hyperref}
\hypersetup{
  colorlinks=true,
  linkcolor=blue,
  filecolor=magenta,
  urlcolor=cyan,
  citecolor=blue
}
\usepackage[linesnumbered,boxed,ruled,vlined]{algorithm2e}
\usepackage{tikz}
\usepackage{cleveref}
\usepackage{makecell}

\usepackage[a4paper]{geometry}
\newgeometry{
textheight=9in,
textwidth=6.5in,
top=1in,
headheight=14pt,
headsep=25pt,
footskip=30pt
}
\newtheorem{theorem}{Theorem}

\newtheorem{claim}[theorem]{Claim}
\newtheorem*{claim*}{Claim}

\newtheorem{fact}[theorem]{Fact}
\newtheorem{lemma}[theorem]{Lemma}

\newtheorem{corollary}[theorem]{Corollary}
\theoremstyle{definition}

\newtheorem{definition}[theorem]{Definition}

\newtheorem*{remark*}{Remark}

\usepackage{float}
\floatstyle{plain}
\newfloat{assumption}{H}{}
\floatname{assumption}{Assumption}
\crefname{assumption}{Assumption}{Assumptions}
\newcommand{\norm}[1]{\left\Vert#1\right\Vert}
\newcommand{\abs}[1]{\left\vert#1\right\vert}
\DeclareMathOperator*{\argmax}{arg\,max}
\DeclareMathOperator*{\argmin}{arg\,min}

\def\oPr{\mathbf{Pr}}
\renewcommand{\Pr}[2][]{ \ifthenelse{\isempty{#1}}
  {\oPr\left[#2\right]}
  {\oPr_{#1}\left[#2\right]} } % Use \Pr[a]{b} for \mathbf{Pr}_a[b], \Pr{b} for  \mathbf{Pr}[b]
\def\oE{\mathbb{E}}
\newcommand{\E}[2][]{ \ifthenelse{\isempty{#1}}
  {\oE\left[#2\right]}
  {\oE_{#1}\left[#2\right]} }
\DeclareMathOperator*{\oVar}{\mathbf{Var}}
\newcommand{\Var}[2][]{ \ifthenelse{\isempty{#1}}
  {\oVar\left[#2\right]}
  {\oVar_{#1}\left[#2\right]} }
\def\oEnt{\mathbf{Ent}}
\newcommand{\Ent}[2][]{ \ifthenelse{\isempty{#1}}
  {\oEnt\left[#2\right]}
  {\oEnt_{#1}\left[#2\right]} }
  
\newcommand{\lap}{\mathrm{Lap}}
\newcommand{\disc}{\ensuremath{\mathrm{disc}}}
\newcommand{\poly}{\ensuremath{\mathrm{poly}}}

\newcommand{\Erdos}{Erd\H{o}s\xspace}
\newcommand{\Renyi}{R\'enyi\xspace}

\newcommand{\LSKt}{\ell_{3}(G)}
\newcommand{\LSKh}{\ell_{h}(G)}
\newcommand{\LSKtwo}{\ell_{2}(G)}
\newcommand{\hatLSKt}{\ell_{3}(\hat{G})}
\newcommand{\barLSKt}{\ell_{3}(\overline{G})}
\newcommand{\tildeLSKt}{\tilde{\ell}_{3}(\hat{G})}

\allowdisplaybreaks

\title{Differentially Private Synthetic Graphs Preserving Triangle-Motif Cuts\footnote{This work is supported in part by NSFC Grant 62272431 and Innovation Program for Quantum Science and Technology (Grant No. 2021ZD0302901).}}
\date{}
\author{Pan Peng\thanks{School of Computer Science and Technology, University of Science and Technology of China, Hefei, China. \textnormal{E-mail: \url{ppeng@ustc.edu.cn}, \url{hangyuxu@mail.ustc.edu.cn}}.}
\and 
Hangyu Xu\footnotemark[1]
}

\begin{document}

\maketitle

\begin{abstract}
  We study the problem of releasing a differentially private (DP) synthetic graph $G'$ that well approximates the triangle-motif sizes of all cuts of any given graph $G$, where a motif in general refers to a frequently occurring subgraph within complex networks. Non-private versions of such graphs have found applications in diverse fields such as graph clustering, graph sparsification, and social network analysis. Specifically, we present the first $(\varepsilon,\delta)$-DP mechanism that, given an input graph $G$ with $n$ vertices, $m$ edges and local sensitivity of triangles $\LSKt$, generates a synthetic graph $G'$ in polynomial time, approximating the triangle-motif sizes of all cuts $(S,V\setminus S)$ of the input graph $G$ up to an additive error of $\tilde{O}(\sqrt{m\LSKt}n/\varepsilon^{3/2})$. Additionally, we provide a lower bound of $\Omega(\sqrt{mn}\LSKt/\varepsilon)$ on the additive error for any DP algorithm that answers the triangle-motif size queries of all $(S,T)$-cut of $G$.  
  Finally, our algorithm generalizes to weighted graphs, and our lower bound extends to any $K_h$-motif cut for any constant $h\geq 2$.
\end{abstract}

\section{Introduction}
A graph can take on various forms: a social network, where each edge represents a friendship relation; a financial graph, where each edge represents a transaction relation; or a healthcare network containing patient disease information. In numerous applications, there is a pressing need to disseminate valuable information derived from these graphs while ensuring the protection of individual privacy. Differentially private (DP) algorithms, as pioneered by Dwork et al. 
\cite{dwork2006calibrating}, 
have demonstrated powerful abilities to address such tasks and have garnered significant attention in both academia and industry. Essentially, these DP algorithms ensure that, regardless of the adversary's knowledge about the graph, the privacy of individual users remains intact in the algorithm's output. Recently, there have been numerous studies on DP algorithms for graphs. Typical examples include privately releasing subgraph counting \cite{chen2013recursive,karwa2011private,zhang2015private,blocki2022make,nguyen2024faster}, degree sequence/distribution \cite{hay2009accurate,karwa2012differentially,proserpio2012workflow}, and cut queries \cite{gupta2012iterative,blocki2012johnson,upadhyay2013random,arora2019differentially,eliavs2020differentially,liu2024optimal}. We refer to the survey \cite{li2023private} for more examples and applications.

Of particular interest is the exploration of releasing a \emph{synthetic graph} that preserves the \emph{size} of all cuts of the original graph in a differentially private manner \cite{gupta2012iterative,blocki2012johnson,upadhyay2013random,arora2019differentially,eliavs2020differentially,liu2024optimal}. With such a synthetic graph, analysts can compute answers to queries concerning the cut size information of the graph without acquiring any private information from the input. 

The synthetic graphs described above effectively and privately preserve valuable lower-order structures such as the edge connectivity. However, they fail to provide \textbf{higher-order structural insights} of complex networks, particularly small network subgraphs (e.g. triangles), which are considered as fundamental building blocks for complex networks \cite{milo2002network}. Specifically, \emph{motifs} -- frequently occurring subgraphs -- play a crucial role across various domains: 

\begin{itemize}
   \item 
   \textbf{Frequent patterns:} Networks exhibit rich higher-order organizational structures. For example, triangle-motifs play a vital role in social networks as friends of friends tend to become friends themselves \cite{wasserman1994social}, two-hop paths are crucial for deciphering air traffic patterns \cite{rosvall2014memory} and feed-forward loops and bi-fans are recognized as significant interconnection patterns in various networks \cite{milo2002network}. 
   \item 
   \textbf{Graph clustering:} Works by \cite{benson2016higher, tsourakakis2017scalable} proposed graph clustering methods based on the concept of motif conductance, which measures the quality of a cluster by leveraging the number of instances of a given motif crossing a cut, rather than merely counting edges. These algorithms typically begin by constructing a motif weighted graph, where each edge is weighted by the number of copies of a given motif it contains, before applying spectral clustering. This approach effectively identifies clusters with high internal motif connectivity and low connectivity between clusters. Furthermore, motifs based embeddings provide a stronger inductive bias by more accurately capture the rich underlying community or cluster structures \cite{zhang2018arbitrary, nassar2020using}. 
   \item 
   \textbf{Motif cut sparsifier:} Kapralov et al.
   \cite{kapralov2022motif} introduced the notion of a \emph{motif cut sparsifier}, which serves as a {sparse weighted graph} well approximating the count of motifs crossing each cut within the original graph. Such a motif cut sparsifier stands as a natural extension of the concept of a cut sparsifier \cite{benczur1996approximating} and hypergraph cut sparsifier \cite{chen2020near}. 
\end{itemize}

Motivated by the aforementioned considerations on differential privacy in graphs, the widespread application of motifs and the utility of motif cut sparsifiers, we ask how one can efficiently release a synthetic graph that well preserves the motif size of all cuts:

\emph{Given a weighted graph $G=(V,E)$ and a motif $M$, how can we efficiently find another graph $G'=(V,E')$ in a differentially private manner such that for every $S\subset V$, the weight of the $M$-motif cut $(S,V\setminus S)$ in $G$ is approximated in $G'$ with a small error?}

By releasing a synthetic graph that maintains the size of the motif cut, data analysts can query the motif cut each time it is required; existing graph algorithms related to motif cuts can be directly applied to the synthetic graph without compromising privacy. For example, once the synthetic graph $G'$ is generated, motif cut sparsifiers \cite{kapralov2022motif} and higher-order clustering methods \cite{benson2016higher} can be utilized on $G'$ without disclosing sensitive information.

In this paper, we address the above question for the most fundamental motif, i.e., the triangle-motif. We give the first efficient DP algorithms for releasing synthetic graphs that well preserve the triangle-motif weights of all cuts. 
In addition, we provide lower bounds on the incurred additive error for any such DP algorithms. 

\subsection{Basic Definitions and Our Contributions}
Before we formally state our main result, we first introduce some  definitions. Given a graph $G=(V,E,\mathbf{w})$ with weight vector $\mathbf{w}\in\mathbb{R}_{+}^{\binom{V}{2}}$, and a graph 
$M=(V_M,E_M)$ which we assume to be a frequently occurring subgraph of $G$ and which is referred to as a \emph{motif}. An instance of motif $M$ is a subgraph of $G$, which is isomorphic to $M$, and the \emph{weight of an instance $I=(V_I,E_I)$} is defined\footnote{
As justified in \cite{kapralov2022motif}, in integer-weighted graphs, motifs are often viewed as unweighted multigraphs, where each edge is replaced by multiple copies based on its weight. Extending this idea, the motif weight is naturally defined as the product of its edge weights.
} as the product of its edge weights, i.e., $w(I)=\prod_{e\in E_I}w(e)$. 

Let $S,T$ be two disjoint subsets of $V$. The cut $(S,T)$ refers to the set of edges with one endpoint in $S$ and the other endpoint in $T$. If one of the edges of an instance $I=(V_I,E_I)$ crosses the cut $(S,T)$, we say that $I$ crosses this cut. It is also equivalent to $V_I\cap S\ne \emptyset$ and $V_I\cap T\ne \emptyset$ when the motifs are connected. Then the motif size of the cut $(S,T)$ is defined as the sum of weights of motif instances cross the cut. Formally, we have the definition below.

\begin{definition}[Motif size of a cut $(S,T)$
]
\label{def.motif_size}
Given $G=(V,E,\mathbf{w})$ and a connected motif $M$, let $\mathcal{M}(G,M)$ be the set of all instances or copies of $M$ in $G$. Then the \emph{$M$-motif size}  
of a cut $(S,T)$ is defined as
        $\mathrm{Cut}_{M}^{(G)}(S,T)
        =\sum_{I\in \mathcal{M}(G,M)
        :I\ \mathrm{crosses}\ (S,T)}w(I)$.
\end{definition}

As mentioned before, we will mainly focus on the triangle-motif, which is one of the most fundamental and well studied motifs \cite{tsourakakis2011triangle,satuluri2011local,benson2016higher,seshadhri2020impossibility}. We will use '$\triangle$' to denote '$M$' when referring specifically to triangle-motifs. Now we give the formal definition of differential privacy. Let $\mathcal{D}$ be some domain of datasets. 

\begin{definition}[$(\varepsilon,\delta)$-differential privacy; \cite{dwork2006calibrating}]\label{def.dp}
    Let $\mathcal{M}:\mathcal{D}\rightarrow \mathcal{R}$ be a randomized algorithm or mechanism, where $\mathcal{R}$ is the output domain. For fixed $\varepsilon >0$ and $\delta\in [0,1)$, we say that $\mathcal{M}$ preserves $(\varepsilon,\delta)$-differential privacy if for any measurable set $S\subset \mathcal{R}$ and any pair of neighboring datasets $x,y\in \mathcal{D}$, it holds that 
    $\mathsf{Pr}[\mathcal{M}(x)\in S] \leq \mathsf{Pr}[\mathcal{M}(y)\in S]\cdot e^{\varepsilon} + \delta.$
    If $\delta = 0$, we also say $\mathcal{M}$ preserves pure differential privacy (denoted by $\varepsilon$-DP).
\end{definition}

We consider the standard notion of edge privacy and two graphs are called \emph{neighboring} if the two vectors corresponding to their edge weights differ by at most $1$ in the $\ell_1$ norm (see \Cref{def:neighboring} for the formal definition). 

Let $\LSKt$ denotes the local sensitivity of triangle-motif cuts of $G$. Note that, its local sensitivity is defined as the maximum triangle-motif cut difference between $G$ and its neighboring graphs, equivalently, $\LSKt=\max_{(i,j)\in\binom{V}{2}}\sum_{s\in V\setminus \lbrace i,j\rbrace}
\mathbf{w}_{(i,s)}\mathbf{w}_{(j,s)}$.

Our main algorithmic contribution is the following DP algorithm for releasing a synthetic graph for approximating the triangle weights of all cuts of an input graph. 
We assume that the sum of edge weights is polynomially bounded by $n$. 

\begin{theorem}
    \label{thm:main}
There exists a polynomial time $(\varepsilon,\delta)$-DP algorithm that given an $n$-vertex weighted graph $G$ with total edge weight $W$ and maximum edge weight $w_{\max}$, 
outputs a weighted graph $G'$ 
such that with probability at least $3/4$, 
for any cut $(S,V\setminus S)$,
\[
\abs{\mathrm{Cut}_{\triangle}^{(G)}(S,V\setminus S)-\mathrm{Cut}_{\triangle}^{(G')}(S,V\setminus S)}=O(\sqrt{W\cdot \LSKt}\cdot nw_{\max}/\varepsilon^{\frac{3}{2}}\cdot \log^{2} (n/\delta))).
\]
\end{theorem}

Note that for unweighted graphs with $m$ edges, 
it holds that $W=m$ and $w_{\max}=1$. Thus, our algorithm achieves $\tilde{O}(\sqrt{m\cdot \LSKt}\cdot n/\varepsilon^{3/2})$  
additive error. Note that for certain classes of graphs $G$, $\LSKt$ can be relatively small (in comparison to $n$). For example, if $G$ has maximum degree at most $d$, then $\LSKt\leq d$. (See \Cref{sec:lowerbound} for more discussions of local sensitivity.) Our bound can be compared to the trivial upper bounds $O(m^{3/2}), O(n^{3})$ or $O(mn)$ on the triangle-motif size of any cut in a graph. 
We remark that before our work, there was no known polynomial-time DP algorithm for this problem. We also give a $(\varepsilon,0)$-DP algorithm that outputs a synthetic graph (in \Cref{appendix:randomresponse}) with error $\tilde{O}(n^{\frac{5}{2}})$ based on randomized response, using an analysis similar to \cite{gupta2012iterative} for the edge case. Note that the bound $O(\sqrt{m\cdot \LSKt}\cdot n)$ from \Cref{thm:main} is never worse than, and can often be much better than, the $\tilde{O}(n^{\frac{5}{2}})$ bound from randomized response. 

Our algorithms are built upon the framework of solving a related convex problem using the mirror descent approach and adding noise appropriately, similar to the approach in \cite{eliavs2020differentially}. However, unlike the edge case, answering motif cut queries on a graph does \emph{not} align with the extensively studied problem of query release for exponentially sized families of linear queries on a dataset, due to the inherently non-linear nature of motif cut queries. Consequently, we had to resolve several technical challenges (see \Cref{sec:techniques}).

Furthermore, since the convex problem we are using can also be adjusted for other $3$-vertex motifs (e.g., wedges), our algorithmic results can be readily extended to these settings as well. 

Next, we present a lower bound on the additive error for any DP algorithm that answers triangle-motif size queries for all $(S, T)$-cuts of $G$.

\begin{theorem}\label{thm:lowerbound}
Let $\mathcal{M}$ be an $(\varepsilon,\delta)$-differentially private mechanism, and let $G$ be a graph generated from $G(n,p)$ with $((\log n)/n)^{1/2}\ll p\leq \frac12$. In this case, $G$ has 
$m=\Theta(n^2p)$ edges with high probability. If $\mathcal{M}$ answers the triangle-motif size queries of all $(S,T)$-cut on $G$, or on a scaled version of version of $G$ with total edge weight $W$, up to an additive error $\alpha$ with probability at least $3/4$, 
then 
$\alpha\geq \Omega\left(\max\left(
{\frac{\sqrt{mn}\cdot\LSKt}{\varepsilon}}(1-c),{\frac{\sqrt{Wn}\cdot\LSKt}{\sqrt{\varepsilon}}}(1-c)\right)\right)$, 
where\footnote{Note that $c$ is a very small number as $\delta$ is usually set to be smaller than inverse of any polynomial in $n$.} $c=\frac{12(e-1)\delta}{e^{\varepsilon}-1}$.
\end{theorem}

We remark that by the proof of the above lower bound, we can also obtain a lower bound of $\Omega(\sqrt{\frac{mn}{\varepsilon}}\ell_3(G)\cdot(1-c))$ for unweighted graphs which may contain multiple edges (see the proof of \Cref{thm:lowerfinal}). 
Furthermore, our lower bound also generalizes to all $K_h$ motifs, i.e., the complete graph on $h$ vertices, for any constant $h\geq 2$. For the $K_h$ motifs, our lower bound on the additive error becomes 
$\Omega\left(\max\left(
{\frac{\sqrt{mn}\cdot\LSKh}{\varepsilon}}(1-c),{\frac{\sqrt{Wn}\cdot\LSKh}{\sqrt{\varepsilon}}}(1-c)\right)\right)$ (while the corresponding requirement for $p$ is $((\log n)/n)^{1/(h-1)}\ll p\leq \frac12$), where $\LSKh$ is the local sensitivity of $K_h$-motif cuts of $G$ (see \Cref{thm:lowerfinal}). Note that it also recovers the $\Omega(\sqrt{mn})$ lower bound for the edge motif case (as here $\LSKtwo=1$).  
Our lower bound is based on the connections to discrepancy theory shown by  \cite{muthukrishnan2012optimal} and a lower bound on the discrepancy of some coloring function on hypergraphs.

\subsection{Related Work}\label{sec:relatedwork}
The research on differentially private algorithms for answering edge cut queries has garnered significant attention within the community. Gupta, Roth, and Ullman 
\cite{gupta2012iterative} introduced a framework called iterative database construction algorithms (IDC) and utilized this framework to design various mechanisms that can provide private responses to edge cut queries. These mechanisms include the Multiplicative Weight Update IDC, which is based on the private Multiplicative Weight Update algorithm \cite{hardt2010multiplicative}, with an additive error of approximately $\tilde{O}(m^{\frac{1}{2}}n^{\frac{1}{2}}{/\varepsilon^{\frac{1}{2}}})$, as well as the Frieze and Kannan IDC, which is based on the Frieze and Kannan low-rank decomposition algorithm \cite{frieze1999quick} and has an additive error of approximately ${O}(m^{\frac{1}{4}}n{/\varepsilon^{\frac{1}{2}}})$.

However, the above algorithms can only release the value of edge cuts. In order to release a synthetic graph that approximates every edge cut, Gupta, Roth, and Ullman 
\cite{gupta2012iterative} gave an $(\varepsilon,0)$-differentially private algorithm with additive error $O(n\sqrt{n}/\varepsilon)$, which is based on randomized response. After that, Blocki et al. 
\cite{blocki2012johnson} showed an algorithm that can release a synthetic graph answering $k$ predetermined queries with additive error at most $O(n\sqrt{\log k}/\varepsilon)$, hence $O(n\sqrt{n}/\varepsilon)$ when considering all edge cuts, by a nice utilization of Johnson-Lindenstrauss transform. Subsequent advancements by Upadhyay \cite{upadhyay2013random}, Arora and Upadhyay \cite{arora2019differentially}, Eli{\'a}{\v{s}}, Kapralov, Kulkarni, and Lee 
\cite{eliavs2020differentially}, and Liu, Upadhyay, and Zou 
\cite{liu2024optimal} introduced refined algorithms. In particular, Eli{\'a}{\v{s}}, Kapralov, Kulkarni, and Lee 
\cite{eliavs2020differentially} came up with an mirror descent based algorithm which can achieve $\tilde{O}(\sqrt{\Vert\mathbf{w}\Vert_1n{/\varepsilon}})$ error guarantee, and Liu, Upadhyay, and Zou 
\cite{liu2024optimal} designed an algorithm that can  achieve $\tilde{O}(\sqrt{mn}{/\varepsilon})$ error guarantee, which is based on a topology sampler and the EKKL approach.

There have been several theoretical work on counting motifs (e.g. triangles) differentially privately \cite{chen2013recursive,karwa2011private,zhang2015private,blocki2022make,nguyen2024faster}. Chen and Zhou 
\cite{chen2013recursive} provided a solution of subgraph counting to achieve node DP, for any kind of subgraphs. There also exist efficient polynomial time private algorithms for subgraph counting focused on $k$-triangles, $k$-stars and $k$-cliques \cite{karwa2011private,zhang2015private,nguyen2024faster}. Several works studied triangle counting problem under the setting of local differential privacy \cite{eden2023triangle,imola2022communication}. Note that our goal problem is to privately release a synthetic graph that can answer all the triangle-motif cut queries including triangle counting (by querying the sparsifier on the $n$ singleton
cuts), hence it is a harder problem than differentially private triangle counting. 

Motif analysis has had a profound impact on graph clustering \cite{benson2016higher,yin2017local,tsourakakis2017scalable}. Differentially private graph clustering also received increasing attention recently in the community \cite{bun2021differentially,cohen2022near,imola2023differentially}.

Graph sparsification is a well-known technique to speed up the algorithms based on edge cut \cite{karger1994random,benczur2015randomized}, which is to construct a sparse graph that approximates all the cuts within a $1+\varepsilon$ factor. Bencz{\'u}r and Karger 
\cite{benczur1996approximating} achieved a cut sparsifier with $O(n/\varepsilon^2)$ size in nearly linear time, and Chen, Khanna and Nagda 
\cite{chen2020near} constructed a hypergraph sparsifier with near-linear size. 

For the study of motif sparsification, Tsourakakis, Kolountzakis and Miller 
\cite{tsourakakis2011triangle} came up with an algorithm that can obtain a sparse subgraph that preserves the global triangle count. After that,  Kapralov et al. 
\cite{kapralov2022motif} obtained a stronger algorithm, which can release a sparse subgraph that preserve count of motifs crossing each cut. Moreover, Sotiropoulos and Tsourakakis 
\cite{sotiropoulos2021triangle} introduced a triangle-aware spectral sparsifier, which is sparsifier with respect to edges that has better space bounds for graphs containing many triangles.

\subsection{Our Techniques}\label{sec:techniques}

Our algorithm is inspired by the work of Eli{\'a}{\v{s}}, Kapralov, Kulkarni, and Lee \cite{eliavs2020differentially} for privately releasing synthetic graph for edge cut, which we abbreviated as the EKKL approach. We will start by recalling their techniques and then show why the immediate extension of their approach still fails. Finally, we show how we overcome the difficulties and present our algorithm.

\paragraph{The EKKL approach for the edge-motif cut structure} 

The main idea underlying the EKKL approach is as follows. Let $\overline{G}$ be a graph with sum of edge weights at most $W$. 
For any fixed graph $G$, the cut difference (or distance), i.e., the maximum difference in weight of some $(S,T)$-cut in $\overline{G}$ and $G$, can be bounded by the SDP:
$\max_{\mathbf{X}\in\mathcal{D}}\lbrace(\begin{array}{cc}
    \mathbf{0} & \mathbf{A}-\overline{\mathbf{A}} \\
    \mathbf{A}-\overline{\mathbf{A}} & \mathbf{0}
\end{array})\bullet\mathbf{X}\rbrace$, 
where $\mathbf{A}$ and $\overline{\mathbf{A}}$ denote the weighted adjacency matrices of $G$ and $\overline{G}$, respectively, and the domain \[\mathcal{D}=\left\lbrace\mathbf{X}\in\mathbb{R}^{2n}:\mathbf{X}\ \mathrm{is}\ \mathrm{symmetric},
    \mathbf{X}\succeq \frac{1}{n}\mathbf{I}_{2n},
    \mathrm{and}\ \mathbf{X}_{ii}=1\ \mathrm{for}\ \forall{i}
    \right\rbrace.\] Then to find a synthetic graph $G$ that well approximates $\overline{G}$ privately, the algorithm in \cite{eliavs2020differentially} tries to solve the following optimization problem:
\[\min_{\mathbf{w}\in\mathcal{X}'}
    \max_{\mathbf{X}\in\mathcal{D}}\lbrace(\begin{array}{cc}
    \mathbf{0} & \mathbf{A}-\overline{\mathbf{A}} \\
    \mathbf{A}-\overline{\mathbf{A}} & \mathbf{0}
\end{array})\bullet\mathbf{X}
+\lambda\log\det(\mathbf{X})\rbrace
,\] 
where $\mathcal{X}'=\lbrace
\mathbf{w}\in\mathbb{R}_+^{\binom{V}{2}}
:\sum_{e\in\binom{V}{2}}\mathbf{w}_e=W\rbrace$ and $W$ is the total edge weight. Note that in the above, the regularizer $\lambda\log\det\mathbf{X}$ term is added to the original SDP 
to control the privacy. 

Since the objective function is convex with respect to $\mathbf{w}$, the EKKL approach finds a nearly-optimal weight vector $\mathbf{w}$ by stochastic mirror descent (see e.g. \cite{bubeck2015convex}), which can approximately solve a convex optimization using an iterative process. In each iteration, the mirror descent algorithm requires an approximation to the gradient of the objective function. 
It further add noise in each iteration by using Johnson-Lindenstrauss transform 
for privacy guarantee. 
Then one can derive the tradeoffs between privacy, utility, the parameter $\lambda$ and the number of iterations, from which we can obtain a $(\varepsilon,\delta)$-DP algorithm for the edge cut synthetic graph with small additive error. 

Directly applying the EKKL approach to the associated hypergraph or the triangle-motif weighted graph do not seem work due to the \emph{non-linearity} property of triangle cut.  
That is, the sum of the triangle-motif sizes for cut $(S,V\setminus S)$ of two graphs $G_1$ and $G_2$, is not equal to the triangle-motif size for cut $(S,V\setminus S)$ of the sum of the two graphs $G_1+G_2$, or formally, $\mathrm{Cut}_{\triangle}^{(G_1)}(S,V\setminus S)
+\mathrm{Cut}_{\triangle}^{(G_2)}(S,V\setminus S)
\ne \mathrm{Cut}_{\triangle}^{(G_1+G_2)}(S,V\setminus S)$. We compare in more details about the difference of the EKKL approach and ours in \Cref{appendix:alg:comparison} and also discuss some tempting but unsuccessful approaches in \Cref{appendix:discuss:tempting_methods}.

\paragraph{Our approach for the triangle-motif cut structure}
Now we discuss our approach. For a graph $G=(V,E,\mathbf{w})$, we let $\mathbf{A}_\triangle$ be the triangle adjacency matrix of $G$ such that $(\mathbf{A}_\triangle)_{i,j}$ is the sum of weights of triangle instances in which both $i$ and $j$ are involved. Note that $\mathbf{A}_\triangle$ is a matrix whose entries depend non-linearly on $\mathbf{w}$. We consider the following optimization problem extended from the one for the edge case: 
\begin{equation}
    \label{eq.our_primal_opt}
    \min_{\mathbf{w}\in\mathcal{X}'}
\max_{\mathbf{X}\in\mathcal{D}}\lbrace(\begin{array}{cc}
    \mathbf{0} & \mathbf{A}_{\triangle}-\overline{\mathbf{A}}_{\triangle} \\
    \mathbf{A}_{\triangle}-\overline{\mathbf{A}}_{\triangle} & \mathbf{0}
\end{array})\bullet\mathbf{X}
+\lambda\log\det(\mathbf{X})\rbrace,
\end{equation}
where $\overline{\mathbf{A}}_{\triangle}$ is the triangle adjacency matrix of the input graph $\bar{G}=(V,E,\overline{\mathbf{w}})$ with sum of edge weights at most $W$. 
Intuitively, finding a weight vector $\mathbf{w}$ that minimizes the inner SDP of \Cref{eq.our_primal_opt} will result in a synthetic graph with a low triangle-motif cut difference between $\mathbf{w}$ and $\overline{\mathbf{w}}$. It is crucial that the objective function is defined in terms of the weight function $\mathbf{w}$ of the target graph $G$ instead of $\mathbf{w}_{\triangle}$, which is the edge weight vector of the triangle-motif weighted graph ${G}_{\triangle}$. 

However, due to the non-linearity property of triangle-motif cut, \Cref{eq.our_primal_opt} is no longer convex with respect to $\mathbf{w}$, thus it cannot be solved by convex optimization techniques. Our approach is to add a convexity regularizer\footnote{Here, we use the term $6nw_{\max}\sum_{e\in\binom{V}{2}}(\mathbf{w}_e-\overline{\mathbf{w}}_e)^2$ to provide a basic understanding of the main idea; in our actual proof, we incorporate a slightly more complicated term to achieve a smaller error.} $6nw_{\max}\sum_{e\in\binom{V}{2}}(\mathbf{w}_e-\overline{\mathbf{w}}_e)^2$ to 
the objective function to control the convexity, while ensuring that we do not add too much error. Define
\[
F_{\triangle}(\mathbf{w},\mathbf{X}):=(\begin{array}{cc}
    \mathbf{0} & \mathbf{A}_{\triangle}-\overline{\mathbf{A}}_{\triangle} \\
    \mathbf{A}_{\triangle}-\overline{\mathbf{A}}_{\triangle} & \mathbf{0}
\end{array})\bullet\mathbf{X}
+\lambda\log\det(\mathbf{X})+6nw_{\max}\sum_{e\in\binom{V}{2}}(\mathbf{w}_e-\overline{\mathbf{w}}_e)^2
\]
Then we aim to solve: 
$\min_{\mathbf{w}\in\mathcal{X}}
\max_{\mathbf{X}\in\mathcal{D}}F_{\triangle}(\mathbf{w},\mathbf{X})$, 
for some appropriately defined domain $\mathcal{X}$.

With the convexity regularizer, we can apply Danskin's theorem \cite{danskin2012theory} to show that our objective function 
$f_{\triangle}(\mathbf{w}):=\max_{\mathbf{X}\in\mathcal{D}}F_{\triangle}(\mathbf{w},\mathbf{X})$ 
is convex. However, solving the corresponding optimization using the stochastic mirror descent becomes more challenging. Specifically, we need to utilize the gradient of $ f_{\triangle}(\mathbf{w}) $ for updating the mirror descent, which corresponds to finding the minimum of the Bregman divergence associated to some mirror map function (see \Cref{pre.convex}), and bounding the final additive error. 
Unlike the edge case, where the gradient is a constant, the gradient of $ f_{\triangle}(\mathbf{w}) $ is a relatively complex function that depends on $ \mathbf{w} $ (see Lemma~\ref{l.grad_conv}). This dependency arises from the higher-order nature of the triangle-motif and introduces additional technical difficulties.

Our solution is to impose additional constraints $ \{\mathbf{w}_e \leq w_{\max}, e\in\binom{V}{2}\} $ on $ \mathcal{X}' $ to obtain a more restricted domain $ \mathcal{X} $, and then 
introduce a new mirror descent update rule based on a greedy method. Specifically, we transform the problem of updating $\mathbf{w}$, i.e., finding the minimizer of the Bregman divergence, into another convex optimization problem (see \Cref{optDphiwy}). Utilizing the Karush-Kuhn-Tucker (KKT) conditions \cite{ghojogh2021kkt}, we derive the necessary and sufficient conditions that the optimal solution of this optimization problem must satisfy. 
Based on these conditions, we develop an efficient algorithm (\Cref{alg.update}) to update the weight vector $\mathbf{w}$ using a greedy method.

Once the mirror descent step is defined and an updated solution $\mathbf{w}$ is obtained, we apply the Johnson-Lindenstrauss transform to privately release $\mathbf{X}$. This allows us to analyze the trade-offs between privacy, utility, the parameter $\lambda$, and the number of iterations, as in the edge case. Based on these trade-offs, we derive a $(\varepsilon,\delta)$-DP algorithm for generating a triangle-motif cut synthetic graph with a small additive error.

\paragraph{Overview of lower bound}
Our lower bound is established within the discrepancy framework introduced by Muthukrishnan and Nikolov 
\cite{muthukrishnan2012optimal}. This framework was also utilized in proving the lower bound for DP algorithms concerning edge motif cut in \cite{eliavs2020differentially}. Essentially, if there exists a DP algorithm $\mathcal{M}$ for the motif cut problem with additive error smaller than the discrepancy of the motif cut size function over certain classes of graphs, then algorithm $\mathcal{M}$ can be exploited to approximately recover the input, thereby compromising privacy and resulting in a contradiction. To leverage this framework effectively, we need to show that the discrepancy of this function over some class of graphs is relatively large to exclude DP algorithms with small additive error. Specifically, to derive the lower bound for the $K_h$-motif cut for any complete graph $K_h$ with $h$ vertices, we employ the discrepancy of $3$-colorings (i.e. each motif is colored with $+1,-1$ or $0$) of $h$-uniform hypergraphs, which is a generalization of the corresponding discrepancy of graphs. In particular, we identify a set of properties (e.g., each vertex has roughly the same degree and each edge belongs to roughly the same number of $K_h$ instances) sufficient to ensure that: 1) a graph from $G(n,p)$ will satisfy with high probability, and 2) the corresponding $h$-uniform hypergraph (defined by treating any subset of $h$ vertices as a hyperedge) exhibits large discrepancy.

\paragraph{Key Technical Differences between Our Approach and the EKKL Approach}
\label{appendix:alg:comparison}
Here we highlight and summarize the key technical differences between our approach and the EKKL approach \cite{eliavs2020differentially}: 
\begin{itemize}
    \item We introduce a convexity regularizer to ensure the optimization problem for privately releasing a graph preserving the triangle-motif cut structure is convex.
    \item After adding the convexity regularizer, the mirror descent step used by the EKKL approach becomes invalid as the  gradient of our objective function depends on $\mathbf{w}$, due to the higher-order structure of the triangle-motif. Thus, we reformulate the problem of updating the descent, i.e., for minimizing the Bregman divergence, as a new convex optimization problem with appropriate constraints.
    \item We introduce a new greedy algorithm for the mirror descent update step. This greedy algorithm is guaranteed to output a solution that satisfies the KKT conditions and thus ensures valid updates. In contrast, the EKKL method uses a straightforward update rule for the descent step.
    \item Our lower bound requires proving the discrepancy of $3$-colorings of $h$-uniform hypergraphs for any constant $h\geq 2$, while the work in \cite{eliavs2020differentially} only proves the discrepancy of graphs which corresponds to $h=2$.
\end{itemize}

\section{Preliminaries}
Here, we give the definitions of differential privacy, motif adjacency matrix, and introduce the convex optimization that will be utilized by the algorithms and their analysis.

\subsection{Differential Privacy}
\label{pre.dp}
The definition of differential privacy (Definition~\ref{def.dp}) relies on the definition of neighboring datasets. For weighted graphs, we have the following definition of neighboring graphs:

\begin{definition}[Neighboring graphs]
    \label{def:neighboring}
    Given weighted graphs $G$ with edge weight vector $\mathbf{w}$ and $G'$ with edge weight vector $\mathbf{w}'$, $G$ and $G'$ are called neighboring graphs if $\mathbf{w}$ and $\mathbf{w}'$ differ by at most 1 in the $\ell_1$ norm, i.e., $\norm{\mathbf{w}-\mathbf{w}'}_1\leq1$.
\end{definition}

For an unweighted graph $G$, we let $\mathbf{w}_e=1$ if $G$ has the edge $e$; otherwise, $\mathbf{w}_e=0$. Therefore, for unweighted graphs, $G$ and $G'$ are neighboring graphs if they differs by exactly $1$ edge.

A key feature of differential privacy algorithms is that they preserve privacy under post-processing. That is, without any auxiliary information about the dataset, any analyst cannot compute a function that makes the output less private. 

\begin{lemma}[Post processing~\cite{dwork2014algorithmic}]\label{l.post_processing}
  Let $\mathcal{M}:\mathcal{D}\rightarrow \mathcal{R}$ be a $(\varepsilon,\delta)$-differentially private algorithm. Let $f:\mathcal{R}\rightarrow \mathcal{R}'$ be any function, then $f\circ \mathcal{M}$ is also $(\varepsilon,\delta)$-differentially private.
\end{lemma}

Sometimes we need to repeatedly use differentially private mechanisms on the same dataset, and obtain a series of outputs. The privacy guarantee of the outputs can be derived by following theorems.

\begin{lemma}[Adaptive composition~\cite{dwork2006calibrating}]\label{l.adaptive_composition}
  Suppose $\mathcal{M}_1(x):\mathcal{D} \rightarrow \mathcal{R}_1$ is $(\varepsilon_1,\delta_1)$-differentially private and $\mathcal{M}_2(x,y):\mathcal{D} \times \mathcal{R}_1\rightarrow \mathcal{R}_2$ is $(\varepsilon_2,\delta_2)$-differentially private with respect to $x$ for any fixed $y\in \mathcal{R}_1$, then the composition $(\mathcal{M}_1(x), \mathcal{M}_2(x,\mathcal{M}_1(x)))$
  is $(\varepsilon_1 + \varepsilon_2, \delta_1 + \delta_2)$-differentially private.
\end{lemma}

\begin{lemma}
    [Advanced composition lemma~\cite{dwork2010boosting}] 
    \label{l.adv_composition}
    For parameters $\varepsilon>0$ and $\delta,\delta'\in [0,1]$, the composition of $k$ $(\varepsilon,\delta)$-differentially private algorithms is a $(\varepsilon', k\delta+\delta')$-differentially private algorithm, where 
    $\varepsilon' = \sqrt{8k\log(1/\delta')}.$
\end{lemma}

Now, we introduce basic mechanisms that preserve differential privacy, which are ingredients that build our algorithm. First, we define the sensitivity of query functions.

\begin{definition}
  [$\ell_p$-sensitivity]\label{def.sens} Let $f:\mathcal{D}\rightarrow \mathbb{R}^k$ be a query function on datasets. The sensitivity of $f$ (with respect to $\mathcal{X}$) is 
  $\Delta_p (f) = \max_{x,y\in \mathcal{D} \atop x\sim y} \|f(x) - f(y)\|_p$.
\end{definition}

Based on the definition of sensitivity and Laplace distribution, we can get a mechanism that preserve differential privacy as follows.

\begin{lemma}[Laplace mechanism]\label{l.laplace}
  Suppose $f:\mathcal{D}\rightarrow \mathbb{R}^k$ is a query function with $\ell_1$ sensitivity $\Delta_1(f)\leq \Delta$. Then the mechanism
  $\mathcal{M}(D) = f(D) + (Z_1,\cdots,Z_k)^\top$
  where $Z_1,\cdots, Z_k$ are i.i.d random variables drawn from $\lap\left({\frac{\Delta}{\varepsilon}}\right)$. Given $b>0$,  $\lap\left({b} \right)$ is the Laplace distribution with density $$\lap(x;b) := \frac{1}{2b} \exp\left(-\frac{|x|}{b}\right).$$
\end{lemma}

The Laplace distribution has the following concentration bound,
\begin{lemma}
    [Laplace concentration bound]\label{l.lap_bound}
    If $Y\sim\lap(b)$, then for any $t>0$, we have
    $$\mathsf{Pr}[\vert Y\vert\geq tb]=\exp(-t)$$
\end{lemma}

Another important concept is probabilistic differential privacy. It is defined as follows.

\begin{definition}[$(\varepsilon,\delta)$-probabilistic differential privacy]\label{def.pdp}
    For fixed $\varepsilon >0$ and $\delta\in [0,1)$, we say that $\mathcal{M}$ preserves $(\varepsilon,\delta)$-probabilistic differential privacy if $\mathcal{M}$ is $\varepsilon$-differentially private with probability at least $(1-\delta)$, i.e., for any pair of neighboring datasets $x,y\in \mathcal{D}$, there is a set $S^\delta\subset \mathcal{R}$ with $\mathsf{Pr}[\mathcal{M}(x)\in S^\delta]\leq\delta$, s.t. for any measurable set $S\subset \mathcal{R}$, it holds that 
    $$\mathsf{Pr}[\mathcal{M}(x)\in S] \leq \mathsf{Pr}[\mathcal{M}(y)\in S]\cdot e^{\varepsilon}.$$
\end{definition}

If a mechanism preserves $(\varepsilon,\delta)$-probabilistic differential privacy, then it also preserves $(\varepsilon,\delta)$-differential privacy, while the opposite direction does not hold \cite{meiser2018approximate}.

\subsection{Motif Cut and Motif Adjacency Matrix}
\label{pre.motif_cut}
We will make use of the motif adjacency matrix to deal with the motif cut, which is defined as follows.

\begin{definition}[Motif adjacency matrix~\cite{benson2016higher}]\label{def.motif_matrix}
    Given $G=(V,E,\mathbf{w})$ and a motif $M$, a motif adjacency matrix of $G$ with respect to $M$ is defined by,
    \begin{equation*}
        (\mathbf{A}_M)_{i,j}
        =\sum_{I\in \mathcal{M}(G,M)
        :i,j\in V_I}w(I).
    \end{equation*}
That is, $(\mathbf{A}_M)_{i,j}$ is the sum of weights of motif instances in which both $i$ and $j$ are involved.
\end{definition}

Notably, the computational time to form a motif adjacency matrix $\mathbf{A}_M$ is bounded by the time to find all instances of the motif in the graph. And obviously, for a motif on $k$ nodes, we can compute $\mathbf{A}_M$ 
in $\Theta(n^k)$ time by checking all the $k$-tuples of nodes in a graph. When the motif is triangle, there exist more efficient algorithms to list all the motifs \cite{benson2016higher}.

\begin{fact}
Let $\mathbf{D}_M^{(e)}$ denote the derivative of $\mathbf{A}_M$ at $\mathbf{w}_e$ for some $e=(k,\ell)\in \binom{V}{2}$.

\begin{equation*}
    \left(\mathbf{D}_M^{(k,\ell)}\right)_{i,j}
    =\sum_{I\in \mathcal{M}(K^n,M)
    :i,j,k,\ell\in V_I}w^{(e)}(I),
\end{equation*}
where $w^{(e)}(I)=\prod_{e'\in E_I:e'\ne e}w(e')$.   
\end{fact}
It is worth noting that $\mathbf{D}_M^{(e)}$ is actually the divergence of motif adjacency matrix between graph $G$ and its neighboring graph $G'$, which differs from $G$ in only one edge $e$ by $1$ weight. 

Furthermore, the form of the second-order or higher order derivative of $\mathbf{A}_M$ are similar to the above.

In this paper, we focus on the special case when the motif is triangle. The ``$M$'' in the relevant notation will be replaced by ``$\triangle$''.  

In addition, when we consider triangles, the derivatives of $\mathbf{A}_{\triangle}$ have explicit form as follows. Recall that $\mathbf{D}_{\triangle}^{(k,\ell)}$ denotes the derivative of $\mathbf{A}_{\triangle}$ at $\mathbf{w}_e$, where $e=(k,\ell)$. 
\begin{fact}
\label{def.D.E}
Let $\mathbf{E}_{\triangle}^{((i,j),(k,\ell))}$ denote the second-order derivative of $\mathbf{A}_{\triangle}$ at $\mathbf{w}_{(i,j)}$ and $\mathbf{w}_{(k,\ell)}$, and let $\mathbf{B}_{(e)}$ denote a matrix which has value $1$ at the entry corresponding to $e$ and value $0$ otherwise. Then it holds that
\begin{equation*}
    \left(\mathbf{D}_{\triangle}^{(k,\ell)}\right)_{i,j}
    =\begin{cases}
        \sum_{s\ne k,\ell}\mathbf{w}_{(k,s)}\mathbf{w}_{(s,\ell)}, &(k,\ell)=(i,j)\\
        \mathbf{w}_{(k,j)}\mathbf{w}_{(j,\ell)}, &k=i\ and\ j\ne\ell\\
        0, &o.w.\\
    \end{cases}
\end{equation*}
\begin{equation*}
    \mathbf{E}_{\triangle}^{((i,j),(k,\ell))}
    =\begin{cases}
        \mathbf{w}_{(j,\ell)}\cdot(\mathbf{B}_{(i,j)}
        +\mathbf{B}_{(i,\ell)}+\mathbf{B}_{(j,\ell)}), &i=k\ and\ j\ne \ell\\
        0, &o.w.\\
    \end{cases}
\end{equation*}
\end{fact}

\subsection{Cut Norm and Its Approximation}
\label{pre.cut_norm}
The following method to bound the additive error for edge cut is introduced by \cite{alon2004approximating}.

Consider graphs $\overline{G}$ and $G$, where $\overline{G}$ can be considered as the original graph, and $G$ can be considered as an approximation to $\overline{G}$. Let $\overline{\mathbf{A}}$ and $\mathbf{A}$ be their adjacency matrices. Then for a fixed cut $(S,V\setminus S)$, the edge-cut size 
of $\overline{G}$ is $\sum_{v\in S,u\in V\setminus S}\left(\overline{\mathbf{A}}\right)_{u,v}$ and the edge-cut size of $G$ is $\sum_{v\in S,u\in V\setminus S}\left(\mathbf{A}\right)_{u,v}$. So we can see that the edge-cut difference between $G$ and $\overline{G}$ is the following expression: 
\begin{equation*}
    \max_{S\subset V}
    \left\lbrace\left|
    \sum_{v\in S,u\in V\setminus S}\left(\overline{\mathbf{A}}\right)_{u,v}
    -\sum_{v\in S,u\in V\setminus S}\left(\mathbf{A}\right)_{u,v}
    \right|\right\rbrace
    =\max_{\mathbf{x},\mathbf{y}\in \{0,1\}^n,
    \mathbf{x}+\mathbf{y}=\mathbf{1}}\left\lbrace
    \left\vert\mathbf{x}^\top\left(
    \overline{\mathbf{A}}-\mathbf{A}
    \right)\mathbf{y}\right\vert
    \right\rbrace
\end{equation*}

This expression can be bounded by the cut norm which is defined as follows.

\begin{definition}[Cut norm \cite{frieze1999quick}]\label{def.cut_norm}
    For a matrix $\mathbf{M}\in \mathbb{R}^{n\times n}$ its \emph{cut norm} is defined as
    \begin{equation*}
        \left\Vert \mathbf{M}\right\Vert_{cut}
        =\max\left\lbrace
        \left|\mathbf{x}^\top\mathbf{M}\mathbf{y}\right|
        :\mathbf{x},\mathbf{y}\in \{0,1\}^n
        \right\rbrace
    \end{equation*}
\end{definition}
The following lemma was due to \cite{alon2004approximating}.
\begin{lemma}[\cite{alon2004approximating}]
The cut norm of the matrix $\mathbf{M}$ can be approximated up to a constant factor using the following SDP.
\begin{equation*}
    \max\left\lbrace
        \sum_{i,j=1}^n\mathbf{M}_{i,j}
        \mathbf{u}_i^\top\mathbf{v}_j
        :\mathbf{u}_i,\mathbf{v}_i\in \mathbb{R}^n,
        \Vert\mathbf{u}_i\Vert=\Vert\mathbf{v}_i\Vert=1,\forall{i}
        \right\rbrace
\end{equation*}   
\end{lemma}

For any two $n\times n$ matrices $\mathbf{B}, \mathbf{C}$, we let $\mathbf{B}\bullet\mathbf{C}
=\mathrm{tr}\left(\mathbf{B}^\top\mathbf{C}\right)
=\sum_{i,j=1}^n\mathbf{B}_{i,j}\mathbf{C}_{i,j}$. We have the following lemma. 

\begin{lemma}
\label{lemma:diff_primal_sdp}
    The edge-cut difference between $G$ and $\overline{G}$ can be bounded by the following SDP up to a constant factor:

    \begin{equation}\label{eq.primal_sdp}
    \max\left\lbrace
    \left(\begin{array}{cc}
        \mathbf{0} & \mathbf{A}-\overline{\mathbf{A}} \\
        \mathbf{A}-\overline{\mathbf{A}} & \mathbf{0}
    \end{array}\right)
    \bullet \mathbf{X}
    :\mathbf{X}\ \mathrm{is}\ \mathrm{symmetric},
    \mathbf{X}\succeq \mathbf{0},
    \mathrm{and}\ \mathbf{X}_{i,i}=1\ \mathrm{for}\ \forall{i}
    \right\rbrace
    \end{equation}
\end{lemma}

\begin{proof}

Let $\mathbf{z}_i=(\mathbf{0},\mathbf{u}_i)\in\mathbb{R}^{2n}$ for $i=1,\dots,n$ and $\mathbf{z}_i=(\mathbf{0},\mathbf{v}_i)\in\mathbb{R}^{2n}$ for $i=n+1,\dots,2n$.  
Define $\mathbf{X}\in \mathbb{R}^{2n\times 2n}$ s.t. $\mathbf{X}_{i,j}=\mathbf{z}_i^\top\mathbf{z}_j$ for $\forall i,j$. Note that $\mathbf{X}$ is symmetric, $\mathbf{X}\succeq \mathbf{0}$, and $\mathbf{X}_{i,i}=1$ for $\forall{i}$. Since $\mathbf{M}:=\mathbf{A}-\overline{\mathbf{A}}$ is symmetric, then, 
\begin{equation*}
\sum_{i,j}^n\mathbf{M}_{i,j}\mathbf{u}_i^\top\mathbf{v}_j
=\frac{1}{2}\sum_{i,j}^n\left(
\mathbf{M}_{i,j}\mathbf{u}_i^\top\mathbf{v}_{j}
+\mathbf{M}_{j,i}\mathbf{u}_j^\top\mathbf{v}_{i}\right)
=\frac{1}{2}\sum_{i,j}^n\mathbf{M}_{i,j}
\left(\mathbf{z}_i^\top\mathbf{z}_{n+j}
+\mathbf{z}_j^\top\mathbf{z}_{n+i}\right)
=\frac{1}{2}\left(\begin{array}{cc}
        \mathbf{0} & \mathbf{M} \\
        \mathbf{M} & \mathbf{0}
    \end{array}\right)
    \bullet \mathbf{X}.
\end{equation*}
\end{proof}

\subsection{Convex Optimization}
\label{pre.convex}
Our algorithm needs to solve a minimization problem for a convex function while preserving differential privacy. The method we use for convex optimization is stochastic mirror descent, which is fully analyzed by \cite{bubeck2015convex}.   
Given a convex function $f(x)$ defined over a convex set $\mathcal{X}$, a mirror map $\Phi(x)$ (a strongly convex function), and step length $\eta$, we define here the Bregman divergence associated to $\Phi$ as 
\[D_{\Phi}(x,y)
      =\Phi(x)-\Phi({y})-\nabla\Phi({y})^\top ({x}-{y}).
\]
Then the algorithm is described as follows.

\begin{algorithm}
      \caption{Stochastic Mirror Descent}
      \label{alg.mirror}
      \KwIn{A convex function $f(x)$ defined over a convex set $\mathcal{X}$, a mirror function $\Phi(x)$, noise distributions $\Gamma_{\ell}$, step length $\eta$, success probability $\beta$.}
      \KwOut{$\hat{x}\in\mathcal{X}$ s.t. $f(\hat{x})$ approximates $\min_{x\in\mathcal{X}}f(x)$.}
\begin{algorithmic}[1]
      \STATE Set $L=\log_3(\frac{1}{\beta})$
      
      \FOR{$\ell=1,\dots,L$}
        \STATE Choose $x^{(1)}\in \argmin_{x\in\mathcal{X}}\Phi(x)$\label{alg_md_line3}

        \STATE Choose $\gamma_\ell$ from distribution $\Gamma_\ell$
        
        \FOR{$t=1,\dots,T$}
            \STATE Compute an unbiased estimator $g^{(t)}$ of $\nabla f(x^{(t)})$

            \STATE Update $g^{(t)}\gets g^{(t)}+\gamma_\ell$
        
            \STATE Choose $y^{(t+1)}$ s.t. $\nabla \Phi(y^{(t+1)})
            =\nabla \Phi(x^{(t)})-\eta g^{(t)}$\label{alg1:ytplus1}
        
            \STATE Let $x^{(t+1)}=\argmin_{x\in\mathcal{X}}D_\Phi(x,y^{(t+1)})$
        \ENDFOR
        \STATE Let $\hat{x}_\ell=\frac{1}{T}\sum_{t=1}^T x^{(t)}$\label{alg_md_line11}
      \ENDFOR
      \RETURN $\hat{x}=\argmin(f(\hat{x}_1),\dots,f(\hat{x}_L))$;
\end{algorithmic}
\end{algorithm}
Denote $R_{\ell,t}$ as the randomness while comp{u}ting $g^{(t)}$ in the outer iteration $\ell$. For any fixed outer iteration $\ell$, the difference between $\mathrm{E}_{\lbrace\gamma_\ell,R_{\ell,t}\rbrace}[f(\hat{x}_\ell)]$ and $\min_{x\in\mathcal{X}}f(x)$ is bounded by the following theorem. Here, $\mathrm{E}_{\lbrace\gamma_\ell,R_{\ell,t}\rbrace}$ represents the expectation taken over the randomness corresponding to random variables $\gamma_\ell,R_{\ell,t}$, where $\ell=1,\dots,L$ and $t=1,\dots,T$.

\begin{theorem}
    [Stochastic mirror descent~\cite{bubeck2015convex}]\label{thm.mirror}
Let $\Phi$ be a $\rho$-strongly convex map with respect to $\Vert\cdot\Vert$. 
    Given a convex function $f$ defined over convex set $\mathcal{X}$ with  
$x^*=\argmin_{x\in\mathcal{X}}f(x)$. 
Additionally, assume the noises satisfy $\E{\gamma_\ell}=0$. Consider \Cref{alg.mirror} for some fixed outer iteration $\ell\leq L$.

Assume that $\Phi(x^*)-\min_{x\in\mathcal{X}}\Phi(x)=R^2$, $\mathbb{E}_{\lbrace\gamma_\ell,R_{\ell,t}\rbrace}[g^{(t)}]=\nabla f(x^{(t)})$ and $\mathbb{E}_{\lbrace\gamma_\ell,R_{\ell,t}\rbrace}[\Vert g^{(t)}\Vert_*^2]\leq B^2$ for all $t\leq T$, where $\Vert\cdot\Vert_*$ denotes the norm dual to $\Vert\cdot\Vert$. After $T$ iterations with step length $\eta=\frac{R}{B}\sqrt{\frac{2}{T}}$, denote the output of the Line~\ref{alg_md_line3} to Line~\ref{alg_md_line11} in \Cref{alg.mirror} 
as $\hat{x}_{\ell}\in\mathcal{X}$, then $\mathbb{E}_{\lbrace\gamma_\ell,R_{\ell,t}\rbrace}[f(\hat{x}_{\ell})]\leq f(x^*)+RB\sqrt{\frac{2}{\rho T}}$. 
\end{theorem}

The above theorem is a generalization of the stochastic mirror descent algorithm (Theorem 6.1 in  \cite{bubeck2015convex}). Following the proof of  
Theorem 4.2 in \cite{bubeck2015convex}, we can have $\sum_{t=1}^T{g^{(t)}}^\top(x^{(t)}-x)\leq\frac{R^2}{\eta}+\frac{\eta}{2\rho}\sum_{t=1}^{T}\norm{g^{(t)}}_{*}^2$. Then since $f$ is a convex function, we have 
\begin{align*}
&\E[\lbrace\gamma_\ell,R_{\ell,t}\rbrace]{f(\frac{1}{T}\sum_{t=1}^Tx^{(t)})-f(x)}
\leq\frac{1}{T}\E[\lbrace\gamma_\ell,R_{\ell,t}\rbrace]{\sum_{t=1}^T(f(x^{(t)})-f(x))}\\
\leq&\frac{1}{T}\E[\lbrace\gamma_\ell,R_{\ell,t}\rbrace]{\sum_{t=1}^T\nabla f(x^{(t)})^\top(x^{(t)}-x)}
=\frac{1}{T}\E[\lbrace\gamma_\ell,R_{\ell,t}\rbrace]{\sum_{t=1}^T \E[\lbrace\gamma_\ell,R_{\ell,t}\rbrace]{g^{(t)}}^\top(x^{(t)}-x)}\\
=&\frac{1}{T}\E[\lbrace\gamma_\ell,R_{\ell,t}\rbrace]{\sum_{t=1}^T {g^{(t)}}^\top(x^{(t)}-x)}
\leq\frac{1}{T}\E[\lbrace\gamma_\ell,R_{\ell,t}\rbrace]{\frac{R^2}{\eta}+\frac{\eta}{2\rho}\sum_{t=1}^{T}\norm{g^{(t)}}_{*}^2}
\leq\frac{R^2}{\eta}+\frac{\eta B^2}{2\rho}.
\end{align*}
The theorem is proved when we choose $\eta=\frac{R}{B}\sqrt{\frac{2}{T}}$.

Note that by Markov inequality, for each $\ell\le L$, it holds that with probability at least $2/3$, $f(\hat{x}_{\ell})\leq 3f(x^*)+ 3RB\sqrt{\frac{2}{\rho T}}$. Therefore, with probability at least $1-(1/3)^{L}\geq 1-\beta$, at least one of $\hat{x}_{1},\dots,\hat{x}_{L}$, say $\hat{x}_{i_0}$, satisfies that $f(\hat{x}_{i_0})\leq 3f(x^*)+ 3RB\sqrt{\frac{2}{\rho T}}$. Thus, we have the following corollary. 

\begin{corollary}\label{col.mirror}
    Let $\Phi$ be a $\rho$-strongly convex map with respect to $\Vert\cdot\Vert$.  
    Given a convex function $f$ defined over convex set $\mathcal{X}$ with  
$x^*=\argmin_{x\in\mathcal{X}}f(x)$. Assume that $\Gamma_\ell$ has zero expectation for each $\ell$. Assume that $\Phi(x^*)-\min_{x\in\mathcal{X}}\Phi(x)=R^2$, $\mathbb{E}_{\lbrace\gamma_\ell,R_{\ell,t}\rbrace}[g^{(t)}]=\nabla f(x^{(t)})$ and $\mathbb{E}_{\lbrace\gamma_\ell,R_{\ell,t}\rbrace}[\Vert g^{(t)}\Vert_*^2]\leq B^2$ for all $t\leq T$ and $\ell\leq L$,  
where $\Vert\cdot\Vert_*$ denotes the norm dual to $\Vert\cdot\Vert$. 
\Cref{alg.mirror} with parameter $\eta=\frac{R}{B}\sqrt{\frac{2}{T}}$ will output $\hat{x}\in\mathcal{X}$ s.t. $f(\hat{x})\leq 3f(x^*)+3RB\sqrt{\frac{2}{\rho T}}$ with probability at least $1-\beta$.
\end{corollary}

However, 
when we need to compute the optimal solution of some convex optimization precisely, we can use KKT conditions \cite{ben2001lectures,ghojogh2021kkt} instead.

\begin{lemma}
    [KKT conditions, Slater’s condition~\cite{slater2013lagrange}] \label{l.kkt}
    For a constrained optimization problem:
    \begin{align*}
        &\mathrm{minimize}_{x\in\mathcal{X}} f(x)\\
        \mathrm{s.t.}\ \ &g_i(x)\leq 0, 
        \forall i\in\lbrace 1,\dots,m_1\rbrace\\
        &h_i(x)= 0, 
        \forall i\in\lbrace 1,\dots,m_2\rbrace
    \end{align*}

    Then we have:
\begin{enumerate}
    \item The optimal solution $\mathbf{x}$ must satisfy KKT conditions as follows:
\[
    \begin{cases}
        \nabla_{x} f(x)
        +\sum_i^n\lambda_i \nabla_{x}g_i(x)
        +\sum_i^n\mu_i \nabla_{x}h_i(x)=0\\
        \lambda_i g_i(x)=0, \forall i=1,\dots,m_1\\
        g_i(x)\leq 0, 
        \forall i= 1,\dots,m_1\\
        h_i(x)= 0, 
        \forall i= 1,\dots,m_2\\
        \lambda_i\geq 0,
        \forall i= 1,\dots,m_2
    \end{cases}
\]
\item We say the optimization satisfies Slater’s condition, if there exists an inner point $x$ of $\mathcal{X}$ 
satisfying that:
    \begin{align*}
        g_i(x)< 0, 
        \forall i\in\lbrace 1,\dots,m_1\rbrace\\
    h_i(x)= 0, 
        \forall i\in\lbrace 1,\dots,m_2\rbrace
    \end{align*}
    Namely, there is a inner point making the inequality constraints strictly feasible.
\item If $f(x)$ is a convex function over a convex set $\mathcal{X}$, and the optimization satisfies Slater’s condition, then KKT conditions are the necessary and sufficient condition for the optimal solution.
\end{enumerate}
\end{lemma}

\section{The Algorithm}
\label{s.alg}
For the sake of notation convenience, we let $\hat{G}=(V,\hat{E},\hat{\mathbf{w}}_e)$ denote the input graph, and let $G$ with edge vector $\mathbf{w}$ denote the output graph of our algorithm. Our algorithm outputs $G$ $(\varepsilon,\delta)$-differentially privately, and guarantees that each triangle-motif cut of $\hat{G}$ and $G$ will be close.

\paragraph{Preprocessing}
We use $W$ and $\mathbf{u}$ 
to denote the differentially privately released approximations of the sum of edge weights and the upper bound of each edge weight of $\hat{G}$ respectively. Then we reweigh the graph $\hat{G}$ to obtain $\overline{G}$ with adjacency matrix $\overline{\mathbf{A}}$, where the sum of edge weights is $W$. Specifically, we require that $W\geq\hat{W}$ and $\mathbf{u}\geq\overline{\mathbf{w}}$ with high probability. In the non-degenerate case when $\hat{W}$ and $\barLSKt$ are moderately large, we can guarantee that $W=\Theta(\hat{W})$ and $\tildeLSKt=\Theta(\hatLSKt)=\Theta(\barLSKt)$. In the following, we treat $\overline{G}$ as the input graph with public $W$ and $\mathbf{u}$. More details are deferred to \Cref{appendix:alg:preprocessing}.

Recall that $\mathbf{A}_{\triangle}$ is the adjacency matrix of the triangle-motif weighted graph of $G$. 
We let $\mathbf{D}^{(e)(t)}_{\triangle}$ denote the derivative of $\mathbf{A}_{\triangle}^{(t)}$ at $\mathbf{w}_e^{(t)}$, which is defined in Fact~\ref{def.D.E}. {For more details, refer to 
\Cref{pre.motif_cut}.}

\paragraph{An optimization problem} 
Our algorithm will be based on the following optimization problem, which in turns is based on \Cref{appendix:alg_analysis:sdp} that relates the triangle-motif cut difference and the following SDP form approximation. 
\begin{mdframed}
\begin{align}
    \label{eq.F_triangle}
    F_{\triangle}(\mathbf{w},\mathbf{X})&=\left(\begin{array}{cc}
        \mathbf{0} & \mathbf{A}_{\triangle}-\mathbf{\overline{A}}_{\triangle} \\
        \mathbf{A}_{\triangle}-\mathbf{\overline{A}}_{\triangle} & \mathbf{0}
    \end{array}\right)\bullet\mathbf{X}
    +\lambda \log \det\mathbf{X}\nonumber \\ 
    &+\sum_{(i,j)\in\binom{V}{2}}3(\mathbf{w}_{(i,j)}-\mathbf{\overline{w}}_{(i,j)})^2
    \cdot\sum_{s\in V\setminus \lbrace i,j\rbrace}(\mathbf{u}_{(i,s)}+\mathbf{u}_{(j,s)})
\end{align}
\begin{equation}
\label{eq.f_triangle}
    f_{\triangle}(\mathbf{w})=\max_{\mathbf{X}\in\mathcal{D}}F_{\triangle}(\mathbf{w},\mathbf{X}),
\end{equation}
where the domains of $\mathbf{w}$ and $\mathbf{X}$ are defined as follows:
\begin{equation}
    \mathcal{X}=\lbrace
\mathbf{w}\in\mathbb{R}_+^{\binom{V}{2}}
:\sum_{e\in\binom{V}{2}}\mathbf{w}_e=W, \mathbf{w}_e\leq \mathbf{u}_e\ \mathrm{for}\ \forall e\rbrace
\end{equation}
\begin{equation}
\label{eq.D}
    \mathcal{D}=\left\lbrace\mathbf{X}\in\mathbb{R}^{2n}:\mathbf{X}\ \mathrm{is}\ \mathrm{symmetric},
    \mathbf{X}\succeq \frac{1}{n}\mathbf{I}_{2n},
    \mathrm{and}\ \mathbf{X}_{ii}=1\ \mathrm{for}\ \forall{i}
    \right\rbrace
\end{equation}
\end{mdframed}

Note that, the independent variable of our function is $\mathbf{w}$, while $\mathbf{A}_{\triangle}$ is used in \Cref{eq.F_triangle}. However, $(\begin{array}{cc}
        \mathbf{0} & \mathbf{A}_{\triangle}-\mathbf{\overline{A}}_{\triangle} \\
        \mathbf{A}_{\triangle}-\mathbf{\overline{A}}_{\triangle} & \mathbf{0}
    \end{array})\bullet\mathbf{X}$
is neither private nor convex. Therefore, we use the regularizer $\lambda \log \det\mathbf{X}$ to control the stability and $\sum_{(i,j)\in\binom{V}{2}}3(\mathbf{w}_{(i,j)}-\mathbf{\overline{w}}_{(i,j)})^2
    \cdot\sum_{s\in V\setminus \lbrace i,j\rbrace}(\mathbf{u}_{(i,s)}+\mathbf{u}_{(j,s)})$ to control the convexity. 

\paragraph{Our algorithm} Now, we state our algorithm (\Cref{alg.triangle_cut_app}), which iteratively and privately solves the above optimization problem and outputs a graph $G$ that approximates each triangle-motif cut of $\hat{G}$. 
Specifically, our algorithm is an instantiation of the stochastic mirror descent algorithm (\Cref{alg.mirror}), that provides an iterative approach to solve an optimization problem defined by a convex function $f$. In each iteration, the algorithm computes an unbiased estimator $g$ of the gradient $\nabla f$, and then updates the solution based on $g$ and the mirror function $\Phi$.

We invoke \Cref{alg.mirror} with $f(x)=f_{\triangle}(\mathbf{w})$, $\Phi(x)=\sum_{e\in\binom{V}{2}}\mathbf{w}_e\log(\mathbf{w}_e)$, and mirror update step from \Cref{alg.update}. Additionally, we set $\norm{\cdot}$ as $l_1$ norm, hence $\norm{\cdot}_*$ is $l_{\infty}$ norm.

Denote $N(\mathbf{x},\mathbf{\Sigma})$ as the multivariate normal distribution with mean $\mathbf{x}$ and covariance matrix $\Sigma$. 
Additionally, $\tildeLSKt$, $U_{\triangle}$, $U_{\Lambda}$ are some quantities defined in \Cref{appendix:alg:preprocessing}.

\begin{algorithm}[h!]
      \caption{Private triangle-motif cut Approximation}
      \label{alg.triangle_cut_app}
      \KwIn{Graph $\overline{G}$ with adjacency matrix $\overline{\mathbf{A}}$, where $\overline{G}$ is the graph after preprocessing $\hat{G}$, step length $\eta$.}
      \KwOut{Privately release graph $G$, s.t. $G$ approximates each triangle-motif cut of $\hat{G}$.}
\begin{algorithmic}[1]
    \STATE Set $T=\Theta(\frac{W\cdot(\varepsilon U_{\triangle}+U_{\Lambda})}{n\log(n/\delta)\tildeLSKt})$,$L=\log_3(\frac{3}{\beta})$, $\lambda=\Theta(\varepsilon^{-1})\tildeLSKt\sqrt{T}\log^{\frac{3}{2}}(\frac{T}{\delta})\log(\frac{3}{\beta})$, $\varepsilon_1=\varepsilon_2=\varepsilon_3=\frac{\varepsilon}{6}$, $\varepsilon_4=\frac{\varepsilon}{6L}$
      
    \FOR{$\ell=1,\dots,L$}
    
        \STATE Choose $\mathbf{w}^{(1)}$ such that $\mathbf{w}_e^{(1)}=\frac{W}{\binom{n}{2}}$, for all $e\in\binom{V}{2}$\label{alg_cutapp_line3}
        
        \STATE Choose random variables $\nu_e\sim\mathrm{Lap}(\frac{1}{\varepsilon_4})$ and release $\tilde{\mathbf{w}}_e=\overline{\mathbf{w}}_e+\nu_e$
        \FOR{$t=1,\dots,T$
        }
            \STATE Find the maximizer $\mathbf{X}^{(t)}=\argmax_{\mathbf{X}\in\mathcal{D}}
            F_{\triangle}(\mathbf{w}^{(t)},\mathbf{X})$, where $F_{\triangle}$ is defined in \Cref{eq.F_triangle}\label{alg_cutapp_line6}
        
            \STATE Choose a random vector $\zeta \sim N(\mathbf{0},\mathbf{I}_{2n})$ and release $(\mathbf{X}^{(t)})^{\frac{1}{2}}\zeta$
        
            \STATE Compute the approximate gradient for all $e=(i,j)\in\binom{V}{2}$: 
            \begin{equation*}
            \mathbf{g}_e^{(t)}=\left(\mathbf{X}^{(t)})^{\frac{1}{2}}\zeta\zeta^\top(\mathbf{X}^{(t)})^{\frac{1}{2}}\right)
            \bullet \left(\begin{array}{cc}
          \mathbf{0} & \mathbf{D}_{\triangle}^{(e)(t)} \\
          \mathbf{D}_{\triangle}^{(e)(t)} & \mathbf{0}
            \end{array}\right)+6\sum_{s\in V\setminus \lbrace i,j\rbrace}(\mathbf{u}_{(i,s)}+\mathbf{u}_{(j,s)})(\mathbf{w}_e^{(t)}-\tilde{\mathbf{w}}_e)
            \end{equation*}
            
            \STATE Mirror Descent Step: $\mathbf{w}^{(t+1)}=\mathrm{MD\_Update}(\mathbf{w}^{(t)},\mathbf{g}^{(t)},W,\mathbf{u},\eta)$\label{alg_cutapp_line9}
        \ENDFOR
        \STATE Let $\mathbf{w}_{\ell}=\frac{1}{T}\sum_{i=1}^T \mathbf{w}^{(t)}$\label{alg_cutapp_line11}
    \ENDFOR
    \RETURN $\mathbf{w}=\argmin(f_{\triangle}(\mathbf{w}_{1}),\dots,f_{\triangle}(\mathbf{w}_{L}))$
\end{algorithmic}
\end{algorithm}

\paragraph{Mirror descent update} In Line~\ref{alg_cutapp_line9} of \Cref{alg.triangle_cut_app}, we use the algorithm $\mathrm{MD\_Update}$ to update our solution. We now describe this mirror descent update step. Given a weight vector $\mathbf{w}^{(t)}$ for some $t\geq 1$, the estimated gradient $\mathbf{g}^{(t)}$, and $W,\mathbf{u},\eta$, we update $\mathbf{w}^{(t)}$ to $\mathbf{w}^{(t+1)}$ using \Cref{alg.update}.

\begin{algorithm}[h!]
      \caption{MD\_Update}
      \label{alg.update}
      \KwIn{$\mathbf{w}^{(t)},\mathbf{g}^{(t)},W,\mathbf{u},\eta$.}
      \KwOut{$\mathbf{w}^{(t+1)}$.}
\begin{algorithmic}[1]
      \STATE Set $\mathbf{y}^{(t+1)}_e=\mathbf{w}^{(t)}_e\exp(-\eta\mathbf{g}^{(t)}_e)$ for $\forall e\in\binom{V}{2}$  and let $N=\binom{n}{2}$\label{alg_upt_line1}
      
      \STATE Sort edges in non-increasing order so that $\frac{\mathbf{y}^{(t+1)}_{e_1}}{\mathbf{u}_{e_1}}\geq \frac{\mathbf{y}^{(t+1)}_{e_2}}{\mathbf{u}_{e_2}}\geq \dots\geq \frac{\mathbf{y}^{(t+1)}_{e_{N}}}{\mathbf{u}_{e_{N}}}$
      
      \STATE Compute $S_i=\sum_{j=i}^{N}\mathbf{y}^{(t+1)}_{e_j}$ for $i=1,\dots,N$
      
      \STATE Let $W_1=W$
      
      \FOR{$i=1,\dots,N$}
      \STATE $\mathbf{w}_{e_i}^{(t+1)}=\min(\frac{W_i\cdot\mathbf{y}^{(t+1)}_{e_i}}{S_i},\mathbf{u}_{e_i})$
      
      \STATE Set $W_{i+1}=W_i-\mathbf{w}_{e_i}^{(t+1)}$
      \ENDFOR 
      \RETURN $\mathbf{w}^{(t+1)}$
\end{algorithmic}
\end{algorithm}

The above algorithm corresponds to the stochastic mirror descent framework from \Cref{alg.mirror} by setting $f(x)=f_{\triangle}(\mathbf{w})$ which is defined in \Cref{eq.f_triangle}, and $\Phi(x)=\Phi(\mathbf{w})=\sum_{e\in\binom{V}{2}}\mathbf{w}_e\log(\mathbf{w}_e)$.

{As a point of comparison}, $\sum_{e}\mathbf{w}_e=W$ is the only constraint $\mathbf{w}$ needs to satisfy in the method proposed by \cite{eliavs2020differentially}, hence the update step is simply $\mathbf{w}^{(t+1)}_e=\frac{W\mathbf{y}_e^{(t)}}{\sum_e\mathbf{y}_e^{(t)}},$ for $\forall e\in\binom{V}{2}$. However, in our setting, since we add constraints $\mathbf{w}_e\leq\mathbf{u}_e$ for $\forall e\in\binom{V}{2}$, we have to use a more complicated method, as described above, to update weights.

As a point of comparison, $\sum_{e}\mathbf{w}_e=W$ is the only constraint $\mathbf{w}$ needs to satisfy in the method proposed by \cite{eliavs2020differentially}, hence the update step is simply $\mathbf{w}^{(t+1)}_e=\frac{W\mathbf{y}_e^{(t)}}{\sum_e\mathbf{y}_e^{(t)}},$ for $\forall e\in\binom{V}{2}$. However, in our setting, since we add constraints $\mathbf{w}_e\leq\mathbf{u}_e$ for $\forall e\in\binom{V}{2}$, we have to use a more complicated method, as described above, to update weights.

\section{Analysis of the Algorithm}
\label{s.alg_analysis}
\label{appendix:alg_analysis}
In this section, we provide the analysis of our algorithm (\Cref{alg.triangle_cut_app}) and prove the correctness of \Cref{alg.triangle_cut_app}. In \Cref{appendix:alg_analysis:sdp}, we derive the SDP approximation of the triangle-motif cut difference. In \Cref{appendix:alg_analysis:grad_convex}, we derive the gradient of $f_{\triangle}$ and prove its convexity. Then in \Cref{s.mirror_d_upt}, we prove the correctness of \Cref{alg.update} for mirror descent update step. Then we give privacy analysis in \Cref{appendix:alg_analysis:privacy} and determines the value of $\lambda$ to guarantee the $(\varepsilon,\delta)$-DP property.  
\Cref{appendix:alg_analysis:precision} contains the utility analysis, in which we show the difference between $f_{\triangle}$ and cut norm, and use Corollary~\ref{col.mirror} to determine the value of $T$ and bound the additive error. Finally, in \Cref{appendix:alg_analysis:time}, we analyze the running time of the algorithm. 

\subsection{The SDP Approximation}
\label{appendix:alg_analysis:sdp}
Focusing on the special case when the motif is triangle, the motif adjacency matrix has a very direct relationship with motif size of cut in this case. Denote $\mathbf{\mathbf{1}_S}$ as the indicative vector of the vertex set $S$, which has value one only in the coordinates corresponding to points in $S$. The following fact is a straightforward generalization of a result (for unweighted graphs) from \cite{benson2016higher}. 
\begin{fact}
\label{fact.triangle_cut_motifmatrix}
It holds that
\begin{equation*}
    \mathrm{Cut}_{\triangle}^{(G)}(S,V\setminus S)
    =\frac{1}{2}\mathbf{\mathbf{1}_S}^\top\mathbf{A}_\triangle\mathbf{1}_{{V\setminus S}}
\end{equation*}
\end{fact}
\begin{proof}
We note that for any cut of $G$, a triangle is either crossed zero times or twice. Thus
\begin{align*}
    \mathrm{Cut}_{\triangle}^{(G)}(S,V\setminus S)
&=\sum_{I\in \mathcal{M}(G,\triangle)
:I\ \mathrm{crosses}\ (S,V\setminus S)}w(I)\\
&=\frac{1}{2}\sum_{(i,j)\in\binom{V}{2}
:(i,j)\ \mathrm{crosses}\ (S,V\setminus S)}
\sum_{k:(i,j,k)\in\mathcal{M}(G,\triangle)}w((i,j,k))\\
&=\frac{1}{2}\sum_{(i,j)\in\binom{V}{2}
:(i,j)\ \mathrm{crosses}\ (S,V\setminus S)}(\mathbf{A}_{\triangle})_{i,j}\\
&=\frac{1}{2}\mathbf{\mathbf{1}_S}^\top\mathbf{A}_\triangle\mathbf{1}_{{V\setminus S}}
\end{align*}   
\end{proof}

This also implies that the triangle-motif cut of $G$ is exactly half of the cut of the triangle-motif weighted graph of $G$. We note that this property also applies to any other $3$-vertex motif (such as the length-$2$ path), allowing for easy generalization of our algorithm. However, for motifs with more than $3$ vertices, there lacks a similar characterization (see e.g. \cite{benson2016higher}). 

Inspired by the setting of edge cut \cite{eliavs2020differentially} described in \Cref{pre.cut_norm}, we have the following method to bound the additive error for triangle-motif cut. Consider graphs $G_1$ and $G_2$. Let $(\mathbf{A}_{1})_\triangle$ and $(\mathbf{A}_2)_\triangle$ be their triangle adjacency matrices. Then for a fixed cut $(S,V\setminus S)$, the $\triangle$-motif size 
of $G_1$ is $\sum_{v\in S,u\in V\setminus S}\left((\mathbf{A}_1)_\triangle\right)_{u,v}$ and the $\triangle$-motif size of $G_2$ is $\sum_{v\in S,u\in V\setminus S}\left((\mathbf{A}_2)_\triangle\right)_{u,v}$. So we can see that the triangle-motif cut difference between $G_1$ and $G_2$ is twice of the following expression: 
\begin{equation*}
    \max_{S\subset V}
    \left\lbrace\left|
    \sum_{v\in S,u\in V\setminus S}\left((\mathbf{A}_1)_\triangle\right)_{u,v}
    -\sum_{v\in S,u\in V\setminus S}\left((\mathbf{A}_2)_\triangle\right)_{u,v}
    \right|\right\rbrace
    =\max_{\mathbf{x},\mathbf{y}\in \{0,1\}^n,
    \mathbf{x}+\mathbf{y}=\mathbf{1}}\left\lbrace
    \left\vert\mathbf{x}^\top\left(
    (\mathbf{A}_1)_\triangle-(\mathbf{A}_2)_\triangle
    \right)\mathbf{y}\right\vert
    \right\rbrace
\end{equation*}

According to \cite{frieze1999quick}, the cut on a graph can be bounded by the cut norm which is defined in \Cref{pre.cut_norm}. Thus by an approximation in \cite{alon2004approximating} and Lemma~\ref{lemma:diff_primal_sdp}, we have the following lemma:

\begin{lemma}
\label{l.triangle_diff_primal_sdp}
    The triangle-motif cut difference between $G_1$ and $G_2$ can be bounded by the following SDP up to a constant factor:

    \begin{equation}\label{eq.primal_sdp_triangle}
    \max\left\lbrace
    \left(\begin{array}{cc}
        \mathbf{0} & (\mathbf{A}_2)_\triangle-(\mathbf{A}_1)_\triangle \\
        (\mathbf{A}_2)_\triangle-(\mathbf{A}_1)_\triangle & \mathbf{0}
    \end{array}\right)
    \bullet \mathbf{X}
    :\mathbf{X}\ \mathrm{is}\ \mathrm{symmetric},
    \mathbf{X}\succeq \mathbf{0},
    \mathrm{and}\ \mathbf{X}_{i,i}=1\ \mathrm{for}\ \forall{i}
    \right\rbrace
    \end{equation}
\end{lemma}

To sum up, the triangle-motif cut difference between $G_1$ and $G_2$ 
can be bounded by \Cref{eq.primal_sdp_triangle} up to a constant factor.

\subsection{Gradient and Convexity}
\label{appendix:alg_analysis:grad_convex}
Recall $\mathcal{D}$ is the domain of $\mathbf{X}$, which is defined in \Cref{eq.D}. We first state some of its useful properties.

\begin{lemma}
    \label{l.xi_eigen}
It holds that
\begin{enumerate}
    \item For any $\mathbf{X}\in \mathcal{D}$, we have $\mathbf{X}_{ij}\in[-1,1]$ for any $i$ and $j$.
    \item For any $\mathbf{X}\in \mathcal{D}$ with eigenvalues $\lambda_1,\dots,\lambda_{2n}$, we have $\lambda_i\in[\frac{1}{n},2n]$ for any $i$.
\end{enumerate}
\end{lemma}

\begin{fact}
    \label{fact.semi_definite_prop}
    If $\mathbf{A}\in\mathbb{R}^{n\times n}$ is a semi-definite matrix, then there exists $\mathbf{x}_i\in\mathbb{R}^{n}, i=1,\dots,n$ s.t.
    $\mathbf{A}=\sum_{i=1}^n\mathbf{x}_i^\top\mathbf{x}$.
\end{fact}

\begin{proof}[Proof of Lemma~\ref{l.xi_eigen}]
    Since $\mathbf{X}\in \mathcal{D}$ is a semi-definite matrix, by Fact~\ref{fact.semi_definite_prop}, there exist vectors $\mathbf{x}_1,\dots,\mathbf{x}_{2n}\in\mathbb{R}^{2n}$ such that $\mathbf{X}_{i,j}=\mathbf{x}_i^\top\mathbf{x}_j$. Then by $\mathbf{X}_{ii}=1$ and $\mathbf{X}_{i,i}=\mathbf{x}_i^\top\mathbf{x}_i$, we have $\vert\mathbf{x}_i\vert=1$ for any $i\leq n$. Therefore, $\vert \mathbf{X}_{i,j}\vert\leq\sqrt{\vert\mathbf{x}_i\vert\vert\mathbf{x}_j\vert}\leq1$.

    Then by $\mathbf{X}_{i,i}=1$ for $i=1,\dots,n$, we have $\sum_{i=1}^{2n}\lambda_i=\mathrm{tr}(\mathbf{X})=\sum_{i=1}^{2n}\mathbf{X}_{i,i}=2n$. Moreover, since $\mathbf{X}\succeq\frac{1}{n}\mathbf{I}_{2n}$, it holds that $(\mathbf{X}-\frac{1}{n}\mathbf{I}_{2n})\succeq \mathbf{0}$. Thus we have $\lambda_i\geq\frac{1}{n}$. Therefore, we can conclude that $\lambda_i\in[\frac{1}{n},2n]$ for any $i$.
\end{proof}

We next prove the convexity of $f_{\triangle}(\mathbf{w})$ and compute its gradient.

\begin{lemma}
    \label{l.grad_conv}
The function $f_{\triangle}(\mathbf{w})$ is convex and differentiable with respect to $\mathbf{w}$. Furthermore, for any fixed pair $e\in \binom{V}{2}$, it holds that
\[
\nabla f_{\triangle}(\mathbf{w})_{e}
=\nabla_{\mathbf{w}}F_{\triangle}(\mathbf{w},\mathbf{X}^*)_{e}
=\left(\begin{array}{cc}
          \mathbf{0} & \mathbf{D}_{\triangle}^{(e)} \\
          \mathbf{D}_{\triangle}^{(e)} & \mathbf{0}
\end{array}\right)\bullet \mathbf{X}^*
+6\sum_{s\in V\setminus \lbrace i,j\rbrace:e=(i,j)}(\mathbf{u}_{(i,s)}+\mathbf{u}_{(j,s)})(\mathbf{w}_e-\overline{\mathbf{w}}_e),
\]
where $\mathbf{X}^*$ denotes the maximizer such that $F_{\triangle}(\mathbf{w},\mathbf{X}^*)
=\max_{\mathbf{X}\in\mathcal{D}}F_{\triangle}(\mathbf{w},\mathbf{X})$, for some fixed $\mathbf{w}$.
\end{lemma}

To prove the above lemma, we need the following theorem. 

\begin{theorem}
    [Danskin's theorem~\cite{danskin2012theory}]
    \label{thm.danskin}
    Let $\mathcal{D}\in\mathbb{R}^m$ be a compact subset and $\phi:\mathbb{R}^n\times\mathcal{D}\rightarrow\mathbb{R}$ be a continuous function such that $\phi(\cdot,x)$ is convex for fixed $x\in\mathcal{D}$. Then the function $f:\mathbb{R}^n\rightarrow\mathbb{R}$ defined as $f(z)=\max_{x\in\mathcal{D}}\phi(z,x)$ is convex.
    If there is a unique maximizer $x^*$ such that $\phi(z,x^*)=\max_{x\in\mathcal{D}}\phi(z,x)$ and that $\phi(z,x^*)$ is differentiable at $z$, then $f$ is differentiable at $z$ and $\nabla f(z)=\nabla_z\phi(z,x^*)=\left(\frac{\partial\phi(z,x^*)}{{\partial z}_i}\right)_{i=1}^n$.
\end{theorem}
We will also need the following useful property of a semi-definite matrix.
\begin{fact}
    \label{l.semi_definite}
    A symmetric matrix $\mathbf{A}\in\mathbb{R}^{n\times n}$ with non-negative diagonal entries 
    is semi-definite if $\mathbf{A}$ is diagonally dominant, i.e., 
    $$\mathbf{A}_{i,i}\geq\sum_{j\ne i}\vert\mathbf{A}_{i,j}\vert,\quad\text{for any $1\leq i\leq n$}.$$
\end{fact}
Now we prove Lemma~\ref{l.grad_conv}. 
\begin{proof}[Proof of Lemma~\ref{l.grad_conv}]
Recall in Fact~\ref{def.D.E}, we denote $\mathbf{D}_{\triangle}^{(k,\ell)}$ as the derivative of $\mathbf{A}_{\triangle}$ at $\mathbf{w}_{(k,\ell)}$ and $\mathbf{E}_{\triangle}^{((i,j),(k,\ell))}$ as the second-order derivative of $\mathbf{A}_{\triangle}$ at $\mathbf{w}_{(i,j)}$ and $\mathbf{w}_{(k,\ell)}$. Therefore, for any fixed $\mathbf{X}\in\mathcal{D}$, we have that
\[
\nabla_{\mathbf{w}} F_{\triangle}(\mathbf{w},\mathbf{X})_e=\left(\begin{array}{cc}
          \mathbf{0} & \mathbf{D}_{\triangle}^{(e)} \\
          \mathbf{D}_{\triangle}^{(e)} & \mathbf{0}
\end{array}\right)\bullet \mathbf{X}
+6\sum_{s\in V\setminus \lbrace i,j\rbrace:e=(i,j)}(\mathbf{u}_{(i,s)}+\mathbf{u}_{(j,s)})(\mathbf{w}_e-\overline{\mathbf{w}}_e),
\]
and the second-order partial derivatives are 
\begin{align*}
\nabla_{\mathbf{w}}^2F_{\triangle}(\mathbf{w},\mathbf{X})_{(i,j),(k,\ell)}
&=\left(\begin{array}{cc}
          \mathbf{0} & \mathbf{E}_{\triangle}^{((i,j),(k,\ell))} \\
          \mathbf{E}_{\triangle}^{((i,j),(k,\ell))} & \mathbf{0}
\end{array}\right)\bullet \mathbf{X}
+6\sum_{s\in V\setminus \lbrace i,j\rbrace}(\mathbf{u}_{(i,s)}+\mathbf{u}_{(j,s)})\cdot\mathbf{1}_{(i,j)=(k,\ell)}\\
&=\begin{cases}
    6\sum_{s\in V\setminus \lbrace i,j\rbrace}(\mathbf{u}_{(i,s)}+\mathbf{u}_{(j,s)}), &(i,j)=(k,\ell)\\
    \mathbf{w}_{(j,\ell)}(\sum_{i',j'\in\{i,j,\ell\}}\mathbf{X}_{i',n+j'}
    +\sum_{i',j'\in\{i,j,\ell\}}\mathbf{X}_{n+i',j'}), &i=k\ \mathrm{and}\ j\ne \ell\\
    0, &o.w.
\end{cases}
\end{align*}

Then for any $(i,j)\in\binom{V}{2}$, it holds that,
\begin{align*}
    \sum_{(k,\ell)\ne(i,j)}\vert\nabla_{\mathbf{w}}^2F_{\triangle}(\mathbf{w},\mathbf{X})_{(i,j),(k,\ell)}\vert
    &=\sum_{s\in V\setminus \lbrace i,j\rbrace}\left(\vert\nabla_{\mathbf{w}}^2F_{\triangle}(\mathbf{w},\mathbf{X})_{(i,j),(i,s)}\vert
    +\vert\nabla_{\mathbf{w}}^2F_{\triangle}(\mathbf{w},\mathbf{X})_{(i,j),(j,s)}\vert\right)\\
    &=\sum_{s\in V\setminus \lbrace i,j\rbrace}\left(\mathbf{w}_{(i,s)}+\mathbf{w}_{(j,s)}\right)\vert\sum_{i',j'\in\{i,j,s\}}\mathbf{X}_{i',n+j'}
    +\sum_{i',j'\in\{i,j,s\}}\mathbf{X}_{n+i',j'}\vert\\
    &\leq 6\sum_{s\in V\setminus \lbrace i,j\rbrace}\left(\mathbf{u}_{(i,s)}+\mathbf{u}_{(j,s)}\right)=\nabla_{\mathbf{w}}^2F_{\triangle}(\mathbf{w},\mathbf{X})_{(i,j),(i,j)},
\end{align*}
where the second to last inequality follows from the fact that $\mathbf{w}_e\leq \mathbf{u}_e$ 
and that $\vert\mathbf{X}_{i,j}\vert\leq 1$. 

Thus, by Fact~\ref{l.semi_definite}, for fixed $\mathbf{X}\in\mathcal{D}$, $\nabla_{\mathbf{w}}^2F_{\triangle}(\mathbf{w},\mathbf{X})$ is semi-definite, so $F_{\triangle}(\mathbf{w},\mathbf{X})$ is convex. By \Cref{thm.danskin}, $f_{\triangle}(\mathbf{w})=\max_{\mathbf{X}\in\mathcal{D}}F_{\triangle}(\mathbf{w},\mathbf{X})$ is convex with respect to $\mathbf{w}$.

Moreover, for any fixed $\mathbf{X}\in\mathcal{D}$, it can be verified that $F_{\triangle}(\mathbf{w},\mathbf{X})$ is differentiable at $\mathbf{w}$. Therefore, by \Cref{thm.danskin}, $f_{\triangle}(\mathbf{w})$ is differentiable. Furthermore, 
\[
\nabla f_{\triangle}(\mathbf{w})_{e}
=\nabla_{\mathbf{w}}F_{\triangle}(\mathbf{w},\mathbf{X}^*)_{e}
=\left(\begin{array}{cc}
          \mathbf{0} & \mathbf{D}_{\triangle}^{(e)(t)} \\
          \mathbf{D}_{\triangle}^{(e)(t)} & \mathbf{0}
\end{array}\right)\bullet \mathbf{X}^*
+6\sum_{s\in V\setminus \lbrace i,j\rbrace:e=(i,j)}(\mathbf{u}_{(i,s)}+\mathbf{u}_{(j,s)})(\mathbf{w}_e^{(t)}-\overline{\mathbf{w}}_e),
\]
where $\mathbf{X}^*$ denotes the maximizer such that $F_{\triangle}(\mathbf{w},\mathbf{X}^*)
=\max_{\mathbf{X}\in\mathcal{D}}F_{\triangle}(\mathbf{w},\mathbf{X})$.
\end{proof}

\subsection{Mirror Descent Update}
\label{s.mirror_d_upt}
We now show that the update method \Cref{alg.update} correctly implements the mirror descent framework in our setting. To do so, we note that the mirror function in our algorithm is $\Phi(\mathbf{w})=\sum_{e\in\binom{V}{2}}\mathbf{w}_e\log(\mathbf{w}_e)$, and thus it suffices to prove that the output of \Cref{alg.update} is indeed the minimizer of $D_\Phi(\mathbf{w},\mathbf{y}^{(t+1)})$, where $\mathbf{y}^{(t+1)}$ is the vector satisfying the Line~\ref{alg_upt_line1} in \Cref{alg.update}, i.e., $\mathbf{y}^{(t+1)}_e=\mathbf{w}^{(t)}_e\exp(-\eta\mathbf{g}^{(t)}_e)$ for $\forall e\in\binom{V}{2}$.

\begin{theorem}
    \label{thm.upt_correct}
Denote the output of \Cref{alg.update} as $\mathbf{w}^{(t+1)}$. Then it holds that $D_{\Phi}(\mathbf{w}^{(t+1)},\mathbf{y}^{(t+1)})
=\min_{\mathbf{w}\in\mathcal{X}} D_{\Phi}(\mathbf{w},\mathbf{y}^{(t+1)})$, where  
\[
D_{\Phi}(\mathbf{w},\mathbf{y}^{(t+1)})=\Phi(\mathbf{w})-\Phi(\mathbf{y}^{(t+1)})-(\nabla\Phi(\mathbf{y}^{(t+1)}))^\top(\mathbf{w}-\mathbf{y}^{(t+1)}), 
\]
and 
$\mathbf{y}^{(t+1)}\in\mathbb{R}^{\binom{V}{2}}$ s.t. $\nabla \Phi(\mathbf{y}^{(t+1)})
=\nabla \Phi(\mathbf{w}^{(t)})-\eta \mathbf{g}^{(t)}$.
\end{theorem}

Note that, $\mathbf{y}^{(t+1)}$ intuitively is the update towards the direction of gradient $\mathbf{g}_e^{(t)}$ without any domain constraints, and $D_{\Phi}(\mathbf{w},\mathbf{y}^{(t+1)})$ represents the difference between $\mathbf{w}$ and $\mathbf{y}^{(t+1)}$ to some extent. Namely, we are trying to find the most similar solution $\mathbf{w}^{(t+1)}$ in $\mathcal{X}$ to $\mathbf{y}$. 

\begin{proof}[Proof of \Cref{thm.upt_correct}]
We first show that $\min_{\mathbf{w}\in\mathcal{X}} D_{\Phi}(\mathbf{w},\mathbf{y}^{(t+1)})$ is a convex optimization, which satisfies the Slater's condition (Lemma~\ref{l.kkt}). This will imply the KKT conditions are the necessary and sufficient condition for its optimal solution (by Lemma~\ref{l.kkt}).

Since $\Phi(\mathbf{w})=\sum_{e\in\binom{V}{2}}\mathbf{w}_e\log(\mathbf{w}_e)$, we have 
    $\nabla \Phi({\mathbf{y}}^{(t+1)})
      =\nabla \Phi(\mathbf{w}^{(t)})-\eta \mathbf{g}^{(t)}$, 
that is 
$1+\log(\mathbf{y}_e^{(t+1)})=1+\log(\mathbf{w}_e^{(t)})-\eta\mathbf{g}^{(t)}_e.$ 
      Therefore, we have $\mathbf{y}_e^{(t+1)}=\mathbf{w}_e^{(t)}\exp(-\eta\mathbf{g}_e^{(t)})$ for any $e\in\binom{V}{2}$. 
      Then, note that 
\begin{align*}
D_{\Phi}(\mathbf{w},\mathbf{y}^{(t+1)})
      &=\Phi(\mathbf{w})-\Phi(\mathbf{y}^{(t+1)})-\nabla\Phi(\mathbf{y}^{(t+1)})^\top(\mathbf{w}-\mathbf{y}^{(t+1)})\\
      &=\sum_{e\in\binom{V}{2}}(\mathbf{w}_e\log(\mathbf{w}_e)-\mathbf{w}_e(1+\log(\mathbf{y}_e^{(t+1)})))+R(\mathbf{y}^{(t+1)}).     
\end{align*}where $R(\mathbf{y}^{(t+1)})$ is some remaining term only depends on $\mathbf{y}^{(t+1)}$. Thus, we only need to consider the following optimization problem. 

    \begin{align}
        \min_{\mathbf{w}} D_{\Phi}(\mathbf{w},\mathbf{y}_e^{(t+1)})&=\sum_{e\in\binom{V}{2}}\left(\mathbf{w}_e\log(\mathbf{w}_e)-\mathbf{w}_e\left(1+\log\left(\mathbf{y}^{(t+1)}_e\right)\right)\right)+R(\mathbf{y}^{(t+1)}) \label{optDphiwy}\\
        \mathrm{s.t.}\sum_{e\in\binom{V}{2}}\mathbf{w}_e&=W
        \ \ \ \ \mathrm{and}\ \ \ \ 
        \mathbf{w}_e\leq \mathbf{u}_e, 
        \forall e
    \end{align}

    Recall $\mathcal{X}=\lbrace
\mathbf{w}\in\mathbb{R}_+^{\binom{V}{2}}
:\sum_{e\in\binom{V}{2}}\mathbf{w}_e=W, \mathbf{w}_e\leq \mathbf{u}_e\ \mathrm{for}\ \forall e\rbrace$. Because $\mathbf{u}_e=\overline{\mathbf{w}}_e+\mathrm{Lap}(3/\varepsilon)+3\log(3n^2/\beta)/\varepsilon+\frac{W}{\binom{n}{2}}>\frac{W}{\binom{n}{2}}$ for any $e\in\binom{V}{2}$ with high probability, $\mathbf{w}$ is an inner point of $\mathcal{X}$ 
when $\mathbf{w}_e=\frac{W}{\binom{n}{2}}, \forall e\in\binom{V}{2}$, i.e., making the inequality constraints in $\mathcal{X}$ strictly feasible. It implies that the optimization satisfies Slater's condition. 
    Besides, it can be verified that $D_{\Phi}(\mathbf{w},\mathbf{y}^{(t+1)})$ 
    is a convex function with respect to $w$
    . Thus by Lemma~\ref{l.kkt}, the following conditions are the necessary and sufficient condition for the optimal solution:

    \begin{numcases}{}\label{eq1}
        \nabla_{\mathbf{w}} D_{\Phi}(\mathbf{w},\mathbf{y}^{(t+1)})
        +\sum_{e\in\binom{V}{2}}\lambda_e \nabla_{\mathbf{w}}(\mathbf{w}_e-\mathbf{u}_e)
        +\mu \nabla_{\mathbf{w}}(\sum_{e\in\binom{V}{2}}\mathbf{w}_e-W)=0\\\label{eq2}
        \lambda_e (\mathbf{w}_e-\mathbf{u}_e)=0, \forall e \\\label{eq3}
        \sum_{e\in\binom{V}{2}}\mathbf{w}_e=W\\\label{eq4}
        \mathbf{w}_e\leq \mathbf{u}_e, 
        \forall e\\\label{eq5}
        \lambda_e\geq 0, \forall e
    \end{numcases}
    where Equation (\ref{eq1}) is equivalent to 
        $1+\log(\mathbf{w}_e)-1-\log(\mathbf{y}_e^{(t+1)})+\lambda_e+\mu=0$. 
    Therefore we have, 
    \begin{equation}
\label{eq6}\mathbf{w}_e=\mathbf{y}_e^{(t+1)}\exp(-\lambda_e-\mu)
    \end{equation}
    
Denote $S=\sum_{e\in\binom{V}{2}}\mathbf{y}_e^{(t+1)}$. Assume $\mathbf{w}^*\in\argmin_{\mathbf{w}\in\mathcal{X}}\min_{\mathbf{w}\in\mathcal{X}} D_{\Phi}(\mathbf{w},\mathbf{y}^{(t+1)})$ is a solution with $k$ coordinates $\mathcal{I}=\lbrace i_1,\dots,i_k\rbrace$ satisfying the equality condition of Equation (\ref{eq4}), i.e., $\mathbf{w}_{e_i}=\mathbf{u}_{e_i}$, for any $i\in\mathcal{I}$. We will consider two cases: $k=0$ and $k\neq 0$. 

\textbf{Case 1:} Now we first prove the case that $k=0$. 

When $k=0$, we have $\lambda_e=0$ for $\forall e\in\binom{V}{2}$. Thus, by \Cref{eq6}, 
    \begin{align*}
\mathbf{w}_e^{*}&=\mathbf{y}_e^{(t+1)}\exp(-\mu), \text{ for $\forall e\in\binom{V}{2}$},\\
    W&=\sum_{e\in\binom{V}{2}}\mathbf{w}_e^*=\sum_{e\in\binom{V}{2}}\mathbf{y}_e^{(t+1)}\exp(-\mu)=S\cdot \exp(-\mu),
    \end{align*}
    which further imply that $\exp(-\mu)=\frac{W}{S}$. Thus,
    the optimal solution is 
    \[
\mathbf{w}_e^*=\frac{W\cdot \mathbf{y}_e^{(t+1)}}{S}=\exp(-\mu)\cdot \mathbf{y}_e^{(t+1)}, \text{ for $\forall e\in\binom{V}{2}$}.
    \]
Furthermore, by the above, for any subset $F\subseteq \binom{V}{2}$ of edges, it holds that 
\[\frac{\sum_{e\in F}\mathbf{w}_{e}^{*}}{\sum_{e\in F}\mathbf{y}_{e}^{(t+1)}}=\exp(-\mu).
\]

Since there is no coordinate satisfying the equality condition of Equation (\ref{eq4}), we have $\mathbf{w}^*_e<\mathbf{u}_e$ {for any $e$}. Let $e_1,e_2,\dots,e_{N}$ be the ordered sequence of edges given in \Cref{alg.update}. 
    Note that by the above discussion, $W_1=W$, $S_1=S$, and 
    \[
\mathbf{w}^{(t+1)}_{e_1}=\frac{W_1\cdot \mathbf{y}_{e_{1}}^{(t+1)}}{S_1}=\frac{W\cdot \mathbf{y}_{e_{1}}^{(t+1)}}{S}=\exp(-\mu)\cdot \mathbf{y}_{e_{1}}^{(t+1)}=\mathbf{w}_{e_1}^*.
    \] 
Then by induction we have that $W_{i+1}=W_1-\sum_{\ell=1}^{i}\mathbf{w}_{e_\ell}^{(t+1)}$ and $S_{i+1}=S_1-\sum_{\ell=1}^{i}\mathbf{y}_{e_\ell}^{(t+1)}$, which further implies that $W_{i+1}=S_{i+1}=\exp(-\mu)$. Thus, $\mathbf{w}_{e_{i+1}}^{(t+1)}=\frac{W_{i+1}\cdot \mathbf{y}_{e_{{i}}}^{(t+1)}}{S_{i+1}}=\exp(-\mu)\cdot \mathbf{y}_{e_{{i}}}^{(t+1)}=\mathbf{w}_{e_{i+1}}^*$. Therefore, the output solution $\mathbf{w}^{(t+1)}$ of \Cref{alg.update} is indeed the optimum solution $\mathbf{w}^*$.

\textbf{Case 2:} In the following we consider the case that $k\neq 0$. 
In this case, let $\mathcal{I}=\lbrace i_1,\dots,i_k\rbrace$ denote the index set of the $k$ coordinates such that 
    $\mathbf{w}^*_{e_i}=\mathbf{u}_{e_i}$, for any $i\in\mathcal{I}$, 
    and let $\mathcal{J}=\lbrace j_1,\dots,j_{\binom{n}{2}-k}\rbrace$ deonte the set of remaining indices. 
    
    Then by $\mathbf{w}^*_{e_j}\ne \mathbf{u}_{e_j}$ for any $j\in\mathcal{J}$ and Equation (\ref{eq2}), $\lambda_j=0$ for any $j\in\mathcal{J}$. Thus,
$\mathbf{w}_{e_j}^*=\mathbf{y}_{e_j}^{(t+1)}\exp(-\mu), \text{for any $j\in\mathcal{J}$}$. 
Note that $W=\sum_{\ell=1}^N\mathbf{w}_{e_\ell}^*=\sum_{i\in \mathcal{I}}\mathbf{u}_{e_i}+\sum_{j\in \mathcal{J}}\mathbf{y}_{e_j}^{(t+1)}\exp(-\mu)$.

    Then by Equation (\ref{eq3}), Equation (\ref{eq4}) and Equation (\ref{eq6}), we have 
    $\mathbf{y}_{e_j}^{(t+1)}e^{-\mu}<\mathbf{u}_{e_j}$, for any $j\in\mathcal{J}$, and $e^{-\mu}=\frac{W-\sum_{i\in\mathcal{I}}\mathbf{u}_{e_i}}{S-\sum_{i\in\mathcal{I}}\mathbf{y}_{e_i}^{(t+1)}}$. 
    Then by Equation (\ref{eq5}) and Equation (\ref{eq6}), we have $\mathbf{u}_{e_i}=\mathbf{w}^*_{e_i}\leq \mathbf{y}_{e_i}^{(t+1)}e^{-\mu}$ for any $i\in\mathcal{I}$. That is,
    $\frac{\mathbf{y}_{e_j}^{(t+1)}}{\mathbf{u}_{e_j}}
    <e^{\mu}\leq \frac{\mathbf{y}_{e_i}^{(t+1)}}{\mathbf{u}_{e_i}}$, for any $i\in\mathcal{I}$ and $j\in\mathcal{J}$. 
Thus, all the edges with indices in $\mathcal{I}$ appear before those with indices in in $\mathcal{J}$ according to the non-decreasing order in \Cref{alg.update}. That is, 
\[
\mathcal{I}=\{1,\dots,k\},\mathcal{J}=\{k+1,\dots,N\}, \text{ and $\frac{\mathbf{y}^{(t+1)}_{e_1}}{\mathbf{u}_{e_1}}\geq \dots\geq \frac{\mathbf{y}^{(t+1)}_{e_k}}{\mathbf{u}_{e_k}}\geq \exp(\mu) > \frac{\mathbf{y}^{(t+1)}_{e_{k+1}}}{\mathbf{u}_{e_{k+1}}}\dots\geq \frac{\mathbf{y}^{(t+1)}_{e_{N}}}{\mathbf{u}_{e_{N}}}$.}
\]

Then by \Cref{alg.update}, $W_1=W$, and $S_1=S$. Furthermore, for each $1\leq \ell\leq i\leq k$, $\exp(-\mu)\geq \frac{\mathbf{u}_{e_i}}{\mathbf{y}_{e_i}^{(t+1)}}\geq \frac{\mathbf{u}_{e_\ell}}{\mathbf{y}_{e_\ell}^{(t+1)}}$. Thus, by the fact that $\frac{a_1+b_1}{a_2+b_2}\geq \min\{\frac{a_1}{a_2},\frac{b_1}{b_2}\}$ for any $a_1,b_1\geq 0, a_2,b_2>0$, we have that 
\[
\frac{W_1}{S_1}=\frac{W}{S}=\frac{\sum_{\ell=1}^k\mathbf{u}_{e_\ell}+\sum_{j=k+1}^N\mathbf{y}_{e_j}^{(t+1)}\exp(-\mu)}{\sum_{\ell=1}^k\mathbf{y}_{e_\ell}^{(t+1)}+\sum_{j=k+1}^N\mathbf{y}_{e_j}^{(t+1)}}\geq \min\{\frac{\mathbf{u}_{e_1}}{\mathbf{y}_{e_1}^{(t+1)}},\exp(-\mu)\}=\frac{\mathbf{u}_{e_1}}{\mathbf{y}_{e_1}^{(t+1)}}.
\]
Thus, \Cref{alg.update} outputs $\mathbf{w}_{e_1}^{(t+1)}=\min(\frac{W_1\cdot\mathbf{y}^{(t+1)}_{e_1}}{S_1},\mathbf{u}_{e_1})=\mathbf{u}_{e_1}$.

Now by induction, we have that for $i\leq k$, $W_{i}=\sum_{\ell=i}^k\mathbf{u}_{e_\ell}+\sum_{j=k+1}^N\mathbf{y}_{e_j}^{(t+1)}\exp(-\mu)$, $S_{i}=\sum_{\ell=i}^k\mathbf{y}_{e_\ell}^{(t+1)}+\sum_{j=k+1}^N\mathbf{y}_{e_j}^{(t+1)}$, and thus
\[
\frac{W_i}{S_i}=\frac{\sum_{\ell=i}^k\mathbf{u}_{e_\ell}+\sum_{j=k+1}^N\mathbf{y}_{e_j}^{(t+1)}\exp(-\mu)}{\sum_{\ell=i}^k\mathbf{y}_{e_\ell}^{(t+1)}+\sum_{j=k+1}^N\mathbf{y}_{e_j}^{(t+1)}}\geq \min\{\frac{\mathbf{u}_{e_i}}{\mathbf{y}_{e_i}^{(t+1)}},\exp(-\mu)\}=\frac{\mathbf{u}_{e_i}}{\mathbf{y}_{e_i}^{(t+1)}}.
\]
Thus, \Cref{alg.update} outputs $\mathbf{w}_{e_i}^{(t+1)}=\min(\frac{W_i\cdot\mathbf{y}^{(t+1)}_{e_i}}{S_i},\mathbf{u}_{e_i})=\mathbf{u}_{e_i}$, for each $i\leq k$. 

Now let us consider $j=k+1$. It holds that 
$W_{k+1}=\sum_{\ell=k+1}^N\mathbf{y}_{e_\ell}^{(t+1)}\exp(-\mu)$, $S_{j}=\sum_{\ell=k+1}^N\mathbf{y}_{e_\ell}^{(t+1)}$. Thus, 
\[
\frac{W_{k+1}}{S_{k+1}}=\frac{\sum_{\ell=k+1}^N\mathbf{y}_{e_\ell}^{(t+1)}\exp(-\mu)}{\sum_{\ell=k+1}^N\mathbf{y}_{e_\ell}^{(t+1)}}=\exp(-\mu).
\]
Thus, \Cref{alg.update} outputs $\mathbf{w}_{e_{k+1}}^{(t+1)}=\min(\frac{W_{k+1}\cdot\mathbf{y}^{(t+1)}_{e_{k+1}}}{S_{k+1}},\mathbf{u}_{e_{k+1}})=\exp(-\mu)\cdot\mathbf{y}^{(t+1)}_{e_{k+1}}$, as the last quantity is less than $\mathbf{u}_{e_{k+1}}$. 

Now consider any $j\geq k+1$. By induction, it holds that $W_{j}=\sum_{\ell=j}^N\mathbf{y}_{e_\ell}^{(t+1)}\exp(-\mu)$, $S_{j}=\sum_{\ell=j}^N\mathbf{y}_{e_\ell}^{(t+1)}$, and thus $W_j=\exp(-\mu) S_j$. This further implies that \Cref{alg.update} outputs $\mathbf{w}_{e_{j}}^{(t+1)}=\min(\frac{W_{j}\cdot\mathbf{y}^{(t+1)}_{e_{j}}}{S_{j}},\mathbf{u}_{e_{j}})=\exp(-\mu)\cdot\mathbf{y}^{(t+1)}_{e_{j}}$, as the last quantity is less than $\mathbf{u}_{e_{j}}$. 

   Therefore, \Cref{alg.update} always output $\mathbf{w}^{(t+1)}$ such that $\mathbf{w}^{(t+1)}=\mathbf{w}^*$.
That is, $D_{\Phi}(\mathbf{w}^{(t+1)},\mathbf{y}^{(t+1)})
    =\min_{\mathbf{w}\in\mathcal{X}} D_{\Phi}(\mathbf{w},\mathbf{y}^{(t+1)})$. This finishes the proof of the theorem.
\end{proof}

\subsection{Privacy Analysis}
\label{appendix:alg_analysis:privacy}
The privacy analysis follows from a similar approach to \cite{eliavs2020differentially}. We will first bound the privacy loss in each inner iteration when computing the noisy gradient $\mathbf{g}^{(t)}$, and apply the advanced composition (Lemma~\ref{l.adaptive_composition}) over all steps from Line~\ref{alg_md_line3} to Line~\ref{alg_md_line11} in \Cref{alg.triangle_cut_app}
. Then, we consider the $L$ composition in the outer iteration, hence giving the total privacy guarantee for \Cref{alg.triangle_cut_app}.

For any fixed outer iteration, recall that we treat $\overline{G}$ as the input graph with public $W$ and $\mathbf{u}$ first. Then, we denote $\overline{G}'$ as some neighboring graph of $\overline{G}$, which differs from $\overline{G}$ by one edge. Note that, for the sake of convenience, we first consider two graphs to be neighboring if they differ by exactly one edge. The more standard situation based on the $\ell_1$ norm will be discussed later. 
And denote $G^{(t)}$ as our solution at step $t$. Let $(\mathbf{\overline{X}}^{(t)})$ and $\mathbf{\overline{X}'}^{(t)}$ be the maximizer of $F_{\triangle}(\mathbf{w}^{(t)},\mathbf{X})$ corresponding to $\overline{G}$ and $\overline{G}'$ respectively.

The proofs of the following two lemmas are almost identical to those in \cite{eliavs2020differentially}, so we skip them.

\begin{lemma}
    \label{l.X-barX}
    Let $H(\mathbf{M})=\max_{\mathbf{X}\in\mathcal{D}}{\mathbf{M}\bullet\mathbf{X}+\lambda \log \det (\mathbf{X})+S(\mathbf{M})}$, where $S(\mathbf{M})$ is a function dependent only on $\mathbf{M}$. For two matrices $\mathbf{M}$ and $\mathbf{M'}$ , we denote $\mathbf{X}^*$ as the maximizer of $H(\mathbf{M})$ and $\mathbf{X'}^*$ as the maximizer of $H(\mathbf{M'})$. Then we have 
    $$\Vert(\mathbf{X}^*)^{-\frac{1}{2}}(\mathbf{X'}^*-\mathbf{X}^*)(\mathbf{X}^*)^{-\frac{1}{2}}\Vert_F
    \leq\frac{32}{\lambda}\Vert(\mathbf{X}^*)^{\frac{1}{2}}(\mathbf{M'}-\mathbf{M})(\mathbf{X}^*)^{\frac{1}{2}}\Vert_F$$
\end{lemma}

Intuitively speaking, the above lemma measures the difference between $\mathbf{X'}^*$ and $\mathbf{X}^*$. 
\begin{lemma}[\cite{eliavs2020differentially}]
    \label{l.dp_per_step}
    Let $\overline{\delta}$ a fixed parameter and $\mathbf{X},\mathbf{X'}\in\mathbb{R}^{2n\times 2n}$ be symmetric positive definite matrices s.t. $\Vert\mathbf{X}^{-\frac{1}{2}}(\mathbf{X'}-\mathbf{X})\mathbf{X}^{-\frac{1}{2}}\Vert_F<\frac{1}{2}$. Denote $\mathrm{pdf}_{\mathbf{X}}(\mathbf{x})$ and $\mathrm{pdf}_{\mathbf{X'}}(\mathbf{x})$ the probability density functions of $N(\mathbf{0},\mathbf{X})$ and $N(\mathbf{0},\mathbf{X'})$ respectively. Let $\overline{\varepsilon}=O(\log\frac{1}{\overline{\delta}}\cdot\Vert\mathbf{X}^{-\frac{1}{2}}(\mathbf{X'}-\mathbf{X})\mathbf{X}^{-\frac{1}{2}}\Vert_F)$. Then we have 
    $$\mathrm{pdf}_{\mathbf{X}}(\mathbf{x})\leq
    e^{\overline{\varepsilon}}\mathrm{pdf}_{\mathbf{X'}}(\mathbf{x})$$
    with probability at least $(1-\overline{\delta})$ over $\mathbf{x}\in N(\mathbf{0},\mathbf{X})$.
\end{lemma}

Recall that $\zeta\sim N(\mathbf{0},\mathbf{I}_{2n})$ is defined in \Cref{alg.triangle_cut_app}. Therefore, 
$(\mathbf{X}^{(t)})^{\frac{1}{2}}\zeta\sim N(\mathbf{0},\mathbf{X}^{(t)})$, $(\mathbf{\overline{X}}^{(t)})^{\frac{1}{2}}\zeta\sim N(\mathbf{0},\mathbf{\overline{X}}^{(t)})$. We will instantiate 
Lemma~\ref{l.dp_per_step} with $\mathbf{X}^*=\overline{\mathbf{X}}^{(t)}$, $\mathbf{X'}^*=\mathbf{\overline{X}'}^{(t)}$. Now, let us analyze the privacy guarantee of the algorithm. 

\begin{theorem}[Privacy Guarantee]
    \label{thm.dp_guarantee}
    \Cref{alg.triangle_cut_app} with parameter $\lambda=\Theta(\varepsilon^{-1})\tildeLSKt\sqrt{T}\log^{\frac{3}{2}}(\frac{T}{\delta})$, $\varepsilon_0=O(\frac{\varepsilon}{\sqrt{T\log\frac{T}{\delta}}}) $ is $(\varepsilon,\delta)$-differentially private.
\end{theorem}
    
\begin{proof}
For some fixed outer iteration and some fixed inner iteration $t$, assume $\mathbf{\overline{G}'}$ differs from $\mathbf{\overline{G}}$ in edge $e=(i,j)$ (has one more weight in $e$). As mentioned above, $\mathbf{\overline{D}}_{\triangle}^{(e)}$ is actually the divergence
of motif adjacency matrix between neighboring graphs $\mathbf{\overline{G}}$ and $\mathbf{\overline{G}'}$, i.e., $\mathbf{\overline{A}'}_{\triangle}-\mathbf{\overline{A}}_{\triangle}=\mathbf{\overline{D}}_{\triangle}^{(e)}$. 

Now let $\mathbf{M}=(\begin{array}{cc}
    \mathbf{0} & \mathbf{A}^{(t)}_{\triangle}-\overline{\mathbf{A}}_{\triangle} \\
    \mathbf{A}^{(t)}_{\triangle}-\overline{\mathbf{A}}_{\triangle} & \mathbf{0}
\end{array})$, and $\mathbf{M'}=(\begin{array}{cc}
    \mathbf{0} & \mathbf{A}^{(t)}_{\triangle}-\overline{\mathbf{A}}_{\triangle}' \\
    \mathbf{A}^{(t)}_{\triangle}-\overline{\mathbf{A}}_{\triangle}' & \mathbf{0}
\end{array})$.

Then it holds that that 
\[
\mathbf{M'}-\mathbf{M}=(\begin{array}{cc}
    \mathbf{0} & \mathbf{\overline{A}}_{\triangle}-\mathbf{\overline{A}'}_{\triangle} \\
    \mathbf{\overline{A}}_{\triangle}-\mathbf{\overline{A}'}_{\triangle} & \mathbf{0}
\end{array}).
\] 

We define $\barLSKt=\max_{(i,j)\in\binom{V}{2}}\sum_{s\in V\setminus \lbrace i,j\rbrace}(\mathbf{\overline{w}}_{(i,j)}\mathbf{\overline{w}}_{(i,s)}
+\mathbf{\overline{w}}_{(i,s)}\mathbf{\overline{w}}_{(j,s)}
+\mathbf{\overline{w}}_{(j,s)}\mathbf{\overline{w}}_{(i,j)})$ as the maximum triangle-motif cut difference of $\overline{G}$. Therefore, 
\[
\norm{\mathbf{M'}-\mathbf{M}}_1
=2\norm{\mathbf{\overline{D}}_{\triangle}^{(e)}}_1
=4(\sum_{\ell\ne i,j}\mathbf{\overline{w}}_{(i,\ell)}\mathbf{\overline{w}}_{(\ell,j)})
\leq4\barLSKt.
\] 

Denote $\mathbf{M'}-\mathbf{M}=\sum_{i=1}^{(2n)^2}c_i\mathbf{E}_i$, where $E_{i}$ has only one single non-zero-entry equal to 1. Then we have 
\[
\sum_{i,j}c_ic_j=(\sum_ic_i)^2\leq(\sum_i\abs{c_i})^2\leq(\norm{\mathbf{M'}-\mathbf{M}}_1)^2=16\barLSKt^2
. 
\]

We now instantiate Lemma~\ref{l.X-barX} with $\mathbf{X}^*=\overline{\mathbf{X}}^{(t)}$, $\mathbf{X'}^*=\mathbf{\overline{X}'}^{(t)}$, $\mathbf{M},\mathbf{M}'$ defined as above, and $H(\mathbf{M})=f_{\triangle}(\mathbf{w})$, $S(\mathbf{M})=3\sum_{s\in V\setminus\lbrace i,j\rbrace}(\mathbf{u}_{(i,s)}+\mathbf{u}_{(j,s)})(\mathbf{w}^{(t)}_e-\mathbf{w}_e)^2$. 

Then it holds that,

\begin{align*}
    \norm{(\mathbf{\overline{X}}^{(t)})^{\frac{1}{2}}(\mathbf{M'}-\mathbf{M})(\mathbf{\overline{X}}^{(t)})^{\frac{1}{2}}}_F^2
    &=\mathrm{tr}\left((\mathbf{\overline{X}}^{(t)})^{\frac{1}{2}}(\mathbf{M'}-\mathbf{M})(\mathbf{\overline{X}}^{(t)})^{\frac{1}{2}}(\mathbf{\overline{X}}^{(t)})^{\frac{1}{2}}(\mathbf{M'}-\mathbf{M})(\mathbf{\overline{X}}^{(t)})^{\frac{1}{2}}\right)\\
    &=\mathrm{tr}\left(\mathbf{\overline{X}}^{(t)}(\sum_{i=1}^{(2n)^2}c_i\mathbf{E}_i)\mathbf{\overline{X}}^{(t)}(\sum_{i=1}^{(2n)^2}c_i\mathbf{E}_i)\right)\\
    &=\sum_{i,j}c_ic_j\cdot\mathrm{tr}(\mathrm{\mathbf{\overline{X}}^{(t)}\mathbf{E}_i\mathbf{\overline{X}}^{(t)}\mathbf{E}_j})
    \\
    &=\sum_{i,j}c_ic_j(\mathbf{\overline{X}}^{(t)}\mathbf{E}_i)^\top\bullet(\mathbf{\overline{X}}^{(t)}\mathbf{E}_j)\\
    &\leq 16\barLSKt^2
\end{align*}

Therefore, by Lemma~\ref{l.X-barX}, we have 
\[
\Vert(\mathbf{\overline{X}}^{(t)})^{-\frac{1}{2}}(\mathbf{\overline{X}'}^{(t)}-\mathbf{\overline{X}}^{(t)})(\mathbf{\overline{X}}^{(t)})^{-\frac{1}{2}}\Vert_F
\leq\frac{32}{\lambda}\norm{(\mathbf{\overline{X}}^{(t)})^{\frac{1}{2}}(\mathbf{M'}-\mathbf{M})(\mathbf{\overline{X}}^{(t)})^{\frac{1}{2}}}_F
\leq\frac{32}{\lambda}\cdot4\barLSKt=\frac{128}{\lambda}\barLSKt.
\]

Recall $\zeta$ is distributed as $N(0,\mathbf{I}_{d_{\mathbf{M}}})$.
By Lemma~\ref{l.dp_per_step} and Definition~\ref{def.pdp}, we can say that, the release of $(\mathbf{X}^{(t)})^{\frac{1}{2}}\zeta$ in each mirror descent step preserves $(\overline{\varepsilon},\overline{\delta})$-differential privacy, where $\overline{\varepsilon}=O(\frac{1}{\lambda}\barLSKt\log\frac{1}{\overline{\delta}})$. 

We choose 
$\overline{\delta}=\frac{\delta}{2T}$. 
Then by advanced composition (Lemma~\ref{l.adv_composition}), the privacy guarantee of $T$ 
mirror descent steps from Line~\ref{alg_cutapp_line3} to Line~\ref{alg_cutapp_line11} in \Cref{alg.triangle_cut_app} 
is $\left(\overline{\varepsilon}\cdot4\sqrt{T\log\frac{4}{\delta}},\delta\right)$. Additionally, by Line~\ref{alg_cutapp_line6} 
in \Cref{alg.triangle_cut_app}, we $(\varepsilon_4,0)$-differentially privately release $\nu$ each outer iteration $\ell$. To sum up, by composition (Lemma~\ref{l.adaptive_composition}), the total privacy guarantee in the update process is  $\left(L\overline{\varepsilon}\cdot4\sqrt{T\log\frac{4}{\delta}}+L\varepsilon_4,\delta\right)$.

Recall that the release of $W$, $\mathbf{u}$, 
and $\tildeLSKt$ is $(\varepsilon_1+\varepsilon_2+\varepsilon_3)$-differentially private according to \Cref{def.parameters}. Since $L=\log_3(\frac{3}{\beta})$, $\varepsilon_1=\varepsilon_2=\varepsilon_3=\frac{\varepsilon}{6}$ and $\varepsilon_4=\frac{\varepsilon}{6L}$, by adaptive composition (Lemma~\ref{l.adaptive_composition}), \Cref{alg.triangle_cut_app} is $\left(2\varepsilon/3+L\overline{\varepsilon}\cdot4\sqrt{T\log\frac{4}{\delta}},\delta\right)$-differentially private, where $\overline{\varepsilon}=O(\frac{1}{\lambda}\barLSKt\log\frac{2T}{\delta})$.

Additionally, recall that $\barLSKt=\Theta(\hatLSKt)$ and $\tildeLSKt$ is differentially privately released from $\hatLSKt$. Therefore, when we set $\lambda=\Theta(\varepsilon^{-1})\tildeLSKt\sqrt{T}\log^{\frac{3}{2}}(\frac{T}{\delta})\log(\frac{3}{\beta})$, we can guarantee that the algorithm is $(\varepsilon,\delta)$-DP. 
\end{proof}

Now we consider the standard notion of edge privacy (Definition~\ref{def:neighboring}), that two graphs are called neighboring if the two vectors corresponding to their edge weights differ by at most $1$ in the $\ell_1$ norm. We first consider the graph $\overline{G}$ and its neighboring graph. Let $e_1,e_2,\dots,e_{\binom{n}{2}}$ denote an ordering of all possible edges. Denote $\overline{G}'$ (with edge weight function $\overline{\mathbf{w}}'$) as a neighboring graph of $\overline{G}$ (with edge weight function $\overline{\mathbf{w}}$). Then we denote $G_0=\overline{G}$, and for $i=1,\dots,\binom{n}{2}$, we let $G_i$ denote the graph generated by replacing the weight of edge $e_i$ in $G_{i-1}$ with $\overline{\mathbf{w}}_{e_i}'$. 
It follows that $G_{\binom{n}{2}}=\overline{G}'$. Note that $G_{i-1}$ and $G_i$ differ only by $\abs{\overline{\mathbf{w}}'_{e_i}-\overline{\mathbf{w}}_{e_i}}$ at $e_i$ entry, and $\sum_{i=1}^{\binom{n}{2}}\abs{\overline{\mathbf{w}}'_{e_i}-\overline{\mathbf{w}}_{e_i}}=1$. We let $\mathbf{A}_{\triangle,i}$ denote the adjacency matrix of triangle-motif graph of $G_i$. We have that $\norm{\mathbf{A}_{\triangle,i}-\mathbf{A}_{\triangle,i-1}}_1=\abs{\overline{\mathbf{w}}'_{e_i}-\overline{\mathbf{w}}_{e_i}}\cdot\norm{\mathbf{D}_{\triangle,i}^{(e_i)}}_1\leq\abs{\overline{\mathbf{w}}'_{e_i}-\overline{\mathbf{w}}_{e_i}}\cdot\ell_3(G_i)$. 
In the non-degenerate case (see  \Cref{cormainadditive}), we have that $\barLSKt=\Theta(\hatLSKt)=\Omega(\frac{w_{\max}\log^2(n/\beta)}{\varepsilon^2})$; furthermore, we have $\abs{\ell_3(G_i)-\barLSKt}\leq w_{\max}$, which implies that  $\ell_3(G_i)=\Theta(\barLSKt)$. Thus $\norm{\overline{\mathbf{A}}_{\triangle}-\overline{\mathbf{A}}_{\triangle}'}_1\leq\sum_{i=1}^{\binom{n}{2}}\norm{{\mathbf{A}}_{\triangle,i}-{\mathbf{A}}_{\triangle,i-1}}_1\leq O(\sum_{i=1}^{\binom{n}{2}}\abs{\overline{\mathbf{w}}'_{e_i}-\overline{\mathbf{w}}_{e_i}}\cdot\hatLSKt)\leq O(\hatLSKt)$. The remaining proof follows similarly to the analysis in \Cref{thm.dp_guarantee}.

Now let us recall that the real input graph is $\hat{G}$ instead of $\overline{G}$. While the distance between the actual input graph $\hat{G}$ and its neighboring graph is at most $1$ (in terms of edge differences), the distance between the corresponding rescaled neighboring graph and $\overline{G}$ is not necessarily $1$. Since in the non-degenerate case it holds that $W=\Omega(\frac{\log(1/\beta)}{\varepsilon})$ and $\overline{G}=\frac{\overline{W}}{\hat{W}}\hat{G}$, 
it follows that $\overline{G}$ is a scaled version of $\hat{G}$ by a constant factor. Thus the privacy loss in each inner iteration of \Cref{alg.triangle_cut_app} 
is still $O(\frac{\barLSKt}{\lambda})$.

Thus, 
our algorithm is $(\varepsilon,\delta)$-DP under the standard notion of edge DP (Definition~\ref{def:neighboring}).

\subsection{Utility Analysis}
\label{appendix:alg_analysis:precision}
To analyze the utility, we first compute the triangle-motif cut difference between $\overline{G}$ and the original input graph $\hat{G}$. 

\begin{lemma}
    \label{l.error_by_normalize}
    Given a graph $\hat{G}$ with edge weights $\hat{W}$
    s.t. $\hat{W}=\sum_{e\in\binom{V}{2}}\hat{\mathbf{w}}_e$. Let $W=\hat{W}+\lap(3/\varepsilon)+3\log(3/\beta)/\varepsilon$ and $\overline{\mathbf{w}}=\frac{W}{\hat{W}}\hat{\mathbf{w}}$. Assume that $\hat{W}=\Omega(\frac{\log(1/\beta)}{\varepsilon})$. Then we have 
    $$\vert \mathrm{Cut}_{\triangle}^{(\overline{G})}(S,T)-\mathrm{Cut}_{\triangle}^{(\hat{G})}(S,T)\vert
    \leq O\left(\hatLSKt\frac{\log(1/\beta)}{\varepsilon}
    \right) \ \mathrm{for}\ \mathrm{all}\ S,T \subseteq V$$
    with probability at least $(1-\beta/3)$.
\end{lemma}

\begin{proof}
    Denote $Y=\lap(3/\varepsilon)+3\log(3/\beta)/\varepsilon$. Then by Lemma~\ref{l.lap_bound}, we have $0<Y<6\log(3/\beta)/\varepsilon$ with probability at least $(1-\beta/3)$. Note that $\hat{W}=\Omega(\frac{\log(1/\beta)}{\varepsilon})=\Omega(Y)$. {Thus we have $W=\Theta(\hat{W})$ and $\frac{Y}{W}=O(1)$.} Therefore, 
    \begin{align*}
\mathrm{Cut}_{\triangle}^{(\overline{G})}(S,T)
        =\left(\frac{W+Y}{W}\right)^3\mathrm{Cut}_{\triangle}^{(\hat{G})}(S,T)
        \geq\mathrm{Cut}_{\triangle}^{(\hat{G})}(S,T),\end{align*}
and

    \begin{align*}
\mathrm{Cut}_{\triangle}^{(\overline{G})}(S,T)
        &=\left(\frac{W+Y}{W}\right)^3\mathrm{Cut}_{\triangle}^{(\hat{G})}(S,T)\\ &\leq\mathrm{Cut}_{\triangle}^{(\hat{G})}(S,T)+\left(3\left(\frac{Y}{W}\right)+3\left(\frac{Y}{W}\right)^2+\left(\frac{Y}{W}\right)^3\right)\mathrm{Cut}_{\triangle}^{(\hat{G})}(S,T)\\
        &\leq\mathrm{Cut}_{\triangle}^{(\hat{G})}(S,T)+7\frac{Y}{W}\mathrm{Cut}_{\triangle}^{(\hat{G})}(S,T)\\
        &\leq\mathrm{Cut}_{\triangle}^{(\hat{G})}(S,T)+7Y\barLSKt\\
        &\leq\mathrm{Cut}_{\triangle}^{(\hat{G})}(S,T)+O\left(
        \hatLSKt\frac{\log(1/\beta)}{\varepsilon}
        \right)
    \end{align*}
    The last 
    inequality is derived from the following fact:

    \begin{align*}
W\barLSKt=&\sum_{e\in\binom{V}{2}}\overline{\mathbf{w}}_e\max_{(i,j)\in\binom{V}{2}}\sum_{s\in V\setminus \lbrace i,j\rbrace}(\overline{\mathbf{w}}_{(i,j)}\overline{\mathbf{w}}_{(i,s)}
+\overline{\mathbf{w}}_{(i,s)}\overline{\mathbf{w}}_{(j,s)}
+\overline{\mathbf{w}}_{(j,s)}\overline{\mathbf{w}}_{(i,j)})\\
\geq&\sum_{(i,j)\in\binom{V}{2}}\overline{\mathbf{w}}_{(i,j)}\sum_{s\in V\setminus \lbrace i,j\rbrace}\overline{\mathbf{w}}_{(i,s)}\overline{\mathbf{w}}_{(j,s)}\\
=&(\frac{W+Y}{W})^3
\sum_{(i,j,k)\in\binom{V}{3}}\hat{\mathbf{w}}_{(i,j)}\hat{\mathbf{w}}_{(j,k)}\hat{\mathbf{w}}_{(k,i)}
\geq\mathrm{Cut}_{\triangle}^{(\hat{G})}(S,T)
\end{align*}
\end{proof}

Then to analyze how well $G$ approximates the triangle-motif cut of $\overline{G}$, we need to check the requirements and parameters in the mirror descent (Corollary~\ref{col.mirror}). First, we will prove $\mathbf{g}^{(t)}$ is an unbiased approximation of the gradient for $f_{\triangle}(\mathbf{w}^{(t)})$, and then compute the parameter $B$ by bounding $\mathbb{E}_{{\lbrace\gamma_\ell,R_{\ell,t}\rbrace}}[\norm{g^{(t)}}_{\infty}^2]$. Recall that $B^2\geq\mathbb{E}_{{\lbrace\gamma_\ell,R_{\ell,t}\rbrace}}[\norm{g^{(t)}}_{*}^2]$ and $\norm{\cdot}_*$ is $l_{\infty}$ norm. {Note that, in our algorithm, we let $\gamma_\ell$ be $\nu_e^{(\ell)}$, and let 
$R_{\ell,t}$ be $\zeta$.}

\begin{lemma}\label{l.unbias_B}
    Let $\mathbf{g}^{(t)}$ be the estimated gradient in \Cref{alg.triangle_cut_app}. Then we have, $\mathbb{E}_{{\lbrace\zeta,\nu_e^{(\ell)}\rbrace}}
    [\mathbf{g}^{(t)}]=\nabla f_{\triangle}(\mathbf{w}^{(t)})$ and $\mathbb{E}_{{\lbrace\zeta,\nu_e^{(\ell)}\rbrace}}
    [\norm{\mathbf{g}^{(t)}}_{\infty}^2]=O\left(\left({U}_\triangle+{U}_\Lambda L/\varepsilon\right)^2\log^2n\right)$.
\end{lemma}

As a directly corollary, we can have $B=\Theta\left(\left({U}_\triangle+
{U}_\Lambda {L}/\varepsilon\right)\log n\right)$.

Denote $\norm{\mathbf{A}}_{\mathrm{op}}$ as the operator norm of some matrix $\mathbf{A}$, i.e., $\norm{\mathbf{A}}_{\mathrm{op}}=\sup_{\norm{\mathbf{x}}_2=1}\norm{\mathbf{A}\mathbf{x}}_2$. For symmetric matrix, there is a useful property: $\norm{\mathbf{A}}_{\mathrm{op}}\leq\norm{\mathbf{A}}_F\leq\norm{\mathbf{A}}_1$. Then we have the following concentration inequality, which is needed in the proof of Lemma~\ref{l.unbias_B}.

\begin{theorem}
        [Hanson-Wright Theorem~\cite{rudelson2013hanson}]\label{thm.concen_ineq}
        Let $\mathbf{A}$ be an $n\times n$ matrix with entries $a_{i,j}$. If $X_1,\dots,X_n$ are mean zero, variance one random variables with sub-Guassian tail decay, i.e., 
        for any $t>0$ we have $\mathsf{Pr}[\vert X_i\vert\geq t]\leq2\exp(-\frac{t^2}{K^2})$ for some $K>0$, then
        \begin{equation*}
        \mathsf{Pr}\left[\left\vert
        \mathrm{tr}(\mathbf{A})-\sum_{i,j=1}^na_{i,j}X_iX_j
        \right\vert\geq t\right]
        \leq2\exp\left(-\min\left(
        \frac{t^2}{CK^4\norm{\mathbf{A}}_{F}^2},
        \frac{t}{CK^2\norm{\mathbf{A}}_{\mathrm{op}}}\right)\right)
        \end{equation*}
        where $C$ is some universal positive constant.
\end{theorem}

\begin{proof}[Proof of Lemma~\ref{l.unbias_B}]
Recall that,

$$\mathbf{g}_e^{(t)}=\left((\mathbf{X}^{(t)})^{\frac{1}{2}}\zeta\zeta^\top(\mathbf{X}^{(t)})^{\frac{1}{2}}\right)
      \bullet \left(\begin{array}{cc}
          \mathbf{0} & \mathbf{D}_{\triangle}^{(e)(t)} \\
          \mathbf{D}_{\triangle}^{(e)(t)} & \mathbf{0}
      \end{array}\right)+6\sum_{s\in V\setminus \lbrace i,j\rbrace}(\mathbf{u}_{(i,s)}+\mathbf{u}_{(j,s)})(\mathbf{w}_e^{(t)}-{\tilde{\mathbf{w}}_e^{(\ell)}})$$

Since $\zeta\sim N(\mathbf{0},\mathbf{I}_{2n})$, for any $i$, $\zeta_i\zeta_i$ follows the chi--squared distribution with expectation equal to $1$, and for any $i\ne j$, $\zeta_i\zeta_j$ follows the product normal distribution with expectation equal to $0$, i.e., $\mathbb{E}[\zeta\zeta^\top]=\mathbf{I}_{2n}$. Moreover, {recall that $\tilde{\mathbf{w}}_e^{(\ell)}=\overline{\mathbf{w}}_e+\nu_e^{(\ell)}$, and $\nu_e^{(\ell)}$ is distributed from $\lap(1/\varepsilon_4)$ which has expectation 0}
. {Denote $e=(i,j)$.} Therefore, 
    \begin{align*}
        \mathbb{E}_{{\lbrace\zeta,\nu_e^{(\ell)}\rbrace}}
        \left[\mathbf{g}_e^{(t)}\right]
        &=\mathbb{E}_{{\lbrace\zeta,\nu_e^{(\ell)}\rbrace}}
        \left[\left((\mathbf{X}^{(t)})^{\frac{1}{2}}\zeta\zeta^\top(\mathbf{X}^{(t)})^{\frac{1}{2}}\right)
      \bullet \left(\begin{array}{cc}
          \mathbf{0} & \mathbf{D}_{\triangle}^{(e)(t)} \\
          \mathbf{D}_{\triangle}^{(e)(t)} & \mathbf{0}
      \end{array}\right)\right]\\ 
      &+\mathbb{E}_{{\lbrace\zeta,\nu_e^{(\ell)}\rbrace}}
      \left[6\cdot \sum_{s\in V\setminus \lbrace i,j\rbrace}(\mathbf{u}_{(i,s)}+\mathbf{u}_{(j,s)})(\mathbf{w}_e^{(t)}-\overline{\mathbf{w}}_e+{\nu_e^{(\ell)}})\right]\\
      &=\mathbf{X}^{(t)}
      \bullet \left(\begin{array}{cc}
          \mathbf{0} & \mathbf{D}_{\triangle}^{(e)(t)} \\
          \mathbf{D}_{\triangle}^{(e)(t)} & \mathbf{0}
      \end{array}\right)+6\cdot \sum_{s\in V\setminus \lbrace i,j\rbrace}(\mathbf{u}_{(i,s)}+\mathbf{u}_{(j,s)})(\mathbf{w}_e^{(t)}-\overline{\mathbf{w}}_e)\\
      &=\nabla f_{\triangle}(\mathbf{w})_e
    \end{align*}
    
    The last equation follows from Lemma~\ref{l.grad_conv}. Therefore, $\mathbb{E}[\mathbf{g}^{(t)}]=\nabla f_{\triangle}(\mathbf{w}^{(t)})$.
    
Next we bound $\mathbb{E}[\Vert\mathbf{g}\Vert_{\infty}^2]$ by bound the two terms of $\mathbf{g}_e^{(t)}$ separately.

For any $e$, we denote $\mathbf{M}^{(e)}=(\mathbf{X}^{(t)})^{\frac{1}{2}}
    (\begin{array}{cc}
          \mathbf{0} & \mathbf{D}_{\triangle}^{(e)(t)} \\
          \mathbf{D}_{\triangle}^{(e)(t)} & \mathbf{0}
      \end{array})(\mathbf{X}^{(t)})^{\frac{1}{2}}$. Then we can write 
      \begin{align*}
          \left((\mathbf{X}^{(t)})^{\frac{1}{2}}\zeta\zeta^\top(\mathbf{X}^{(t)})^{\frac{1}{2}}\right)
      \bullet \left(\begin{array}{cc}
          \mathbf{0} & \mathbf{D}_{\triangle}^{(e)(t)} \\
          \mathbf{D}_{\triangle}^{(e)(t)} & \mathbf{0}
      \end{array}\right)
      &=\mathrm{tr}\left(\left(\begin{array}{cc}
          \mathbf{0} & \mathbf{D}_{\triangle}^{(e)(t)} \\
          \mathbf{D}_{\triangle}^{(e)(t)} & \mathbf{0}
      \end{array}\right)^\top
      \left((\mathbf{X}^{(t)})^{\frac{1}{2}}\zeta\zeta^\top(\mathbf{X}^{(t)})^{\frac{1}{2}}\right)\right)\\
      &=\mathrm{tr}\left(\zeta^\top(\mathbf{X}^{(t)})^{\frac{1}{2}}
      \left(\begin{array}{cc}
          \mathbf{0} & \mathbf{D}_{\triangle}^{(e)(t)} \\
          \mathbf{D}_{\triangle}^{(e)(t)} & \mathbf{0}
      \end{array}\right)
      (\mathbf{X}^{(t)})^{\frac{1}{2}}\zeta\right)\\
      &=\mathrm{tr}\left(\zeta^\top\mathbf{M}^{(e)}\zeta\right)\\
      &=\zeta^\top\mathbf{M}^{(e)}\zeta
      \end{align*}

Since $\Vert\mathbf{D}_{\triangle}^{(i,j)(t)}\Vert_1
        \leq 3\sum_{\ell\ne i,j}\mathbf{w}^{(t)}_{(i,\ell)}\mathbf{w}^{(t)}_{(\ell,j)}
        \leq3\sum_{\ell\ne i,j}\mathbf{u}_{(i,\ell)}\mathbf{u}_{(\ell,j)}
        \leq 3{U}_\triangle$.
    The last second inequality is derived from the assumption in \Cref{fail_prop}. Then, by denoting that $(\begin{array}{cc}
          \mathbf{0} & \mathbf{D}_{\triangle}^{(e)(t)} \\
          \mathbf{D}_{\triangle}^{(e)(t)} & \mathbf{0}
      \end{array})=\sum_{i=1}^{(2n)^2}c_i\mathbf{E}_i$
      where $\mathbf{E}_i$ has only one single non-zero-entry equal to $1$, we can have that 
      $$\mathbf{M}^{(e)}=\sum_{i=1}^{(2n)^2}(\mathbf{X}^{(t)})^{\frac{1}{2}}
      c_i\mathbf{E}_i
      (\mathbf{X}^{(t)})^{\frac{1}{2}}$$
      
    Therefore, similar to the proof of \Cref{thm.dp_guarantee}, we have
    $$\Vert\mathbf{M}^{(e)}\Vert_{\mathrm{op}}\leq\Vert\mathbf{M}^{(e)}\Vert_{F}
    \leq2\Vert\mathbf{D}_{\triangle}^{(e)(t)}\Vert_1\leq 6{U}_\triangle$$
    
    Applying \Cref{thm.concen_ineq}, we have
    \begin{align*}
        \mathsf{Pr}\left[\vert\mathrm{tr}(\mathbf{M}^{(e)})-\zeta^\top\mathbf{M}^{(e)}\zeta\vert\geq 6{U}_\triangle z\right]
        &\leq 2\exp\left(-\min\left(
        \frac{36{U}_\triangle^2 z^2}{CK^4\norm{\mathbf{M}^{(e)}}_{F}^2},
        \frac{6{U}_\triangle z}{CK^2\norm{\mathbf{M}^{(e)}}_{\mathrm{op}}}\right)\right)\\
        &\leq 2\exp\left(-\min\left(
        \frac{36{U}_\triangle^2 z^2}{CK^4\cdot36{U}_\triangle^2},
        \frac{6{U}_{\triangle} z}{CK^2\cdot6{U}_{\triangle}}\right)\right)\\
        &\leq O(1)\cdot\exp^{-\Theta(z)}
    \end{align*}

    Besides, the following holds, 
        $$\vert\mathrm{tr}(\mathbf{M}^{(e)})\vert
    =\sum_{i=1}^{(2n)^2}\left\vert\mathrm{tr}\left(c_i\mathbf{X}^{(t)}
      \mathbf{E}_i\right)\right\vert
    \leq \sum_{i=1}^{(2n)^2}\vert c_i\vert
    =2\Vert\mathbf{D}_{\triangle}^{(e)(t)}\Vert_1\leq6{U}_\triangle$$
    
    Therefore, we have
    
    $$\mathsf{Pr}[\vert\zeta^\top\mathbf{M}^{(e)}\zeta\vert\geq 12{U}_\triangle z]\leq O(1)\cdot\exp^{-\Theta(z)}.$$
    
    Then, for the second term of $\mathbf{g}_e^{(t)}$, since $\mathbf{w}_e^{(t)}\leq \mathbf{u}_e$ and $\overline{\mathbf{w}}_e\leq \mathbf{u}_e$, we have 

    $$3\max_{e\in\binom{V}{2}}\sum_{s\in V\setminus \lbrace i,j\rbrace}(\mathbf{u}_{(i,s)}+\mathbf{u}_{(j,s)})(\mathbf{w}^{(t)}_e-\overline{\mathbf{w}}_e)
    \leq3\max_{e\in\binom{V}{2}}\sum_{s\in V\setminus \lbrace i,j\rbrace}(\mathbf{u}_{(i,s)}+\mathbf{u}_{(j,s)})2\mathbf{u}_{(i,j)}
    \leq 6{U}_\triangle
    $$

    Then by Lemma~\ref{l.lap_bound}, with probability at least $(1-\beta')$, it holds that $\vert\nu_e^{(\ell)}\vert\leq\frac{\log(1/\beta')}{\varepsilon_4}$. Thus, 

    $$3\max_{e=(i,j)\in\binom{V}{2}}\sum_{s\in V\setminus \lbrace i,j\rbrace}(\mathbf{u}_{(i,s)}+\mathbf{u}_{(j,s)})\nu_e^{(\ell)}
    \leq3\max_{e=(i,j)\in\binom{V}{2}}\sum_{s\in V\setminus \lbrace i,j\rbrace}(\mathbf{u}_{(i,s)}+\mathbf{u}_{(j,s)})\frac{\log(1/\beta')}{\varepsilon_4}
    \leq 3{U}_\Lambda\frac{\log(1/\beta')}{\varepsilon_4}$$
    
    Therefore, 

    $$3\max_{e=(i,j)\in\binom{V}{2}}\sum_{s\in V\setminus \lbrace i,j\rbrace}(\mathbf{u}_{(i,s)}+\mathbf{u}_{(j,s)})
    (\mathbf{w}^{(t)}_e-\overline{\mathbf{w}}_e+{\nu_e^{(\ell)}})
    \leq 6{U}_\triangle+3{U}_\Lambda\frac{\log(1/\beta')}{\varepsilon_0}$$
    
    By applying union bound over $\binom{n}{2}<n^2$ coordinates of $\mathbf{g}^{(t)}$, we can conclude that,
    $$\mathsf{Pr}\left[\norm{\mathbf{g}^{(t)}}_{\infty}\geq
    (8z+6){U}_\triangle
    +3{U}_\Lambda\frac{\log(1/\beta')}{\varepsilon_4}\right]
    \leq O(1)n^2e^{-\Theta(z)}+n^2\beta'$$
    
    Recall that $\varepsilon_4=\frac{\varepsilon}{6L}$. By choosing $\beta'=\exp^{-\Theta(z)}$ and $z=s+O(\log n^2)$, where $s>0$, we have 
    $$\mathsf{Pr}\left[\norm{\mathbf{g}^{(t)}}_{\infty}^2\geq
    \left(\left(9{U}_\triangle
    +3{U}_\Lambda L/\varepsilon\right)
    \left(s+O(\log^2n)\right)\right)^2\right]
    \leq O(1)n^2e^{-s-\log n^2}
    $$

    Denote $K=9{U}_\triangle
    +3{U}_\Lambda L/\varepsilon=O\left({U}_\triangle+{U}_\Lambda L/\varepsilon\right)$, we have that 
    \begin{align*}
        \mathbb{E}\left[\norm{\mathbf{g}^{(t)}}_{\infty}^2\right]
        &=\int_{s=0}^{\infty}s\cdot\mathsf{pdf}_{\norm{\mathbf{g}^{(t)}}_{\infty}^2}ds\\
        &\leq O(K^2\log^2n)+O(1)\cdot\int_{s=O({K}^2\log^2n)}^{\infty}sd(1-\mathsf{Pr}[\norm{\mathbf{g}^{(t)}}_{\infty}^2\geq s])\\
        &\leq O(K^2\log^2n)+O(1)\cdot\int_{s=1}^{\infty}K^2(s+O(\log^2n))^2\cdot n^2e^{-s-\log n^2}ds\\
        &\leq O(K^2\log^2n)+O(K^2\log^2n)\cdot\int_{s=1}^{\infty}e^{-s}ds+O(K^2)\int_{s=1}^{\infty}e^{-s}s^2ds\\
        &\leq O\left(({U}_\triangle
    +{U}_\Lambda L/\varepsilon)^2\log^2n\right)
    \end{align*}
    The last inequality comes from that the both integrals can be bounded by a constant.
\end{proof}

Then we bound the parameter $R$ and $\rho$, which are also defined in Corollary~\ref{col.mirror}. Recall that we choose $\norm{\cdot}$ to be $l_1$ norm.

\begin{lemma}\label{l.R_rho}
    Let $\Phi(\mathbf{w})=\sum_{e\in\binom{V}{2}}\mathbf{w}_e\log\mathbf{w}_e$ be a function defined over $\mathcal{X}=\lbrace
\mathbf{w}\in\mathbb{R}_+^{\binom{V}{2}}
:\sum_{e\in\binom{V}{2}}\mathbf{w}_e=W, \mathbf{w}_e\leq \mathbf{u}_e{\ \mathrm{for}\ \forall e\in\binom{V}{2}}\rbrace$. Let $R^2=\Phi(\mathbf{w}^*)-\min_{\mathbf{w}\in\mathcal{X}}\Phi(\mathbf{w})$, where $\mathbf{w}^*=\argmin_{\mathbf{w}\in\mathcal{X}}f_{\triangle}(\mathbf{w})$. Then $\Phi$ is $\frac{1}{W}$-strongly convex with respect to $l_1$ norm, and $R^2=O(W\log n)$.
\end{lemma}
In the above, the notion $\rho$-strongly convex is defined as follows.

\begin{definition}
    [$\rho$-strongly convex]\label{def.rho_convex}
    A function $f(\mathbf{x})$ is $\rho$-strongly convex with respect to $\norm{\cdot}$ if and only if
    $$f(\mathbf{y})\geq f(\mathbf{x})+(\nabla f(\mathbf{x}))^\top(\mathbf{y}-\mathbf{x})+\frac{\rho}{2}\Vert\mathbf{y}-\mathbf{x}\Vert^2$$
\end{definition}

The Inequality described as follows is needed in the proof.

\begin{lemma}
        [Pinsker Inequality~\cite{kullback1967lower}]\label{l.pinsker}
        Let $a_i,b_i\geq 0$, $i=1,\dots,n$ s.t. $\sum_{i=1}^n a_i=\sum_{i=1}^n b_i$. Then
        $$\sum_{i=1}^na_i\log\frac{a_i}{b_i}\geq\frac{(\sum_{i=1}^n\abs{a_i-b_i})^2}{2\sum_{i=1}^na_i}$$
\end{lemma}

\begin{proof}[Proof of Lemma~\ref{l.R_rho}]
    For any $\mathbf{x},\mathbf{y}\in\mathcal{X}$, we have
    \begin{align*}
    \Phi(\mathbf{y})-\Phi(\mathbf{x})-(\nabla\Phi(\mathbf{x}))^\top(\mathbf{y}-\mathbf{x})
    &=\sum_{e\in\binom{V}{2}}\mathbf{y}_e\log\mathbf{y}_e
    -\sum_{e\in\binom{V}{2}}\mathbf{x}_e\log\mathbf{x}_e
    -\sum_{e\in\binom{V}{2}}(1+\log\mathbf{x}_e)
    \cdot(\mathbf{y}_e-\mathbf{x}_e)\\
    &=\sum_{e\in\binom{V}{2}}\mathbf{y}_e\log\frac{\mathbf{y}_e}{\mathbf{x}_e}\\
    &\geq \frac{(\sum_{e\in\binom{V}{2}}\abs{\mathbf{x}_e-\mathbf{y}_e})^2}{2W}\\
    &=\frac{1}{2W}\norm{\mathbf{x}-\mathbf{y}}_1^2
    \end{align*}

    The last second inequality is derived by Lemma~\ref{l.pinsker}. 
    According to the definition of $\rho$-strongly convex, $\Phi(\mathbf{w})$ is $\frac{1}{W}$-strongly convex with respect to $l_1$ norm.

    Then we bound $R$ by giving the upper and lower bound for $\Phi(\mathbf{w})$. Since $\Phi(\mathbf{w})$ is convex, for $\mathbf{w}\in\mathcal{X}$, by Jensen Inequality, we have
    \[\Phi(\mathbf{w})\geq\sum_{e\in\binom{V}{2}}\mathbf{w}_e\log\frac{\sum_{e\in\binom{V}{2}}\mathbf{w}_e}{\binom{n}{2}}\geq - \Omega(W\log\frac{n^2}{W}) = - \Omega(W\log n),
    \]Moreover, 
    \[
\Phi(\mathbf{w})\leq\sum_{e\in\binom{V}{2}}\mathbf{w}_e\log W=W\log W.
\]
Therefore, 

\[
R^2=\Phi(\mathbf{w}^*)-\min_{\mathbf{w}\in\mathcal{X}}\Phi(\mathbf{w})\leq W\log W + O( W\log n)=O(W\log n),
\]
where the last inequality follows from the fact that $W$ is polynomially bounded by $n$. 
\end{proof}

Now that we have bounded all the parameters, we can bound the additive error of \Cref{alg.triangle_cut_app} in the following. First, we give the bound of $f_{\triangle}(\mathbf{w})$ by utilizing Corollary~\ref{col.mirror}. 

\begin{lemma}
\label{lemma:fdeltawbound}
Let $\mathbf{w}$ be the output of \Cref{alg.triangle_cut_app} with parameter $\beta'=\beta/3$, $\eta=\frac{R}{B}\sqrt{\frac{2}{T}}$, where $R=O(\sqrt{W\log n}), B=\Theta\left(\left({U}_\triangle+{U}_\Lambda L/\varepsilon\right)\log n\right)$. 
    Then with probability at least $1-\beta/3$,
    \[
    f_{\triangle}(\mathbf{w})\leq O\left(W\log n\sqrt{\frac{\log n}{T}}({U}_\triangle+{U}_\Lambda L/\varepsilon)+\lambda n\log n\right)
    \]
\end{lemma}

\begin{proof}
Through the above discussion, we have $\rho=\frac{1}{W}$. Now we set $\beta'=\beta/3$, $\eta=\frac{R}{B}\sqrt{\frac{2}{T}}$ in \Cref{alg.triangle_cut_app}. By 
    \Cref{thm.upt_correct} and Corollary~\ref{col.mirror}, with probability at least $1-\beta/3$, it holds that $$f_{\triangle}(\mathbf{w})-3\cdot \min_{\mathbf{x}\in\mathcal{X}}f_{\triangle}(\mathbf{x})\leq 3RB\sqrt{\frac{2}{\rho T}}=
    O\left(W\log n\sqrt{\frac{\log n}{T}}({U}_\triangle+{U}_\Lambda L/\varepsilon)\right)$$

    By Lemma~\ref{l.xi_eigen}, for any $\mathbf{X}\in\mathcal{D}$, $\lambda_i\in[\frac{1}{n},2n], \forall i$. Then, it holds that 
   \[\abs{\log\det\mathbf{X}}=\abs{\sum_{i=1}^{2n}\log\lambda_i}=O(n\log n).
    \]
    Moreover, $\overline{\mathbf{w}}\in\mathcal{X}$. Therefore, 
    \[
\min_{\mathbf{x}\in\mathcal{X}}f_{\triangle}(\mathbf{x})
    \leq f_{\triangle}(\overline{\mathbf{\mathbf{w}}})
    =\lambda\log\det {\mathbf{X}}_{\overline{\mathbf{w}}}
    =\lambda O(n\log n).\]
    where $\mathbf{X}_{\overline{\mathbf{w}}}$ is referred to $\argmax_{\mathbf{X}\in\mathcal{D}}F_{\triangle}(\overline{\mathbf{w}},\mathbf{X})$.
    
This further implies that with probability at least $1-\beta/3$, 
\[
f_{\triangle}(\mathbf{w})
    \leq O\left(W\log n\sqrt{\frac{\log n}{T}}({U}_\triangle+{U}_\Lambda L/\varepsilon)+\lambda n\log n\right).
\]
\end{proof}

Recall that the triangle-motif cut difference between $\overline{G}$ and $\hat{G}$ can be bounded by $O\left({\hatLSKt\frac{\log(1/\beta)}{\varepsilon}
}\right)$ according to Lemma~\ref{l.error_by_normalize}. Therefore, we give the total additive error in the following, by bound the difference between the input graph $\hat{G}$ and the output graph $G$.

\begin{theorem}[Utility Guarantee]
    \label{thm.precision}
Let $G$ be the output graph of \Cref{alg.triangle_cut_app} with $T=\Theta(\frac{W({\varepsilon}{U}_\triangle+{U}_\Lambda)}{n\log(n/\delta)\tildeLSKt})$, $\beta'=\beta/3$, $\eta=\frac{R}{B}\sqrt{\frac{2}{T}}$, where $R=O(\sqrt{W\log n}), B=\Theta\left(\left({U}_\triangle+{U}_\Lambda L/\varepsilon\right)\log n\right)$. {Assume that $\hat{W}=\Omega(\frac{\log(1/\beta)}{\varepsilon})$ and $\hatLSKt=\Omega(\frac{w_{\max}\log^2(n/\beta)}{\varepsilon^2})$.} Then the triangle-motif cut difference between $G$ and the original graph $\hat{G}$ is at most 
\[O\left(
    {\hatLSKt\frac{\log(1/\beta)}{\varepsilon}
    +\sqrt{Wn({U}_\triangle+{U}_\Lambda/\varepsilon)\hatLSKt/\varepsilon}
    \log^2(\frac{n}{\delta\beta})}
    \right)
    \]
with probability at least $(1-\beta)$.
\end{theorem}
\begin{proof}
    We denote that 
    \[
h(\mathbf{w})=\max_{\mathbf{X}\in\mathcal{D}'}H(\mathbf{w},\mathbf{X}),\ \  \text{ where } H(\mathbf{w},\mathbf{X})=
    \left(\begin{array}{cc}
        \mathbf{0} & \mathbf{A}_\triangle-\overline{\mathbf{A}}_\triangle \\
        \mathbf{A}_\triangle-\overline{\mathbf{A}}_\triangle & \mathbf{0}
    \end{array}\right)\bullet \mathbf{X}
\]
   and \[\mathcal{D}'=\left\lbrace\mathbf{X}\in\mathbb{R}^{2n}:\mathbf{X}\ \mathrm{is}\ \mathrm{symmetric},
    \mathbf{X}\succeq \mathbf{0},
    \mathrm{and}\ \mathbf{X}_{ii}=1\ \mathrm{for}\ \forall{i}
    \right\rbrace\]

    Recall that the triangle-motif cut difference between the output graph $G$ and $\overline{G}$ can be bounded by $h(\mathbf{w})$ up to a constant factor according to Lemma~\ref{l.triangle_diff_primal_sdp}. Therefore, if we can bound the gap between $h(\mathbf{w})$ and $f_{\triangle}(\mathbf{w})$, then we can bound $h(\mathbf{w})$, hence giving a bound for the triangle-motif cut difference between $G$ and $\overline{G}$.

    Denote that
    $$f_1(\mathbf{w})=\max_{\mathbf{X}\in\mathcal{D}}F_1(\mathbf{w},\mathbf{X}),\ \ \text{where } F_1(\mathbf{w},\mathbf{X})=
    \left(\begin{array}{cc}
        \mathbf{0} & \mathbf{A}_\triangle-\overline{\mathbf{A}}_\triangle \\
        \mathbf{A}_\triangle-\overline{\mathbf{A}}_\triangle & \mathbf{0}    \end{array}\right)\bullet \mathbf{X},$$
    $$f_2(\mathbf{w})=\max_{\mathbf{X}\in\mathcal{D}}F_2(\mathbf{w},\mathbf{X}),\ \  \text{where } F_2(\mathbf{w},\mathbf{X})=
    \left(\begin{array}{cc}
        \mathbf{0} & \mathbf{A}_\triangle-\overline{\mathbf{A}}_\triangle \\
        \mathbf{A}_\triangle-\overline{\mathbf{A}}_\triangle & \mathbf{0}
    \end{array}\right)\bullet \mathbf{X}
    +\lambda\log\det\mathbf{X},$$
and \[\mathcal{D}=\left\lbrace\mathbf{X}\in\mathbb{R}^{2n}:\mathbf{X}\ \mathrm{is}\ \mathrm{symmetric},
    \mathbf{X}\succeq \frac{1}{n}\mathbf{I}_{2n},
    \mathrm{and}\ \mathbf{X}_{ii}=1\ \mathrm{for}\ \forall{i}
    \right\rbrace\]

    Note that $f_2(\mathbf{w})$ and $f_{\triangle}(\mathbf{w})$ have the same maximizer. By the definition of $f_\triangle(\mathbf{w})$ and $f_2(\mathbf{w})$ and that $\sum_{(i,j)\in\binom{V}{2}}3(\mathbf{w}_{(i,j)}-\overline{\mathbf{w}}_{(i,j)})^2
    \cdot\sum_{s\in V\setminus \lbrace i,j\rbrace}(\mathbf{u}_{(i,s)}+\mathbf{u}_{(j,s)})\geq 0$, we have 
    \[f_{\triangle}(\mathbf{w})\geq f_2(\mathbf{w}).\] 

    As proved in Lemma~\ref{lemma:fdeltawbound}, $\abs{\log\det\mathbf{X}}=O(\lambda n\log n)$ for any $\mathbf{X}\in\mathcal{D}$. Let the maximizer of $f_1(\mathbf{w})$ be $\mathbf{X}^{(1)}$. Then it holds that, 
    \[f_2(\mathbf{w})
    \geq F_2(\mathbf{w},\mathbf{X}^{(1)})
    \geq F_1(\mathbf{w},\mathbf{X}^{(1)})-C\cdot \lambda n \log n
    =f_1(\mathbf{w})-C\cdot \lambda n \log n,\]
    for some constant $C>0$.

    Recall that $$\mathcal{D}=\left\lbrace\mathbf{X}\in\mathbb{R}^{2n}:\mathbf{X}\ \mathrm{is}\ \mathrm{symmetric},
    \mathbf{X}\succeq \frac{1}{n}\mathbf{I}_{2n},
    \mathrm{and}\ \mathbf{X}_{ii}=1\ \mathrm{for}\ \forall{i}
    \right\rbrace$$
    and $$\mathcal{D}'=\left\lbrace\mathbf{X}\in\mathbb{R}^{2n}:\mathbf{X}\ \mathrm{is}\ \mathrm{symmetric},
    \mathbf{X}\succeq \mathbf{0},
    \mathrm{and}\ \mathbf{X}_{ii}=1\ \mathrm{for}\ \forall{i}
    \right\rbrace$$

    Let $\mathbf{X}^{(2)}$ be the maximizer of $h(\mathbf{w})\in\mathcal{D}'$  
    . Let $\mathbf{X}^{(3)}=(1-\frac{1}{n})\cdot\mathbf{X}^{(2)}+\frac{1}{n}\cdot \mathbf{I}_{2n}$. We can see that $\mathbf{X}^{(3)}\in\mathcal{D}$. Then we have, 
    \begin{align*}
    f_1(\mathbf{w})
    &\geq F_1(\mathbf{w},\mathbf{X}^{(3)})\\
    &=\left(\begin{array}{cc}
        \mathbf{0} & \mathbf{A}_\triangle-\overline{\mathbf{A}}_\triangle \\
        \mathbf{A}_\triangle-\overline{\mathbf{A}}_\triangle & \mathbf{0}
    \end{array}\right)\bullet\left(
    (1-\frac{1}{n})\cdot\mathbf{X}^{(2)}
    +\frac{1}{n}\mathbf{I}_{2n}\right)\\
    &=(1-\frac{1}{n})\cdot\left(\begin{array}{cc}
        \mathbf{0} & \mathbf{A}_\triangle-\overline{\mathbf{A}}_\triangle \\
        \mathbf{A}_\triangle-\overline{\mathbf{A}}_\triangle & \mathbf{0}
    \end{array}\right)\bullet\mathbf{X}^{(2)}\\
    &\geq\frac{n-1}{n}H(\mathbf{w},\mathbf{X}^{(2)})\\
    &=\frac{n-1}{n}h(\mathbf{w})\geq \frac12h(\mathbf{w})
    \end{align*}

    In a conclusion, we have,

    \[h(\mathbf{w})
        <2f_1(\mathbf{w})<2f_2(\mathbf{w})+O(\lambda n\log n)<2f_\triangle(\mathbf{w})+O(\lambda n\log n)\]

    Since we assume $\hat{W}=\Omega(\frac{\log(1/\beta)}{\varepsilon})$, according to Lemma~\ref{l.error_by_normalize}, the triangle-motif cut difference between $\overline{G}$ and $\hat{G}$ is $O({\hatLSKt\frac{\log(1/\beta)}{\varepsilon}})$. 
    Then to sum up, with probability at least $(1-\beta/3)$, it holds that 
    \begin{align*}
        &\abs{\mathrm{Cut}_{\triangle}^{(G)}(S,T)-\mathrm{Cut}_{\triangle}^{(\hat{G})}(S,T)}\\
    \le &\abs{\mathrm{Cut}_{\triangle}^{(G)}(S,T)-\mathrm{Cut}_{\triangle}^{(\overline{G})}(S,T)}+\abs{\mathrm{Cut}_{\triangle}^{(\overline{G})}(S,T)-\mathrm{Cut}_{\triangle}^{(\hat{G})}(S,T)}\\
    \leq&O(h(\mathbf{w}))+\abs{\mathrm{Cut}_{\triangle}^{(\overline{G})}(S,T)-\mathrm{Cut}_{\triangle}^{(\hat{G})}(S,T)}\\
    \le &O\left({\hatLSKt\frac{\log(1/\beta)}{\varepsilon}
    }\right)+2f_{\triangle}(\mathbf{w})+O(\lambda n\log n)\\
    \leq &O\left({\hatLSKt\frac{\log(1/\beta)}{\varepsilon}
    }
    +W\log n\sqrt{\frac{\log n}{T}}
    \left({U}_\triangle
    +\frac{{U}_\Lambda}{\varepsilon}L\right)
    +\lambda n\log n\right).
    \end{align*} 
    
Recall that, according to \Cref{thm.dp_guarantee}, $\lambda=\Theta(\varepsilon^{-1})\tildeLSKt\sqrt{T}\log^{\frac{3}{2}}(\frac{T}{\delta})\log(3/\beta)$ and $L=\log_3(\frac{3}{\beta})$. By the assumption $\hatLSKt=\Omega(\frac{w_{\max}\log^2(n/\beta)}{\varepsilon^2})$, we have $\tildeLSKt=\Theta(\hatLSKt)$. Moreover, note that, the algorithm fails with at most $2\beta/3$ probability referring to \Cref{s.alg}. Then we can conclude that, by choosing $T=\Theta(\frac{W({\varepsilon}{U}_\triangle+{U}_\Lambda)}{n\log(n/\delta)\tildeLSKt})$, the triangle-motif cut difference between the output graph $G$ (with edge weights $\mathbf{w}$) and $\hat{G}$ is at most
    \begin{align*}
        &O\left(\hatLSKt\frac{\log(1/\beta)}{\varepsilon}
    +W\log n\sqrt{\frac{\log n}{T}}
    \left({U}_\triangle
    +\frac{{U}_\Lambda}{\varepsilon}L\right)
    +\lambda n\log n\right)\\
    \leq&O\left(\hatLSKt\frac{\log(1/\beta)}{\varepsilon}
    +\left(\frac{({U}_\triangle+{U}_\Lambda/\varepsilon)W}{\sqrt{T}}
    +\tildeLSKt\sqrt{T}n\log(\frac{nT}{\delta})/\varepsilon\right)
    \log^{\frac{3}{2}}(\frac{nT}{\delta})\log(\frac{1}{\beta})\right)\\
    \leq&O\left(\hatLSKt\frac{\log(1/\beta)}{\varepsilon}
    +\sqrt{Wn({U}_\triangle+{U}_\Lambda/\varepsilon)\hatLSKt/\varepsilon}
    \log^2(\frac{n}{\delta\beta})
    \right)
    \end{align*}
    with probability at least $(1-\beta)$.
\end{proof}

To provide a more concise guarantee, we first give some bounds to quantities  
$U_{\triangle}$ and $U_{\Lambda}$. Let $w_{\max}=\max_{e\in\binom{V}{2}}{\mathbf{w}}_e=\Theta(\max_{e\in\binom{V}{2}}\overline{\mathbf{w}}_e)$ be the maximum edge weight. 

\begin{lemma}
    \label{l.quantity_bound}
It holds that 
${U}_\triangle= O(nw^2_{\max}\frac{\log^2(n/\beta)}{\varepsilon^2})$, ${U}_\Lambda= O(nw_{\max}\frac{\log(n/\beta)}{\varepsilon})$.
\end{lemma}

\begin{proof}
We have that 
    \begin{align*}
        {U}_\Lambda&=\max_{(i,j)\in\binom{V}{2}} 
\sum_{s\in V\setminus \lbrace i,j\rbrace}(\mathbf{u}_{(i,s)}
+\mathbf{u}_{(j,s)})\leq\max_{(i,j)\in\binom{V}{2}} 
\sum_{s\in V\setminus \lbrace i,j\rbrace}(\overline{\mathbf{w}}_{(i,s)}
+\overline{\mathbf{w}}_{(j,s)}+12\frac{\log(3n^2/\beta)}{\varepsilon}+2\frac{W}{\binom{n}{2}})\\
        &\leq 2dw_{\max}+12n\frac{\log(3n^2/\beta)}{\varepsilon}+4\frac{W}{n}
        \leq O(nw_{\max}\frac{\log(n/\beta)}{\varepsilon}),
    \end{align*}
and that
    \begin{align*}
        {U}_\triangle&=\max_{(i,j)\in\binom{V}{2}} 
\sum_{s\in V\setminus \lbrace i,j\rbrace}(\mathbf{u}_{(i,j)}\mathbf{u}_{(i,s)}
+\mathbf{u}_{(i,s)}\mathbf{u}_{(j,s)}
+\mathbf{u}_{(j,s)}\mathbf{u}_{(i,j)})\\
    &\leq \max_{(i,j)\in\binom{V}{2}} 
\sum_{s\in V\setminus \lbrace i,j\rbrace}(\overline{\mathbf{w}}_{(i,j)}\overline{\mathbf{w}}_{(i,s)}
+\overline{\mathbf{w}}_{(i,s)}\overline{\mathbf{w}}_{(j,s)}
+\overline{\mathbf{w}}_{(j,s)}\overline{\mathbf{w}}_{(i,j)})
+144n\frac{\log^2(3n^2/\beta)}{\varepsilon^2}+2n(\frac{W}{\binom{n}{2}})^2\\\
&+(12\frac{\log(3n^2/\beta)}{\varepsilon}+2\frac{W}{\binom{n}{2}})\max_{(i,j)\in\binom{V}{2}} 
\sum_{s\in V\setminus \lbrace i,j\rbrace}
(\overline{\mathbf{w}}_{(i,j)}
+\overline{\mathbf{w}}_{(i,s)}
+\overline{\mathbf{w}}_{(j,s)})\\
&\leq \barLSKt+O(nw_{\max}(\log^2(n/\beta)+\frac{W}{\binom{n}{2}}))\leq O(dw_{\max}^2+nw_{\max}\frac{\log^2(n/\beta)}{\varepsilon^2})\leq O(nw^2_{\max}\frac{\log^2(n/\beta)}{\varepsilon^2}).
    \end{align*}
\end{proof}

\begin{corollary}\label{cormainadditive}
Let $G$ be the output graph of \Cref{alg.triangle_cut_app} with $T=\Theta(\frac{W({\varepsilon}{U}_\triangle+{U}_\Lambda)}{n\log(n/\delta)\tildeLSKt})$, $\beta'=\beta/3$, $\eta=\frac{R}{B}\sqrt{\frac{2}{T}}$, where $R=O(\sqrt{W\log n}), B=\Theta\left(\left({U}_\triangle+{U}_\Lambda L/\varepsilon\right)\log n\right)$. Denote $\hat{W}$ and $w_{\max}$ as the sum of edge weights and the maximum edge weight of the original graph $\hat{G}$ separately. Then the triangle-motif cut difference between $G$ and $\hat{G}$ is at most 
\[
O\left(\sqrt{\hat{W}\cdot\hatLSKt}\cdot nw_{\max}
    \frac{\log^{{2}}(n/\delta\beta)}{\varepsilon^{{3/2}}}
    \right)
    \]
\end{corollary}
\begin{proof}
We first assume that $\hat{W}=\Omega(\frac{\log(1/\beta)}{\varepsilon})$ and $\hatLSKt=\Omega(\frac{w_{\max}\log^2(n/\beta)}{\varepsilon^2})$, which we refer to as the non-degenerate case. By Lemma~\ref{l.quantity_bound}, ${U}_\triangle= O(nw^2_{\max}\frac{\log^2(n/\beta)}{\varepsilon^2})$, ${U}_\Lambda= O(nw_{\max}\frac{\log(n/\beta)}{\varepsilon})$. Therefore, by \Cref{thm.precision}, the additive error is at most 
    \begin{align*}
    &{O\left(
        \hatLSKt\frac{\log(1/\beta)}{\varepsilon}
    +\sqrt{Wn({U}_\triangle+{U}_\Lambda/\varepsilon)\hatLSKt/\varepsilon}
    \log^2(\frac{n}{\delta\beta})
    \right)}\\
    \leq&O\left(\hatLSKt\frac{\log(1/\beta)}{\varepsilon}+\sqrt{W\hatLSKt}nw_{\max}
    \frac{\log^2(n/\delta\beta)}{\varepsilon^{3/2}}
    \right)\\
    \leq&O\left(\sqrt{W\hatLSKt}nw_{\max}
    \frac{\log^2(n/\delta\beta)}{\varepsilon^{3/2}}
    \right).
    \end{align*}
The last inequality comes from 
the fact that $\hatLSKt\leq nw^2_{\max}$.

    Furthermore, recall that $W=\hat{W}+\lap({3/\varepsilon})+3\log(3/\beta)/\varepsilon$. Note that we have assumed $\hat{W}=\Omega(\frac{\log(1/\beta)}{\varepsilon})$,  
    the additive error is at most
    $O\left(\sqrt{\hat{W}\cdot\hatLSKt}\cdot nw_{\max}
    \frac{\log^{{2}}(n/\delta\beta)}{\varepsilon^{{3/2}}}
    \right)$.

In the degenerate case when the above assumptions are not satisfied, we will show that, the weight sum of the triangles of $\hat{G}$ is at most $O\left(\sqrt{\hat{W}\cdot\hatLSKt}\cdot nw_{\max}
    \frac{\log^{{2}}(n/\delta\beta)}{\varepsilon^{{3/2}}}
    \right)$. That is, we can release an empty graph and achieving the same utility guarantee.

We first show that:
\begin{align*}
    \sum_{(i,j,k)\in\binom{V}{3}}\hat{w}_{(i,j)}\hat{w}_{(i,s)}\hat{w}_{(j,s)}\leq
    \sum_{(i,j)\in\binom{V}{2}}\hat{w}_{(i,j)}\cdot\sum_{s\ne i,j}\hat{w}_{(i,s)}\hat{w}_{(j,s)}=\hat{W}\cdot\hatLSKt
\end{align*}

Additionally, we have $\hat{W}\leq n^2w_{\max}$ and $\hatLSKt\leq nw^2_{\max}$. Thus it follows:

If $\hat{W}=o(\frac{\log(1/\beta)}{\varepsilon})$, then the weight sum of the triangles of $\hat{G}$ is at most 
\begin{align*}
    \hat{W}\cdot\hatLSKt
    = o(\sqrt{\hat{W}\cdot\hatLSKt}\sqrt{\frac{\log(1/\beta)}{\varepsilon}}\sqrt{nw^2_{\max}})
    =o\left(\sqrt{\hat{W}\cdot\hatLSKt}\cdot nw_{\max}
    \frac{\log^{{2}}(n/\delta\beta)}{\varepsilon^{{3/2}}}
    \right).
\end{align*}

If $\hatLSKt=o(\frac{w_{\max}\log^2(n/\beta)}{\varepsilon^2})$, then the weight sum of the triangles of $\hat{G}$ is at most 
\begin{align*}
    \hat{W}\cdot\hatLSKt
    = o(\sqrt{\hat{W}\cdot\hatLSKt}\cdot\sqrt{n^2w_{\max}\cdot\frac{w_{\max}\log^2(n/\beta)}{\varepsilon^2}})
    =o\left(\sqrt{\hat{W}\cdot\hatLSKt}\cdot nw_{\max}
    \frac{\log^{{2}}(n/\delta\beta)}{\varepsilon^{{3/2}}}
    \right).
\end{align*}
    
\end{proof}

Finally, we note that the correctness of \Cref{thm:main} follows from \Cref{thm.dp_guarantee} and Corollary~\ref{cormainadditive} with error probability $\beta=\frac{1}{4}$.

\subsection{Complexity Analysis}
\label{appendix:alg_analysis:time}
\begin{lemma}
    \label{l.total_time}
    \Cref{alg.triangle_cut_app} can be implemented in time $\tilde{O}(n^6W\log^{O(1)}(n))$  with the same level of guarantee for additive error and privacy.
\end{lemma}
Recall that we assume the bound of the sum of edge weights to be polynomial. Thus \Cref{alg.triangle_cut_app} runs in polynomial time. The proof of the lemma is similar to \cite{eliavs2020differentially}, so here we only give a proof sketch.
\begin{proof}[Proof Sketch of Lemma~\ref{l.total_time}]
According to \Cref{thm.precision} and Lemma~\ref{l.quantity_bound}, we have $T=\Theta(\frac{W({\varepsilon}{U}_\triangle+{U}_\Lambda)}{n\log(n/\delta)\tildeLSKt})=\tilde{O}(W)$. In each inner iteration of \Cref{alg.triangle_cut_app}, note that \Cref{alg.update} runs in $O(n^2)$ time. Denote $\mathbf{X}^*$ be the maximizer of the SDP $\max_{\mathbf{X}\in\mathcal{D}}F_{\triangle}(\mathbf{w},\mathbf{X})$. We use the algorithm of \cite{lee2015faster} to find an approximate solution of the SDP in Line~\ref{alg_cutapp_line6} of \Cref{alg.triangle_cut_app} 
in time $\tilde{O}(n^6\log^{O(1)}(n/\mu))$, such that, $\norm{(\mathbf{X}^*)^{-\frac{1}{2}}(\mathbf{X}^*-\mathbf{X})(\mathbf{X}^*)^{-\frac{1}{2}}}_F\leq\mu$. Therefore, within $LT$ iterations, \Cref{alg.triangle_cut_app} can be implemented in time $\tilde{O}(n^6W\log^{O(1)}(n))$. According to \cite{eliavs2020differentially}, the additive error only differs by a constant factor when we choose $\mu=1/n^{O(1)}$. And the privacy loss in each iteration of \Cref{alg.triangle_cut_app} 
is still $O(\frac{\hatLSKt}{\lambda})$, hence the privacy guarantee is the same as before.
\end{proof}

\section{The Lower Bound}\label{sec:lowerbound}
Given graph $G$, let $\LSKh=\max_{(i,j)\in\binom{V}{2}}\sum_{V'=(i,j,v_1,v_2,\dots,v_{h-2})\in\binom{V}{h}}
\prod_{(k,\ell)\in\binom{V'}{2}\setminus\lbrace(i,j)\rbrace}\mathbf{w}_{(k,\ell)}$ denotes the local sensitivity of $K_h$-motif cuts of $G$. Note that, its local sensitivity is defined as the maximum $K_h$-motif cut difference between $G$ and its neighboring graphs. For unweighted graph with maximum degree at most $d$, we have $\LSKh=O(d^{h-2})$. For unweighted dense graph, we have $\LSKh=\Theta(n^{h-2})$. For unweighted graph generated from $G(n,p)$, we have $\LSKh=\Theta(n^{h-2}p^{\frac{h^2-h-2}{2}})$ with high probability.

We show the following lower bound. Note that \Cref{thm:lowerbound} follows from the following theorem by setting $h=3$ and $\beta=\frac{3}{4}$. 
\begin{theorem}\label{thm:lowerfinal}
Let $\mathcal{M}$ be an $(\varepsilon,\delta)$-differentially private mechanism, and let $G$ be a graph generated from $G(n,p)$ with $(\frac{\log n}{n})^{1/(h-1)}\ll p\leq \frac12$. If $\mathcal{M}$ answers the $K_h$-motif size queries of all $(S,T)$-cut on $G$, or on a scaled version of $G$ with total edge weight $W$,  up to an additive error $\alpha$ with probability at least $\beta$, then:

\[\alpha\geq \Omega\left(\max\left(
\frac{\sqrt{mn}\cdot\LSKh}{\varepsilon}(1-c),\frac{\sqrt{Wn}\cdot\LSKh}{\sqrt{\varepsilon}}(1-c)\right)\right)\]
where $c=\frac{e-1}{e^{\varepsilon}-1}\cdot \frac{9\delta}{\beta}$, 
and $m=\Theta(pn^2)$ is the number of edges of $G$.
\end{theorem}

The proof is based on generalizing the lower bound in \cite{eliavs2020differentially} for the edge case using the discrepancy of $3$-coloring of $h$-uniform hypergraphs. Specifically, for an unweighted graph $G=(V,E)$, we consider the matrix $\mathbf{A}$ with $\binom{n}{h}$ columns corresponding to $h$-tuple of vertices (or a copy of $K_h$) and rows corresponding to the pairs of  sets $S,T\subset V$ such that 
\begin{displaymath}
\mathbf{A}_{(S,T),I}=\left\{ \begin{array}{ll}
1 & \textrm{if $I\in (S\times T)$}\\
0 &\textrm{otherwise}
\end{array}
\right.
\end{displaymath}

Note that $\mathbf{A}$ is fixed and does not depend on $G$. Let $\mathbf{x}\in \{0,1\}^{\binom{n}{2}}$ be the indicator vector of $E$. Let $\mathbf{x}_{K_h}=f_{K_h}(x)\in \{0,1\}^{\binom{n}{h}}$ be the indicator vector of $K_h$ in $G$, i.e. for each tuple $i_1,i_2,\dots,i_h$ of $h$ different indices, $(\mathbf{x}_{K_h})_{i_1,i_2,\dots,i_h}=1$ if the subgraph induced by $i_1,i_2,\dots,i_h$ is $K_h$ and $0$ otherwise. Then $\mathbf{A}\mathbf{x}_{K_h}$ specifies the $K_h$-motif size of all $(S,T)$-cuts in $G$, i.e., we have 
\[
(\mathbf{A}\mathbf{x}_{K_h})_{(S,T)}=\sum_{I\in \mathcal{M}(G,K_h)} \mathbf{1}_{\text{$I$ crosses $(S,T)$}}.
\]
Then we define the discrepancy of $A$ in terms of the set of $3$-colorings over the set of all $K_h$ motifs, and give bounds on its discrepancy by using random graphs, which will then imply our lower bound by a reduction from \cite{muthukrishnan2012optimal}. Specifically, we define the discrepancy of a matrix as follows.
\begin{definition}\label{def:descrepancy}
Let $\mathbf{B}$ be a $0/1$ matrix with $\binom{n}{h}$ columns and $\mathcal{C}\subseteq \{-1,0,1\}^{\binom{n}{h}}$ be the set of allowed $K_h$ colorings. We define

\[
\disc_{\mathcal{C}}(\mathbf{B})=\min\{\norm{\mathbf{B}\chi}_\infty: \chi\in \mathcal{C}\}
\]
\end{definition}

We will make use of the following lemma. 
\begin{lemma}[Lemma 6 in \cite{bollobas2006discrepancy}]\label{lem:simpleinequality}
Let $\{a_i,1\leq i\leq n\}$ be a sequence of real numbers. Let $\rho_i\in \{0,1\}$ be i.i.d. Bernoulli, for $1\leq i\leq n$. Then 
\[
\mathbb{E}\left[\abs{\sum_{i=1}^n\rho_i a_i}\right]\geq \frac{\sum_{i=1}\abs{a_i}}{\sqrt{8 n}}.
\]
\end{lemma}

Now we prove our main lemma in this section, i.e., Lemma~\ref{lemma:discrepancymain}, which is a generalization of a result of \cite{eliavs2020differentially} (see also \cite{bollobas2006discrepancy}). 

\begin{lemma}\label{lemma:discrepancymain}
Let $\gamma,\sigma\in (0,1/2)$. Let $\mathcal{C}_{\sigma,\gamma}$ be the set of all vectors $\chi=\mathbf{x}_{K_h}-\mathbf{x}_{K_h}'$ where $\mathbf{x},\mathbf{x}'$ are the indicator vector of edges of graphs, denoted by $G=(V,E)$ and $G'=(V,E')$, respectively, such that 
\begin{enumerate}
    \item $\norm{\mathbf{x}-\mathbf{x}'}_1\geq \sigma \gamma n^2$;
    \item for each vertex $v\in V$, its degree belongs to the interval $[\gamma n/2,2\gamma n]$;
\item for each edge $e\in E\cup E'$, the number of $K_h$-instances containing $e$ belongs to the interval $[\frac{\LSKh}{2},2\LSKh]$;

\item\label{condition:Khinstance} for $1\leq i\leq h$ and any subset $B$ with $i$ distinct vertices in $G$ and $G'$, the number of vertices $t$ such that $t$ is connected to all vertices in $B$ belongs to the interval $\left[\frac{n\gamma^{i}}{2}, 2n\gamma^{i}\right]$.
\end{enumerate}

Then for the matrix $\mathbf{A}$ defined above, we have
\[
\disc_{\mathcal{C}_{\sigma,\gamma}}(\mathbf{A})\geq 2^{-h-5}\sigma\cdot 
\gamma^{1/2}n^{3/2}\LSKh.
\]
\end{lemma}
\begin{proof}
For any $i\leq h$ and any set $W\subseteq V$, we let $W^{(i)}$ denote the set of all subsets $R\subseteq W$ with exactly $i$ distinct vertices. For any $i\leq h$ and a set $R\subseteq V^{(i)}$, and $h-i$ other vertices $t_{i+1},\dots,t_{h}$, we let $\chi_{R,t_{i+1},\dots,t_{h}}$ denote the coloring of the potential $K_h$ instance induced by the vertex set $R\cup\{t_{i+1},\dots,t_{h}\}$. We define $\chi_R=\sum_{t_{i+1},\dots,t_{h}\in V^{(h-i)}}\chi_{R,t_{i+1},\dots,t_{h}}$ and define $\chi_{R,t}=\sum_{t_{i+2},\dots,t_{h}\in V^{(h-i-1)}}\chi_{R,t,t_{i+2},\dots,t_{h}}$. We extend the definition of $\chi$ to any tuple $i_1, \dots, i_h \in V^{(h)}$ by setting $\chi_{i_1, \dots, i_h}=0$ whenever there exists a repeated index in the tuple.

For two disjoint sets $P,Q\subset V$, we let 
\[\chi_{1,h-1}(P,Q)=\sum_{t_1\in P,B\in Q^{(h-1)}}\chi_{t_1,B}=\sum_{t_1\in P, t_2,\dots,t_h\in Q}\chi_{t_1,\dots,t_h}.
\]
We need the following lemma:
\begin{lemma}
\label{lemma:disc_1toh_to_any}
    If $\disc_{\mathcal{C}_{\sigma,\gamma}}(\mathbf{A})\leq M$, then for any $\chi\in\mathcal{C}_{\sigma,\gamma}$ and disjoint subsets $P, Q\subset V$, 
    \[\abs{\chi_{1,h-1}(P,Q)} \leq2^{2h^2}M.\]
\end{lemma}

\begin{proof}
Let $Z$ be a random subset of $P$, where each vertex is chosen independently with probability $p$. Then
\[\oE[(A\mathbf{x}_{K_h})_{(Z\cup Q,Z\cup Q)}]=\sum_{i=0}^hp^i\cdot\sum_{t_1\in P,\dots,t_i\in P,t_{i+1}\in Q\dots,t_h\in Q}\chi_{t_1,\dots,t_h}\]

By Markov Inequality, we have the fact that if $P(x)$ is a $h$-degree polynomial with $\sup_{x\in[0,1]}\abs{P(x)}\leq1$ then every coefficient of $P(x)$ has absolute value at most $2^hh^{2h}/h!$. Since $M\geq\disc_{\mathcal{C}_{\sigma,\gamma}}(\mathbf{A})\geq\sup_{p\in[0,1]}\abs{\sum_{i=0}^hp^i\cdot\sum_{t_1\in P,\dots,t_i\in P,t_{i+1}\in Q\dots,t_h\in Q}\chi_{t_1,\dots,t_h}}$, we have $\abs{\sum_{t_1\in P,t_2\in Q,\dots,t_h\in Q}\chi_{t_1,\dots,t_h}} \leq2^{h}h^{2h}M/h!\leq2^{2h^2}M$.
\end{proof}

We will show that for any $\chi\in \mathcal{C}_{\sigma,\gamma}$, 
we can find disjoint $P,Q\subset V$ such that 
\[
\abs{\chi_{1,h-1}(P,Q)} \geq \Omega_h(\sigma\cdot 
{\gamma^{1/2}n^{3/2}\LSKh}).
\]
This will then finish the proof of the lemma. 

We first prove the following useful claim. 
\begin{claim}\label{claim:hyperpartition}
Let $H=(W,E)$ be a graph such that for any fixed $B$ with $|B|=i-1$, the number of $t\in W$ such that $t$ is connected to all vertices in $B$ is at most $2n\gamma^{i-1}$. Let $S\cup T$ be a random partition of $W$, i.e., each vertex is independently assigned to each class with probability $1/2$. Then
\[
\mathbb{E}\left[\sum_{B\in S^{(i-1)}}\abs{\sum_{t\in T}\chi_{B,t}}\right]\geq i\cdot 2^{-i-3}\sum_{L\in W^{(i)}}\abs{\chi_L}/\sqrt{n\gamma^{i-1}}  
\]
\end{claim}
\begin{proof}
Let $\rho_v\in \{0,1\}$ be i.i.d. Bernoulli for each $v$. Note that for any given $B\in S^{(i-1)}$, by Lemma~\ref{lem:simpleinequality}, it holds that 
\[
\mathbb{E}\left[\abs{\sum_{t\in T\setminus B}\chi_{B,t}}\right]=\mathbb{E}\left[\abs{\sum_{t\in W\setminus B}\rho_t\chi_{B,t}}\right]\geq \frac{\sum_{t\in W\setminus B}\abs{\chi_{B,t}}}{\sqrt{8\cdot 2n\gamma^{i-1}}}=\frac{\sum_{L\in W^{(i)}}\abs{\chi_{L}}}{\sqrt{16n\gamma^{i-1}}},
\]
where we make use of the fact that for any fixed $B\in S^{(i-1)}$, the number of $t$ such that $\chi_{B,t}$ is non-zero is at most $2n\gamma^{i-1}$, as such a $t$ should at least connect to all vertices in $B$. 
Since the event that $B\subset S$ and the random variable 
$\sum_{t\in W\setminus B}\abs{\chi_{B,t}}$ are independent, and each $L\in W^{(i)}$ occurs $i$ times as $L=B\cup \{t\}$, we have that 
\begin{align*}
\mathbb{E}\left[\sum_{B\in S^{(i-1)}}\abs{\sum_{t\in T}\chi_{B,t}}\right]&=\sum_{B\in W^{(i-1)}}\mathsf{Pr}[B\subset S]\cdot \mathbb{E}\left[\abs{\sum_{t\in T}\chi_{B,t}}\right]\\
&\geq \sum_{B\in W^{(i-1)}}2^{-i+1}\frac{\sum_{t\in W\setminus B}
\abs{\chi_{B,t}}}{\sqrt{16n\gamma^{i-1}}}\\
&= i2^{-i-1}\sum_{L\in W^{(i)}}\abs{\chi_L}/\sqrt{16 n\gamma^{i-1}}    
\end{align*}
\end{proof}

Consider $V = S_h \cup T_{h-1}$ as a random partition, and let $T_{h-1} = S_{h-1} \cup T_{h-2}, \dots, T_2 = S_2 \cup T_1$ represent random bipartitions. In each bipartition, every vertex is independently assigned to either vertex class with a probability of $1/2$. For each $i\leq h-1$ and a subset $R\subseteq T_i$ with $i$ distinct vertices (i.e., $R\in T_i^{(i)}$), we define
\[
y_i(R)=\sum_{R,s_{i+1}\in S_{i+1},\dots,s_{h}\in S_h}\chi_{R,s_{i+1},\dots,s_{h}}.
\]

Let $T_{h}=V$ and define $y_h=\chi$. Then for $1\leq i<h$ and $R\in T_i^{(i)}$, we have
\[
y_i(R)=\sum_{s_{i+1}\in S_{i+1},\dots,s_{h}\in S_h}\chi_{R,s_{i+1},\dots,s_{h}}=\sum_{s\in S_{i+1}}\sum_{s_{i+2},\dots,s_{h}} \chi_{R\cup\{s\},s_{i+2},\dots,s_{h}}=\sum_{s\in S_{i+1}} y_{i+1}(R\cup \{s\}).
\]

Now by our assumption on the properties of the graphs $G$ and $G'$, each edge $e=(u,v)\in E\triangle E'$, it holds that the $K_h$ instances containing $e$ are either in $G$ or in $G'$, but not both; and each edge belongs to at least $\frac{{\LSKh}
}{2}$ number of $K_{h}$ instances. Therefore,
\[
|y_2(\{u,v\})|=\sum_{s_3\in S_3,\dots,s_h\in S_h}|\chi_{u,v,s_3,\dots,s_h}|\geq  \frac{{\LSKh}
}{2}
\]
Note that by our definition of $T_2$, the probability that the two endpoints of any edge $e$ belong to $T_2$ is at least $\frac{1}{2^{h+1}}$. Thus, 
\[
\mathbb{E}[\sum_{R\in T_2^{(2)}}\abs{y_2(R)}]\geq 2^{-h-1}\cdot \sigma \gamma n^2 \cdot \frac{{\LSKh}
}{2} 
.
\]

By Claim~\ref{claim:hyperpartition}, we have that given $T_{i+1}$ and $y_{i+1}$,  
\[
\mathbb{E}[\sum_{R\in T_i^{(i)}}\abs{y_i(R)}]\geq (i+1)2^{-(i+4)}\sum_{L\in T_{i+1}^{(i+1)}}\abs{y_{i+1}(L)}/
{\sqrt{n\gamma^{i}}}
\]

Therefore, 
\begin{align*}
\mathbb{E}\left[\sum_{s\in T_1}\abs{\sum_{s_{2}\in S_{2},\dots,s_{h}\in S_h}\chi_{s,s_{2},\dots,s_{h}}}\right]&=\mathbb{E}[\sum_{s\in T_1}\abs{y_1(s)}]\\
& \geq 2^{-4}\cdot \E{\sum_{R\in T_{2}^{(2)}}\abs{y_2(R)}/\sqrt{n\gamma}}\\
& \geq 2^{-h-6}\cdot \sigma\cdot 
{\gamma n^2\cdot \LSKh} /\sqrt{n\gamma}\\
&\geq 2^{-h-6}\cdot\sigma \cdot \gamma^{1/2}n^{3/2}\cdot \LSKh
:= \Upsilon 
\end{align*}

Now we let $S_1^+=\{s\in T_1:\sum_{s_{2}\in S_{2},\dots,s_{h}\in S_h}\chi_{s,s_{2},\dots,s_{h}}>0\}$, and $S_1^-=T_1\setminus S_1^+$. Then we have 
\begin{align*}
\Upsilon  &\leq \E{\sum_{s\in T_1}\abs{\sum_{s_{2}\in S_{2},\dots,s_{h}\in S_h}\chi_{s,s_{2},\dots,s_{h}}}}\\
&= \E{\sum_{s\in S_1^+}\left(\sum_{s_{2}\in S_{2},\dots,s_{h}\in S_h}\chi_{s,s_{2},\dots,s_{h}}\right)}+\E{\sum_{s\in S_1^-}\left(-\sum_{s_{2}\in S_{2},\dots,s_{h}\in S_h}\chi_{s,s_{2},\dots,s_{h}}\right)}\\
&:=\E{\abs{\chi_{S_1^+,S_2,\dots,S_h}}} + \cdot \E{\abs{\chi_{S_1^-,S_2,\dots,S_h}}}
\end{align*}

Therefore at least one of $\E{\abs{\chi_{S_1^+,S_2,\dots,S_h}}},  \E{\abs{\chi_{S_1^-,S_2,\dots,S_h}}}$ is at least 
$\Upsilon/2$. Without loss of generality, we assume that $\E{\abs{\chi_{S_1^+,S_2,\dots,S_h}}}\geq \Upsilon/2$. 

Now for a nonempty set $\mathcal{I}\subseteq\{2,\dots,h\}$, we let $V_\mathcal{I}=\cup_{i\in \mathcal{I}}S_i$ and let 
\[\chi_\mathcal{I}=\sum_{s\in S_1^+, B\in V_\mathcal{I}^{(h-1)}, \abs{B\cap S_i}>0, \forall i\in \mathcal{I}} \chi_{s,B}.\]  
Note that 
\[\chi_{1,h-1}(S_1^+,V_\mathcal{I}) 
=\sum_{\emptyset\neq \mathcal{J}\subseteq \mathcal{I}}\chi_{\mathcal{J}}.\]
Note that $\chi_{\{2,\dots,h\}}=\chi_{S_1^+,S_2,\dots,S_h}$.

Now consider the family 
\[
\mathcal{F}=\{\mathcal{J}\subseteq \{2,\dots,h\}: \abs{\chi_{\mathcal{J}}}\geq (2h)^{-h+\abs{\mathcal{J}}}\cdot \abs{\chi_{S_1^+,S_2,\dots,S_h}}\}.
\]
Note that $\mathcal{F}$ is not empty, as the set $\{2,\dots,h\}$ belongs to $\mathcal{F}$. Now we let $\mathcal{J}_0$ be the set in $\mathcal{F}$ with the minimal size. Note that by definition, for all $\mathcal{I}\subsetneq \mathcal{J}_0$, it holds that $\abs{\chi_{\mathcal{I}}}< (2h)^{-h+\abs{\mathcal{I}}}\cdot \abs{\chi_{S_1^+,S_2,\dots,S_h}}$.

Therefore, 
\begin{align*}
\max_{\mathcal{J}\subseteq \{2,\dots,h\}} \abs{\chi_{1,h-1}(S_1^+,V_{\mathcal{J}})} 
& \geq \abs{\chi_{1,h-1}(S_1^+,V_{\mathcal{J}_0})}\\
&\geq \abs{\chi_{\mathcal{J}_0}}-\sum_{\emptyset\neq \mathcal{I}\subsetneq \mathcal{J}_0} \abs{\chi_{\mathcal{I}}}\\
&\geq \left((2h)^{-h+\abs{\mathcal{J}_0}}-\sum_{i=1}^{\abs{\mathcal{J}_0}-1}h^{\abs{\mathcal{J}_0}-i}(2h)^{-h+i}\right)  \abs{\chi_{S_1^+,S_2,\dots,S_h}}\\
&\geq  2^{-h^2}\cdot \abs{\chi_{S_1^+,S_2,\dots,S_h}}
\end{align*}

Thus, 
$\E{\max_{\mathcal{J}\subseteq \{2,\dots,h\}}\abs{\chi_{1,h-1}(S_1^+,V_{\mathcal{J}})}}\geq 2^{-h^2-1}\Upsilon$. Thus, there exists a subset $\mathcal{J}\subseteq \{2,\dots,h\}$ with 
\[
\E{\abs{\chi_{1,h-1}(S_1^+,V_{\mathcal{J}})}}\geq 2^{-h^2-1}\Upsilon.
\]

Now we choose $P=S_1^+$ and $Q=V_\mathcal{J}$ which achieve at least the expectation of $\E{\abs{\chi_{1,h-1}(S_1^+,V_{\mathcal{J}})}}$. Then it holds that
\[
\abs{{\chi_{1,h-1}(P,Q)}}\geq 2^{-h^2-1} \Upsilon =\Omega_h(\sigma\cdot \gamma^{1/2} n^{3/2}\ell_h(G)).
\]
By Lemma~\ref{lemma:disc_1toh_to_any}, this finishes the proof of the lemma. 

\end{proof}

We note that by the preconditions in Lemma~\ref{lemma:discrepancymain}, it holds that $m=\Theta(\gamma n^2)$. 
Thus, {$\gamma^{1/2}n^{3/2}=\Theta(\sqrt{\frac{m}{n^2}}n^{3/2})=\Theta(\sqrt{mn}$)}. This
quantity will be used in our lower bound for the discrepancy. 

The remaining proofs are almost the same as the corresponding proofs (with small changes) in \cite{eliavs2020differentially}. We present the proofs here for the sake of completeness. 
\begin{lemma}\label{lemma:dptodiscrepancy}
Let $\mathbf{x}$ be an indicator vector of the edge set of a graph $G=(V,E)$ which satisfies the following properties $(\star)$:
\begin{enumerate}
\item for each vertex $v\in V$, its degree belongs to the interval $[\gamma n/2,2\gamma n]$;
\item for each edge $e\in E$, the number of $K_h$-instances containing $e$ belongs to the interval
 $[\frac{\LSKh}{2},2\LSKh]$;
\item for $1\leq i\leq h$ and any subset $B$ with $i$ distinct vertices in $G$ and $G'$, the number of vertices $t$ such that $t$ is connected to all vertices in $B$ belongs to the interval $\left[\frac{n\gamma^{i}}{2}, 2n\gamma^{i}\right]$.
\end{enumerate}
Let $\mathcal{M}$ be a mechanism for the motif size of all cuts that outputs $\mathbf{y}$ with the input $x$, i.e., $\mathbf{y}=\mathcal{M}(\mathbf{x})$. Suppose that 
\[\norm{\mathbf{y}-\mathbf{A}\cdot \mathbf{x}_{K_h}}_\infty\leq \frac12\disc_{\mathcal{C}_{\sigma,\gamma}}(\mathbf{A}).
\]
Then there exists a deterministic algorithm $\mathcal{A}$ which given as input $\mathbf{y}$ and outputs a vector $\mathcal{A}(\mathbf{y})$ such that
\[
\norm{\mathcal{A}(\mathbf{y})-\mathbf{x}}_1\leq \sigma \gamma n^2.
\]
\end{lemma}
\begin{proof}
Let $\mathcal{A}$ simply be the algorithm that output an indicator vector $x'$ of any graph that satisfies the properties $(\star)$ and that 
\[
\norm{\mathbf{y}-\mathbf{A}\cdot \mathbf{x}_{K_h}'}_{\infty}<\frac12\disc_{\mathcal{C}_{\sigma,\gamma}}(\mathbf{A}).
\]
Note that such an $\mathbf{x}'$ exists, as $\mathbf{x}$ already satisfies the required properties. We consider the vector $\mathbf{x}'-\mathbf{x}$. Assume that $\norm{\mathbf{x}-\mathbf{x}'}_1> \sigma \gamma n^2$. Then $\chi=\mathbf{x}_{K_h}-\mathbf{x}_{K_h}'$ belongs to $\mathcal{C}_{\sigma,\gamma}$ and therefore
\[
\norm{\mathbf{A}\cdot (\mathbf{x}_{K_h}-\mathbf{x}_{K_h}')}_{\infty}\geq \disc_{\mathcal{C}_{\sigma,\gamma}}(\mathbf{A}).
\]

On the other hand, 
\[
\norm{\mathbf{A}\cdot (\mathbf{x}_{K_h}-\mathbf{x}_{K_h}')}_{\infty}\leq \norm{\mathbf{y}-\mathbf{A}\cdot \mathbf{x}_{K_h}}_{\infty} + \norm{\mathbf{y}-\mathbf{A}\cdot \mathbf{x}_{K_h}'}_{\infty}<\disc_{\mathcal{C}_{\sigma,\gamma}}(\mathbf{A}),
\]
which is a contradiction. Thus, it holds that $\norm{\mathbf{x}-\mathbf{x}'}_1\leq  \sigma \gamma n^2$.
\end{proof}

The remaining proof largely follows directly from \cite{eliavs2020differentially}, and we only need to make slight adjustments to adapt their proofs for our case. We state the proofs here for the sake of completeness. 

Let $X$ be the distribution of vectors $\mathbf{x}\in \{0,1\}^{\binom{n}{2}}$, where each coordinate is chosen independently such that $\mathbf{x}_i=1$ with probability $p$. Thus, the distribution $X$ is the distribution of indicator vectors of graphs $G\sim G(n,p)$, where $G(n,p)$ denotes the distribution of \Erdos-\Renyi random graphs. 

\begin{lemma}[\cite{eliavs2020differentially}]
Let $\mathcal{M}$ be an $(\varepsilon,\delta)$-differentially private mechanism and let $Y$ be the probability distribution over the transcripts of $\mathcal{M}(\mathbf{x})$, where $\mathbf{x}$ is drawn from distribution $X$. Then for any $\gamma>0$ and $\mathbf{x}\sim Y$, it holds that with probability $1-\delta'$ over $i\in [n]$ and $\mathbf{x}\gets X_{\mid Y=y}$, we have 
\[
2^{-\varepsilon-\gamma}\frac{1-p}{p}\leq \frac{\mathsf{Pr}_{\mathbf{x}\gets X_{\mid Y=\mathbf{y}}}[\mathbf{x}_i=0\mid \mathbf{x}_{-i}]}{\mathsf{Pr}_{\mathbf{x}\gets X_{\mid Y=\mathbf{y}}}[\mathbf{x}_i=1\mid \mathbf{x}_{-i}]} \leq 2^{\varepsilon+\gamma}\frac{1-p}{p},
\]
where $x_{-i}$ denotes the vector of all coordinates of $x$ excluding $x_i$. 
\end{lemma}

Now we prove the lower bound for $\varepsilon=1$. 

\begin{lemma}\label{lemma:lowereps1}
Let $h\geq 2$ be a constant. Let $G\sim G(n,p)$, where $(\frac{\log n}{n})^{1/(h-1)}\ll p\leq \frac12$, be a random graph and let $\mathcal{M}$ be a $(1,\delta)$-DP mechanism which approximates the $K_h$-motif size of all $(S,T)$-cuts of $G$ up to additive error $\alpha$ with probability $\beta$. Then $\alpha\geq \Omega(\disc_{\mathcal{C}_{\sigma,\gamma}}(\mathbf{A}))$, where $\gamma=p, \sigma=\Omega(1-\frac{9\delta}{\beta})$. 
\end{lemma}
\begin{proof}[Proof Sketch]
We choose $\varepsilon=1$ and $\varepsilon'=\varepsilon +10$. This implies $\delta'=2\delta\cdot \frac{1+e^{-\varepsilon-10}}{1-e^{-10}}\leq 3\delta$, then with probability $1-\delta'$ over over $i\in [n]$ and $\mathbf{x}\gets X_{\mid Y=\mathbf{y}}$, we have 
\begin{eqnarray}
2^{-\varepsilon'}\frac{1-p}{p}\leq \frac{\mathsf{Pr}_{\mathbf{x}\gets X_{\mid Y=\mathbf{y}}}[\mathbf{x}_i=0\mid \mathbf{x}_{-i}]}{\mathsf{Pr}_{\mathbf{x}\gets X_{\mid Y=\mathbf{y}}}[\mathbf{x}_i=1\mid \mathbf{x}_{-i}]} \leq 2^{\varepsilon'}\frac{1-p}{p}.
\label{ineq:indistinguish}
\end{eqnarray}

Then we can prove the lemma by contradiction. That is, we assume that $\mathcal{M}$ has additive error smaller than $\disc_{\mathcal{C}_{\sigma,\gamma}}(\mathbf{A})/2-1$ with probability at least $\beta$. Then we can show that for each possible output $\mathbf{y}$ of the mechanism $\mathcal{M}$, with probability greater than $\delta'$, the inequality (\ref{ineq:indistinguish}) is violated. To do so, we only need to show that 1) with high probability, $\mathbf{x}$ is \emph{good} in the sense that it satisfies the desired properties; and 2) conditioned on the event that $\mathbf{x}$ is good, Equation (\ref{ineq:indistinguish}) is violated with probability greater than $\sigma'$, if we set $\gamma=p$ and $\sigma=2^{-13}\cdot (1-\frac{3\delta'}{\beta})$, which leads to a contradiction. 

Part 2) follows from the same argument as those in  the proof of Lemma 5.3 in \cite{eliavs2020differentially}. For part 1), we describe our changes. Formally, we say that $\mathbf{x}\sim X_{\mid Y=\mathbf{y}}$ is \emph{good}, if $\norm{\mathbf{A}\mathbf{x}_{K_h}-\mathbf{y}}_\infty\leq \disc_{\mathcal{C}_{\sigma,\gamma}}(\mathbf{A})/2-1$ and the properties $(\star)$ given in the statement of Lemma~\ref{lemma:dptodiscrepancy} are satisfied. 

Note that the property that $\mathbf{x}\sim X_{\mid Y=\mathbf{y}}$ is {good} is at least $(1-\frac{1}{\poly(n)})\cdot \beta$. This is true as our assumption that $G\sim G(n,p)$ and that $(\frac{\log n}{n})^{1/(h-1)}\ll p\leq \frac12$, it holds that $G$ satisfies the properties $(\star)$ with probability at least $1-\frac{1}{\poly(n)}$. This then finishes the proof of the lemma. 
\end{proof}

Finally, our lower bound \Cref{thm:lowerfinal} follows from the following lemma from \cite{eliavs2020differentially}, which in turn is built from Lemma 2.1.2 in \cite{bun2016new}. 

\begin{lemma}[\cite{eliavs2020differentially}]
\label{lem:lowerbound_1_to_epsilon}
If there is no $(1,\delta)$-DP mechanism whose error is below $o\left(\overline{T}_G\cdot (1-\frac{9\delta}{\beta})\right)$, where  $\overline{T}_G=\sqrt{mn}\cdot\LSKh$, with probability $\beta$. Then there is no $(\varepsilon,\delta)$-DP mechanism whose error is smaller than $o\left(\frac{\overline{T}_G}{\varepsilon}\cdot(1-c)\right)$ with probability $\beta$, where $c=\frac{e-1}{e^{\varepsilon}-1}\cdot \frac{9\delta}{\beta}$. 
\end{lemma}
\begin{proof}
From Lemma 2.1.2 in \cite{bun2016new}, it holds that for any $(\varepsilon,\delta)$-DP mechanism $\mathcal{M}$, if two graphs $\mathbf{w},\mathbf{w}'$ satisfy that $\norm{\mathbf{w}-\mathbf{w}'}_1\leq \frac{1}{\varepsilon}$, then for any output set $S$, $\mathsf{Pr}[\mathcal{M}(\mathbf{w})\in S]\leq e\mathsf{Pr}[\mathcal{M}(\mathbf{w}')\in S]+\frac{e-1}{e^{\varepsilon}-1}\delta$. 

Now suppose there exists an $(\varepsilon,\delta)$-DP mechanism $\mathcal{M}$ whose error is smaller than $o\left(\frac{\overline{T}_G}{\varepsilon}(1-c)\right)$ with probability $\beta$, where $c=\frac{e-1}{e^{\varepsilon}-1}\cdot \frac{9\delta}{\beta}$. 
Since $\ell_h(\frac{1}{\varepsilon}\cdot\mathbf{w})=\frac{1}{\varepsilon^{\binom{h}{2}-1}}\ell_h(\mathbf{w})$, then $\overline{T}_{\frac{1}{\varepsilon}\cdot\mathbf{w}}=\frac{1}{\varepsilon^{\binom{h}{2}-1}}\overline{T}_{\mathbf{w}}$. Therefore, the mechanism $\varepsilon^{{\binom{h}{2}}}\cdot \mathcal{M}(\frac{1}{\varepsilon}\cdot \mathbf{w})$ is $(1,\frac{e-1}{e^{\varepsilon}-1}\delta)$-DP with additive error
\[
\varepsilon^{{\binom{h}{2}}}\cdot o\left(\frac{\overline{T}_{\frac{1}{\varepsilon}\cdot\mathbf{w}}}{\varepsilon}(1-c)\right) \leq o\left(\overline{T}_{\mathbf{w}}\cdot (1-\frac{e-1}{e^{\varepsilon}-1}\cdot \frac{9\delta}{\beta})\right)
\]
with probability at least $\beta$, which is a contradiction. 
\end{proof}

\begin{proof}[Proof of \Cref{thm:lowerfinal}]
By Lemma~\ref{lemma:discrepancymain}, Lemma~\ref{lemma:lowereps1} and Lemma~\ref{lem:lowerbound_1_to_epsilon}, there is no $(\varepsilon,\delta)$-DP mechanism that can release synthetic graphs preserving $K_h$-motif cuts for graph $G$ with an additive error smaller than $o\left(\frac{\sqrt{mn}\cdot\LSKh}{\varepsilon}\cdot(1-c)\right)$ with probability $\beta$, where $c=\frac{e-1}{e^{\varepsilon}-1}\cdot \frac{9\delta}{\beta}$.

Now consider the scaled version of $G$ described in \Cref{thm:lowerfinal}, where the weight vector $\mathbf{w}$ of $G$ is scaled by a factor of $\frac{1}{\varepsilon}$, resulting in a new weight vector $\frac{1}{\varepsilon} \cdot \mathbf{w}$. Note that the total weight is $W=\frac{m}{\varepsilon}$. When we let $\overline{T}_{\mathbf{w}}=\sqrt{Wn}\cdot\LSKh$, then $\overline{T}_{\frac{1}{\varepsilon}\cdot\mathbf{w}}=\frac{1}{\varepsilon^{\binom{h}{2}-\frac{1}{2}}}\overline{T}_{\mathbf{w}}$. By a similar analysis to Lemma~\ref{lem:lowerbound_1_to_epsilon}, we can establish the bound 
\[\Omega\left(\max\left(\frac{\sqrt{mn}\cdot\LSKh}{\varepsilon}\cdot(1-c),\frac{\sqrt{Wn}\cdot\LSKh}{\sqrt{\varepsilon}}\cdot(1-c)\right)\right).\]
\end{proof}

According to the above proof, we can also obtain a lower bound of $\Omega\left(\sqrt{\frac{mn}{\varepsilon}} \cdot \ell_3(G) \cdot (1 - c)\right)$ 
for unweighted graphs that may contain multiple edges. This holds because we can interpret a scaled version of $G$ as a multigraph, where an edge with weight $1/\varepsilon$ is treated as having multiplicity $1/\varepsilon$.

\section{Conclusion}
In this paper, we present the first polynomial-time algorithm for releasing a synthetic graph that effectively preserves the triangle-motif cut structure of an input graph in a differentially private manner. This algorithm extends previous studies (\cite{gupta2012iterative,blocki2012johnson,upadhyay2013random,arora2019differentially,eliavs2020differentially,liu2024optimal}) on differentially private algorithms that maintain edge-motif cut structures to higher-order organizations. This higher-order property has wide applications in analyzing complex networks \cite{milo2002network,benson2016higher} and has garnered increasing attention in the theoretical computer science community \cite{kapralov2022motif}. We also establish a lower bound of the additive error for DP algorithms that answers the $K_h$-motif cut queries. 

Our work leaves several interesting open questions. One immediate question is to develop nearly matching upper and lower bounds for the entire class of graphs, which would likely require new ideas. 
Another interesting direction is to give a differentially private algorithm for weighted graphs whose additive error depends solely on the number of edges rather than the maximum or total edge weights. This has recently been achieved for edge-motif cut structures using a topology sampler and leveraging the linearity property of these structures \cite{liu2024optimal}. However, the inherent non-linearity of triangle-motif cut structures presents a significant challenge in extending this approach. Finally, it would be interesting to investigate differentially private algorithms for preserving motif cut structures beyond triangles (or $3$-vertex motifs). We believe that our method, combined with optimization techniques for tensors and hypergraphs, could be useful in generating synthetic graphs that preserve $K_h$-motif cut structures for any constant $h$. 

\section*{Acknowledgements}
We thank Jalaj Upadhyay for helpful discussions and the anonymous reviewers for their valuable comments on an earlier version of this paper.

\bibliographystyle{alpha}
\bibliography{main_arxiv}

\appendix

\section{More Discussions}
\subsection{Synthetic Graphs with Multiplicative Errors and Interactive Solutions} 
We note that if one allows for exponential time and multiplicative error,  
then one can achieve a better additive error. This is similar to the edge motif cut case. The reason is as follows: It is known that there exists a polynomial-time algorithm that constructs a motif cut sparsifier with only $\tilde{O}(n/\eta^2)$ edges for any $\eta>0$ (see \cite{kapralov2022motif}). This sparsifier ensures that for every cut $(S,V\setminus S)$, the weighted count of copies of motif $M$ crossing the cut in $G'$ is within a $1+\eta$ factor of the number of copies of $M$ crossing the same cut in $G$.

Given the existence of the above sparsifier, we can apply the exponential mechanism and restrict its range to every potential output graph with $\tilde{O}(n)$ edges (for any constant $\eta$). Moreover, we can use the maximum motif weight cut error as the scoring function. One can then show that the exponential mechanism enables the release of a synthetic graph $G'$ where each motif weight cut of $G$ is approximated within an expected additive error of $\tilde{O}(n^2)$ 
and a multiplicative error of $(1+\eta)$ in expectation. However, the main drawback of this approach lies in its exponential time complexity.

Similar to the synthetic graph for edge counts of cuts \cite{eliavs2020differentially}, the synthetic graph released by our algorithm is not necessarily sparse, i.e., it may not have $\tilde{O}(n)$ edges. If necessary, one can indeed sparsify the output of our algorithm using the motif cut sparsification algorithm given by Kapralov et al. 
\cite{kapralov2022motif} to obtain a DP sparsifier, leveraging the post-processing property of differential privacy. However, this will introduce multiplicative errors.

Finally, we note that an interactive solution for preserving privacy in the motif cut structure could also be considered (see \Cref{sec:relatedwork} for  considerations related to edge cuts). In this scenario, data analysts could specify any cut $(S,V\setminus S)$ with the aim of determining (with some acceptable error) the number of triangles or motifs connecting the two groups while maintaining privacy. However, our primary focus is on a stronger, non-interactive solution, i.e. to release a private synthetic dataset: a new, private graph that approximately preserves the motif cut function of the original graph. 

\subsection{Some Tempting Approaches That Do Not Work} 
\label{appendix:discuss:tempting_methods}
There are several natural approaches that may seem promising for privately releasing synthetic graphs for triangle-motif cuts, which we outline below and briefly explain why they do not work.

One initial approach could involve utilizing the motif cut sparsifier algorithm proposed by Kapralov et al. 
\cite{kapralov2022motif}, followed by incorporating a noise addition mechanism into the sparsification process. However, similar to the challenges faced in the edge cut case, making such sparsification algorithms differentially private is highly challenging. This difficulty arises primarily because the algorithm relies on importance sampling of existing edges and never outputs non-edges, which poses a fundamental obstacle for DP algorithms. 
We refer to the discussion in 
\cite{eliavs2020differentially} for a more detailed explanation of the limitations associated with this approach.

Another natural approach is as follows: One can first convert the original graph to a (triangle) hypergraph by creating a hyperedge for each triangle, and then attempt to apply DP hyperedge cut release algorithms. However, to the best of our knowledge, there is no DP algorithm for releasing a synthetic hypergraph while preserving the hyperedge cuts. Furthermore, even if such an algorithm existed, it would not fully address our problem because we cannot convert a hypergraph back into a graph. In fact, it is possible for two graphs $G$ and $G'$ to be very different from each other, while their corresponding hypergraphs are identical. For example, if both $G$ and $G'$ are triangle-free and far from each other, their hypergraphs will be the same, i.e., the hypergraph with no hyperedge at all.

The third approach is to make use of the triangle-motif-weighted graph associated with the input graph \cite{benson2016higher}. That is, one can first  convert the original graph ${G}$ into a triangle-motif weighted graph ${G}_{\triangle}$ with a weight vector $\mathbf{{w}}_{\triangle}$, where $\mathbf{{w}}_{\triangle}(e)$ denotes the sum of the weights of triangles containing the endpoints of edge $e$ simultaneously. For triangle-motif weighted graphs, there is a useful property: the size of the edge cut $(S,V\setminus S)$ in ${G}_{\triangle}$ is exactly twice the sum of the weights of triangles crossing the cut in $G$. Therefore, a naive approach would be to use a private edge cut release algorithm, such as the one proposed in \cite{liu2024optimal}, on the triangle-motif weighted graph. However, after applying the DP algorithm on ${G}_{\triangle}$, the resulting graph $H$ may not correspond to a triangle-motif weighted graph, i.e. there does not exist a graph whose triangle-motif weighted graph is $H$. As a result, through this approach, we can only privately release the values of triangle-motif cuts but cannot release a synthetic graph. 

\section{Deferred Preprocessing Step in \Cref{s.alg}}
\label{appendix:alg:preprocessing}
We use $W$ and $\mathbf{u}$
to denote the differentially privately released approximations of the sum of edge weights and the upper bound of each edge weight of $\hat{G}$ respectively. Specifically, we do the following. Recall $\varepsilon$ is the parameter for $(\varepsilon,\delta)$-DP, and $\beta$ is some parameter for the success probability.
\begin{enumerate}
\label{def.parameters}
    \item Set $W=\sum_{e\in\binom{V}{2}}\hat{\mathbf{w}}_e+\mathrm{Lap}({1/\varepsilon_1})+\log(3/\beta)/\varepsilon_1$. The term $\log(3/\beta)/\varepsilon_1$ is 
    to guarantee $W\geq\sum_{e\in\binom{V}{2}}\hat{\mathbf{w}}_e$ with high probability;
\item Normalize the weights of $\hat{G}$ to obtain a graph $\overline{G}=(V,\overline{E},\overline{\mathbf{w}})$ with the same vertex and edge sets as $\hat{G}$ (i.e., $\overline{E}=\hat{E}$), while the edge weights $\overline{G}$ of sum up to $W$. That is, $\overline{\mathbf{A}}=(W/\hat{W})\hat{\mathbf{A}}$, where $\overline{\mathbf{A}}$ and $\hat{\mathbf{A}}$ denote adjacency matrices of $\overline{G}$ and $\hat{G}$, respectively; 
\item Set 
$\mathbf{u}_e=\overline{\mathbf{w}}_e+\mathrm{Lap}(1/\varepsilon_2)+\log(6n^2/\beta)/\varepsilon_2+\frac{W}{\binom{n}{2}}$. The term $\log(6n^2/\beta)/\varepsilon_2$ is to guarantee $\mathbf{u}\geq \overline{\mathbf{w}}_e+\frac{W}{\binom{n}{2}} $ with high probability, since we need to guarantee that $\overline{\mathbf{w}}$ and $(\frac{W}{\binom{n}{2}})_{e\in\binom{V}{2}}$ fall in the domain $\mathcal{X}$ in update step.
\item Set $u_{\max}=\max_{e\in\binom{V}{2}}\mathbf{u}_e$. Recall that $\hatLSKt=\max_{(i,j)\in\binom{V}{2}}\sum_{s\in V\setminus \lbrace i,j\rbrace}
\mathbf{\hat{w}}_{(i,s)}\mathbf{\hat{w}}_{(j,s)}$ denotes the local sensitivity of triangle-motif cuts of $\hat{G}$. We further define $\tildeLSKt=\hatLSKt+u_{\max}(\mathrm{Lap}(1/\varepsilon_3)+\log(6n^2/\beta)/\varepsilon_3)$.
\end{enumerate}
By Lemma~\ref{l.laplace} and Lemma~\ref{l.post_processing}, the released $W$, $\mathbf{u}_e$ and $\tildeLSKt$ are $\varepsilon_1,\varepsilon_2,\varepsilon_3$-differentially private respectively. 
By Lemma~\ref{l.lap_bound} and union bound, with probability at least $1-\frac{\beta}{3}-\binom{n}{2}\frac{\beta}{3n^2}\geq 1-2\beta/3$, it holds that\label{fail_prop}
\[
\overline{\mathbf{w}}_e+\frac{W}{\binom{n}{2}}\leq\mathbf{u}_e\leq\overline{\mathbf{w}}_e+\frac{W}{\binom{n}{2}}+2\log(6n^2/\beta)/\varepsilon_1
\]
for any $e\in\binom{V}{2}$ and that \[\sum_{e\in\binom{V}{2}}\hat{\mathbf{w}}_e\leq W\leq\sum_{e\in\binom{V}{2}}\hat{\mathbf{w}}_e+2\log(3/\beta)/\varepsilon_2.\] 
and that
\[\hatLSKt\leq\tildeLSKt\leq \hatLSKt+2u_{\max}\log(6n^2/\beta)/\varepsilon_3\]

For the sake of convenience, we further introduce some quantities used in our algorithm and analysis:
\begin{enumerate}
\label{def.wu_triangle_Lamda}
    \item ${U}_\triangle=\max_{(i,j)\in\binom{V}{2}} 
\sum_{s\in V\setminus \lbrace i,j\rbrace}(\mathbf{u}_{(i,j)}\mathbf{u}_{(i,s)}
+\mathbf{u}_{(i,s)}\mathbf{u}_{(j,s)}
+\mathbf{u}_{(j,s)}\mathbf{u}_{(i,j)})$,
    \item ${U}_\Lambda=\max_{(i,j)\in\binom{V}{2}} 
\sum_{s\in V\setminus \lbrace i,j\rbrace}(\mathbf{u}_{(i,s)}
+\mathbf{u}_{(j,s)})$. 
\end{enumerate}
Note that ${U}_\triangle,{U}_\Lambda$ can be viewed as the maximum pairwise triangle importance and wedge importance, respectively, in the graph with edge weights given by $\mathbf{u}$. 

Assume that $\hat{W}=\Omega(\frac{1/\beta}{\varepsilon})$ and $\hatLSKt=\Omega(\frac{w_{\max}\log^2(n/\beta)}{\varepsilon^2})$. We refer to this as the \emph{non-degenerate} case. (When this assumption does not hold, we call it the \emph{degenerate} case. We will provide a case analysis in Corollary~\ref{cormainadditive}.) Recall that $\barLSKt=\max_{(i,j)\in\binom{V}{2}}\sum_{s\in V\setminus \lbrace i,j\rbrace}
\overline{\mathbf{w}}_{(i,s)}\overline{\mathbf{w}}_{(j,s)}$ is the local sensitivity of triangle cuts in $\overline{G}$. Since $\hat{W}=\Omega(\frac{1/\beta}{\varepsilon})$, we have $\barLSKt=(\frac{\overline{W}}{\hat{W}})^2\hatLSKt=\Theta(\hatLSKt)$. 
Additionally, by the assumption $\hatLSKt=\Omega(\frac{w_{\max}\log^2(n/\beta)}{\varepsilon^2})$ and the fact $\tildeLSKt\leq \hatLSKt+2u_{\max}\log(6n^2/\beta)/\varepsilon_3\leq\hatLSKt+O(w_{\max}\log^2(n/\beta)/\varepsilon_2\varepsilon_3)$, $\varepsilon_2=\varepsilon_3=\varepsilon/6$, it holds that $\tildeLSKt=\Theta(\hatLSKt)$.

\section{Upper Bound of Randomized Response Method}\label{appendix:randomresponse}
In this section, we show that by using the randomized response, we can obtain an $(\varepsilon,0)$-DP algorithm releasing a synthetic graph for triangle-motif cut queries. Specifically, given a graph $G$ with vertex set $V$ of size $n$ and weight vector $\mathbf{w}$, we release the noisy weight vector $\tilde{\mathbf{w}}$ where $\tilde{\mathbf{w}}_e=\mathbf{w}_e+Y_e$ and each random variable $Y_e\sim\lap(1/\varepsilon)$. Using the following analysis similar to \cite{gupta2012iterative}, we show that this algorithm achieves $\tilde{O}(n^{5/2})$ additive error for unweighted graph. Note that our bound of $\tilde{O}(\sqrt{m\LSKt}n)$ is often much better than  $\tilde{O}(n^{\frac{5}{2}})$.

Denote the query set $\mathcal{Q}_{\mathrm{cuts}}$ as a collection containing all the triangle-motif cut queries on $G$. Since $\mathcal{Q}_{\mathrm{cuts}}$ consists of all the $(S,T)$ pairs, it has size at most $2^{2n}$. 
Note that the queries in $Q_{cuts}$ are not linear, unlike those in the edge cut setting discussed in \cite{gupta2012iterative}. Consequently, the Chernoff bound alone cannot be used to address the entire problem. In the following, we will see how Azuma's Inequality can be applied to overcome this challenge.

Note that each query $q$ in $\mathcal{Q}_{\mathrm{cuts}}$ can be view as a vector in $\lbrace0,1\rbrace^{\binom{\abs{V}}{3}}$, namely, \[q(\mathbf{w})=\sum_{(i,j,k)\in\binom{V}{3}}q_{(i,j,k)}\mathbf{w}_{(i,j)}\mathbf{w}_{(j,k)}\mathbf{w}_{(k,i)}.\] For convenience, we denote $\pi(i,j,k)=\lbrace (i,j,k),(j,k,i),(k,i,j)\rbrace$. Thus, for any $q\in\mathcal{Q}_{\mathrm{cuts}}$, we have, 
\begin{align}
    &q(\tilde{\mathbf{w}})-q(\mathbf{w})\nonumber \\
    =&\sum_{(i,j,k)\in\binom{V}{3}}q_{(i,j,k)}\left((\mathbf{w}_{(i,j)}+Y_{(i,j)})(\mathbf{w}_{(j,k)}+Y_{(j,k)})(\mathbf{w}_{(k,i)}+Y_{(k,i)})-\mathbf{w}_{(i,j)}\mathbf{w}_{(j,k)}\mathbf{w}_{(k,i)}\right)\nonumber \\
    =&\sum_{(i,j,k)\in\binom{V}{3}}q_{(i,j,k)}\left(Y_{(i,j)}Y_{(j,k)}Y_{(k,i)}
    +\sum_{(i_1,i_2,i_3)\in \pi(i,j,k)} \left(Y_{(i_1,i_2)}\mathbf{w}_{(i_2,i_3)}\mathbf{w}_{(i_3,i_1)}
    + Y_{(i_1,i_2)}Y_{(i_2,i_3)}\mathbf{w}_{(i_3,i_1)}\right)\right)\nonumber \\
    \leq& w_{\max}^2\sum_{(i,j,k)\in\binom{V}{3}}q_{(i,j,k)}(Y_{(i,j)}+Y_{(j,k)}+Y_{(k,i)})
    +w_{\max}\sum_{(i,j,k)\in\binom{V}{3}}q_{(i,j,k)}\sum_{(i_1,i_2,i_3)\in \pi(i,j,k)}Y_{(i_1,i_2)}Y_{(i_2,i_3)}\nonumber \\
    &+\sum_{(i,j,k)\in\binom{V}{3}}q_{(i,j,k)}Y_{(i,j)}Y_{(j,k)}Y_{(k,i)}\label{eqn:qcut}
\end{align}

By Lemma~\ref{l.lap_bound}, with probability at least $1-\beta/4$, we have that each of the absolute values $\abs{Y_{(i,j)}}$'s is at most $L=O(1/\varepsilon\cdot\log(n/\beta))$. We then bound the three terms above conditioning on this event happening:

The first term $w_{\max}^2\sum_{(i,j,k)\in\binom{V}{3}}q_{(i,j,k)}(Y_{(i,j)}+Y_{(j,k)}+Y_{(k,i)})$ in \Cref{eqn:qcut} can be bounded by using a Chernoff bound. Specifically, since $\sum_{(i,j,k)\in\binom{V}{3}}q_{(i,j,k)}(Y_{(i,j)}+Y_{(j,k)}+Y_{(k,i)})=\sum_{(i,j)\in\binom{V}{2}}\sum_{k\in V}q_{(i,j,k)}Y_{(i,j)}$, it holds that $\sum_{(i,j)\in\binom{V}{2}}\sum_{k\in V}q_{(i,j,k)}Y_{(i,j)}$ is the sum of $\binom{n}{2}$-many independent random variables with mean $0$ and bounded by $[-nL,nL]$. Thus, by Chernoff bound, 
\[\Pr{\abs{\sum_{(i,j)\in\binom{V}{2}}\sum_{k\in V}q_{(i,j,k)}Y_{(i,j)}}\geq\alpha}\leq 2e^{\Omega(\frac{\alpha^2}{(nL)^2\binom{n}{2}})}=2e^{\Omega(\frac{\alpha^2}{n^4L^2})}.\] 

Therefore, by choosing $\alpha=O\left(\varepsilon^{-1}n^2\log(n/\beta)\sqrt{\log(\abs{\mathcal{Q}_{\mathrm{cut}}}/\beta)}\right)$, we have that, with probability at least $1-\beta/4\abs{\mathcal{Q}_{\mathrm{cut}}}$, \[\abs{\sum_{(i,j)\in\binom{V}{2}}\sum_{k\in V}q_{(i,j,k)}Y_{(i,j)}}\leq O\left(\varepsilon^{-1}n^2\log(n/\beta)\sqrt{\log(\abs{\mathcal{Q}_{\mathrm{cut}}}/\beta)}\right).\] By a union bound, the probability of large deviations for any triangle-motif cut query is at most $\beta/4$.

The method to handle with the other two terms is slightly different from the analysis in \cite{gupta2012iterative}. Since they cannot be written as a sum of independent random variables, we cannot use Chernoff bounds. To deal with the independency, we construct martingales and use Azuma's inequality instead.

\begin{lemma}[Azuma's Inequality, \cite{azuma1967weighted}]
    Let $X_0,\dots,X_t$ be a martingale satisfying $\abs{X_{i}-X_{i-1}}\leq c_i$ for any $i\in[t]$. Then for any $\alpha>0$, 

    \[\Pr{\abs{X_t-X_0}\geq\alpha}\leq2\exp\left(\frac{-\alpha^2}{2(\sum_{i=1}^t c_i^2)}\right)\]

\end{lemma}

We first consider the second term in \Cref{eqn:qcut}. The martingale considered is as follows: Given an arbitrary order of the $\binom{n}{2}$ pairs, namely, $e_1,\dots,e_{\binom{n}{2}}$. For the sake of simplicity, we denote $\mathcal{E}_i$ as $\lbrace e_1,e_2,\dots,e_i\rbrace$. Then we let $X_0=0$, and let 
\[X_{i}=X_{i-1}+\sum_{(i_1,i_3)\in\mathcal{E}_{i-1}}q_{(i_1,i_2,i_3)}Y_{(i_1,i_2)}Y_{(i_1,i_3)}+\sum_{(i_2,i_3)\in\mathcal{E}_{i-1}}q_{(i_1,i_2,i_3)}Y_{(i_1,i_2)}Y_{(i_2,i_3)},\]
where we assume $e_i=(i_1,i_2)$. That is, intuitively, at time $i$, we add the pair $(i_1,i_2)$ into $\mathcal{E}_{i-1}$, and add all the terms related to $(i_1,i_2)$ between the existing pairs in $\mathcal{E}_{i-1}$ into $X_{i-1}$, similar to edge-exposure martingale. It is clear that $X_{\binom{n}{2}}=\sum_{(i,j,k)\in\binom{V}{3}}q_{(i,j,k)}\sum_{(i_1,i_2,i_3)\in \pi(i,j,k)}Y_{(i_1,i_2)}Y_{(i_2,i_3)}$, and $X_t$ is a martingale since 
\begin{align*}
&\E{X_i\vert X_0,X_1,\dots,X_{i-1}}\\
=&\E{X_i\vert e_1,e_2,\dots,e_{i-1}}\\
=&X_{i-1}+\E{Y_{e_i}\vert Y_{e_1},Y_{e_2},\dots,Y_{e_{i-1}}}\cdot \left(\sum_{(i_1,i_3)\in\mathcal{E}_{i-1}}q_{(i_1,i_2,i_3)}Y_{(i_1,i_3)}+\sum_{(i_2,i_3)\in\mathcal{E}_{i-1}}q_{(i_1,i_2,i_3)}Y_{(i_2,i_3)}\right)\\
=&X_{i-1}.
\end{align*}

Then since we can bound $\abs{X_{i}-X_{i-1}}$ by $2nL^2$, by Azuma's Inequality, we have,
\[\Pr{\abs{X_{\binom{n}{2}}-X_0}\geq\alpha}\leq2\exp\left(\frac{-\alpha^2}{2\binom{n}{2}(2nL^2)^2}\right)=2\exp\left(\frac{-\alpha^2}{8\binom{n}{2}n^2L^4}\right).\]

By choosing $\alpha=O\left(\varepsilon^{-2}n^2\log^2(n/\beta)\sqrt{\log(\abs{\mathcal{Q}_{\mathrm{cut}}}/\beta)}\right)$, we have that, with probability at least $1-\beta/4\abs{\mathcal{Q}_{\mathrm{cut}}}$, 
\[\abs{\sum_{(i,j,k)\in\binom{V}{3}}q_{(i,j,k)}\sum_{(i_1,i_2,i_3)\in \pi(i,j,k)}Y_{(i_1,i_2)}Y_{(i_2,i_3)}}\leq O\left(\varepsilon^{-2}n^2\log^2(n/\beta)\sqrt{\log(\abs{\mathcal{Q}_{\mathrm{cut}}}/\beta)}\right).\] By a union bound, the probability of large deviations for any triangle-motif cut query is at most $\beta/4$.

For the third term in \Cref{eqn:qcut}, we can define the martingale as: 
\[X_{i}=X_{i-1}+\sum_{(i_1,i_3),(i_2,i_3)\in\mathcal{E}_{i-1}}q_{(i_1,i_2,i_3)}Y_{(i_1,i_2)}Y_{(i_2,i_3)}Y_{(i_3,i_1)}.\] 

By a similar analysis, we have,
\[\abs{\sum_{(i,j,k)\in\binom{V}{3}}q_{(i,j,k)}Y_{(i,j)}Y_{(j,k)}Y_{(k,i)}}\leq O\left(\varepsilon^{-3}n^2\log^3(n/\beta)\sqrt{\log(\abs{\mathcal{Q}_{\mathrm{cut}}}/\beta)}\right)\]for any triangle-motif cut query with probability at least $1-\beta/4$.

To sum up, with probability at least $1-\beta$, 
 for any triangle-motif cut query $q\in\mathcal{Q}_{cuts}$, we have,
\begin{align*}
    &q(\tilde{\mathbf{w}})-q(\mathbf{w})\\
    \leq& w_{\max}^2\sum_{(i,j,k)\in\binom{V}{3}}q_{(i,j,k)}(Y_{(i,j)}+Y_{(j,k)}+Y_{(k,i)})
    +w_{\max}\sum_{(i,j,k)\in\binom{V}{3}}q_{(i,j,k)}\sum_{(i_1,i_2,i_3)\in \pi(i,j,k)}Y_{(i_1,i_2)}Y_{(i_2,i_3)}\\
    &+\sum_{(i,j,k)\in\binom{V}{3}}q_{(i,j,k)}Y_{(i,j)}Y_{(j,k)}Y_{(k,i)}\\
    \leq& O\left(w_{\max}^2\varepsilon^{-1}n^2\log(n/\beta)\sqrt{\log(\abs{\mathcal{Q}_{\mathrm{cut}}}/\beta)}\right)+O\left(w_{\max}\varepsilon^{-2}n^2\log^2(n/\beta)\sqrt{\log(\abs{\mathcal{Q}_{\mathrm{cut}}}/\beta)}\right)\\
    &+O\left(\varepsilon^{-3}n^2\log^3(n/\beta)\sqrt{\log(\abs{\mathcal{Q}_{\mathrm{cut}}}/\beta)}\right)\\
    =&O\left(w_{\max}^2\varepsilon^{-1}\log(n/\beta)+w_{\max}\varepsilon^{-2}\log^2(n/\beta)+\varepsilon^{-3}\log^3(n/\beta)\right)\cdot \left(n^{5/2}+n^{2}\log(1/\beta)\right)
\end{align*}
For unweighted graph, the corresponding bound is $O\left(\varepsilon^{-3}\log^3(n/\beta)\cdot\left(n^{5/2}+n^{2}\log(1/\beta)\right)\right)=\tilde{O}(n^{5/2}/\varepsilon^3)$.

Generally, for $K_h$-motif cut queries, we can release a synthetic graph by randomized response with additive error $O(\varepsilon^{-h}\log^{h}(n/\beta)\cdot\sqrt{\binom{n}{2}}\binom{n}{h-2}\sqrt{\log(\abs{\mathcal{Q}_{\mathrm{cut}}}/\beta)})=\tilde{O}(n^{h-1/2}/\varepsilon^h)$. 
\end{document}